\documentclass[11pt]{article}
\usepackage{mathrsfs}
\usepackage{natbib,graphicx,setspace,lscape,longtable}
\usepackage{stmaryrd}
\usepackage{natbib,epsfig,graphicx}
\usepackage{mathrsfs,amsmath,amsthm,amssymb,color}
\usepackage{enumerate}
\usepackage{threeparttable}

\bibpunct{(}{)}{;}{a}{,}{,}

\usepackage[a4paper,top=2.54cm,bottom=2.54cm,left=2cm,right=2cm]{geometry}

\usepackage{epsfig}
\usepackage{caption}
\usepackage{amsmath,amsthm,amssymb}
\usepackage{enumerate}
\usepackage{array}
\usepackage{bm}
\usepackage{rotating}
\usepackage{float}
\usepackage{multirow}
\usepackage{graphicx}
\usepackage{ulem}
\usepackage{natbib}
\usepackage{color}
\usepackage{float}
\usepackage{booktabs}
\usepackage{subfigure}
\usepackage[hidelinks]{hyperref}
\usepackage{epstopdf}
\usepackage{algorithm}
\usepackage{algorithmicx}
\usepackage{algpseudocode}

\newtheorem{lemma}{{\bf Lemma}}
\newtheorem{corollary}{{\bf Corollary}}
\newtheorem{remark}{{\bf Remark}}
\newtheorem{theorem}{{\bf Theorem}}

\newtheorem{proposition}{{\bf Proposition}}
\def\bep{\begin{proposition}}
\def\eep{\end{proposition}}

\newcommand{\bw}{\mbox{\bf w}}

\newcommand{\bI}{\mbox{\bf I}}

\newcommand{\mE}{\mathbb{E}}

\newcommand{\bW}{\mbox{\bf W}}
\newcommand{\bX}{\mbox{\bf X}}
\newcommand{\mR}{\mathbb{R}}

\newcommand{\bdelta}{\mbox{\boldmath $\delta$}}
\newcommand{\bbeta}{\mbox{\boldmath $\beta$}}

\newcommand{\bSig}{\mbox{\boldmath $\Sigma$}}

\newcommand{\T}{\top}

\newcommand{\var}{\mathrm{var}}

\begin{document}

\title{Residual Importance Weighted Transfer Learning For High-dimensional Linear Regression}
\author{Junlong Zhao$^{1}$,~ Shengbin Zheng$^{1}$, ~Chenlei Leng$^{2,*}$\\
$^{1}$School of Statistics, Beijing Normal University\\
$^{2}$Department of Statistics, University of Warwick}
\date{}
\maketitle

\begin{abstract}
 Transfer learning is an emerging paradigm for leveraging multiple sources to improve the statistical inference on a single target. In this paper, we propose a novel approach named residual importance weighted transfer learning (RIW-TL) for high-dimensional linear models built on penalized likelihood. Compared to existing methods such as Trans-Lasso that selects sources in an all-in-all-out manner, RIW-TL includes samples via importance weighting and thus may permit more effective sample use. To determine the weights, remarkably RIW-TL only requires the knowledge of one-dimensional densities dependent on residuals, thus overcoming the curse of dimensionality of having to estimate high-dimensional densities in naive importance weighting. We show that the oracle RIW-TL provides faster rate than its competitors and develop a cross-fitting procedure to estimate this oracle. We discuss variants of RIW-TL by adopting different choices for residual weighting.  The theoretical properties of RIW-TL and its variants are established  and  compared with those of  LASSO and Trans-Lasso. Extensive simulation and a real data analysis confirm its advantages.
\end{abstract}


\section{Introduction}
Statistical techniques are most demanding in applications where sample sizes are small. This is particularly true when they come with a large number of variables. Although there is now a mature sub-field of statistics focusing on penalized regression for producing sparse models that mitigates the challenge of high dimensionality,  small sample sizes pose a fundamental limitation on the statistical properties of any estimator.

Fortunately, though the sample size for a \textit{target} problem of interest is small, in reality there often exist multiple different but related \textit{source} datasets in many applications. For example,
\begin{enumerate}[$(a)$~~]
\item in biology, for predicting gene expression, the  data for the target tissue may be limited, but the source data for other tissues may be large \citep{li2021};
\item in economics, to study what may affect income in a particular country, we have income data from many other countries;
\item in pattern recognition, you want to build a classifier to identify drones from a target dataset. Instead of building your model from scratch, you may want to explore similar classifiers developed to recognize, for instance, birds.
\end{enumerate}
In these scenarios, there is a possibility to explore the similarity between the sources and the target. This paper is about a new methodology on how to do it for linear regression, for which our goal is to  build a model to relate a response variable  $y\in\mR$ to a predictor vector $\bm x\in\mR^p$ when $p$ is large. The starting point is that we have a data set $\mathcal{S}^{(0)}=\{\bm z_i^{(0)}=(\bm x_i^{(0)}, y_i^{(0)}), i=1, \cdots, n_0\}$ with \textit{i.i.d.} observations which will be referred to as the target data hereafter. The data come from the following linear model
\begin{equation}\label{lm}
 y_i^{(0)} = (\bm x_i^{(0)})^{\top} \bbeta^{(0)} + \epsilon_i^{(0)},
 \end{equation}
 where $\epsilon_i^{(0)}$'s are random noises with $\mathbb{E}(\epsilon_i^{(0)})=0$ and our interest is to estimate the unknown regression coefficients $\bbeta^{(0)}$.
 Note
$$
\bbeta^{(0)} = \underset{\bm\beta \in \mathbb{R}^p}{\rm argmin} \, \mathbb{E}_{\bm z \sim f_0}(y - \bm x^{\top} \bbeta)^2
$$
where $f_0$ is the pdf of $\bm z:=(\bm x^{\top}, y)^{\top}$ in the target data and $\bm z\sim f_0$ denotes that $\bm z$  is a random vector following $f_0$.
Assume that $\bbeta^{(0)}$ is sparse such that the cardinality of its support $s_0:=|\mathrm{supp}(\bbeta^{(0)})|$ satisfies $s_0\ll n_0$ (or approximately so). It is known that the LASSO estimator \citep{Tibshirani1996Regression}
 \begin{equation}\label{eq:lasso}
  \hat\bbeta^{(0)}_{\rm Lasso} = \underset{\bm\beta \in \mathbb{R}^p}{\rm argmin} \, \frac{1}{2n_0}\sum_{i=1}^{n_0}\left\{y_i^{(0)} - (\bm x_i^{(0)})^{\top} \bbeta\right\}^2+\lambda \| \bbeta\|_1,
  \end{equation}
 where $\|\bbeta\|_1$ is the $\ell_1$ norm of $\bbeta$, achieves the following convergence rate \citep{Bickel2009}
 \[ \| \hat\bbeta^{(0)}_{\rm Lasso} - \bbeta^{(0)}\|^2 = O_p\left( \frac{s_0\log p}{n_0} \right),
 \]
where $\| \cdot \|$ is the $\ell_2$ norm.  For the LASSO estimator to convergence, a fundamental requirement on the sample size is $n_0 \gg s_0\log p$.  With source data available, an interesting question arises whether we can exploit these auxiliary datasets to boost the performance and if yes how. In this paper, we will refer to the LASSO estimator as the estimator defined in \eqref{eq:lasso} that uses only the sample in the target data $\mathcal{S}^{(0)}$.

\subsection{Transfer learning}\label{sec:tl}
Transfer learning is a learning paradigm originated in machine learning for leveraging knowledge learned from other tasks to boost performance on the target under investigation. In our regression setup, denote the source data as $K$ independent datasets $\mathcal{S}^{(k)}=\{\bm z_i^{(k)}=(\bm x_i^{(k)}, y_i^{(k)}), i=1, \cdots, n_k\}, k=1, \cdots, K$, each with \textit{i.i.d.} observations that satisfy the following linear model
\begin{equation}\label{lin-model0}
 y_i^{(k)} = (\bm x_i^{(k)})^{\top} \bbeta^{(k)} + \epsilon_i^{(k)},
 \end{equation}
 where $\epsilon_i^{(k)}$'s are random noises with $\mathbb{E}(\epsilon_i^{(k)})=0$ not necessarily following the same distribution of $\epsilon_i^{(0)}$ and
 $\bbeta^{(k)}$ is sparse or approximately so.  If $\bbeta^{(k)} \approx \bm \bbeta^{(0)}$ for some $k$ in a suitable sense, then the $k$th source may be exploited to aid the estimation of $\bbeta^{(0)}$. In reality, of course, we seldom know which sources are relevant. Transfer learning basically aims to identify those complementary sources in order to leverage their information to boost the estimation of the main quantity of interest. Intuitively since the resulting estimator will be based on a larger sample size combining the target and the related sources, it may converge at a faster rate.

Motivated by the rationale above, for high-dimensional models, a dominant class of approaches for transfer learning advocated by \cite{li2021} roughly follows the general recipe discussed below. Define $\bdelta^{(k)}=\bbeta^{(k)}-\bbeta^{(0)}$ as the contrast between the regression coefficients of the $k$th source and the target. Denote the set of informative sources as
\[ \mathcal{J} =\{1\le k \le K: \| \bdelta^{(k)} \|_1 \le h \},
\]
which includes those datasets with contrasts sufficiently small as measured by $h$. Given $\mathcal{J}$, the Trans-Lasso approach in \cite{li2021} estimates the coefficient $\bbeta^{(0)}$ in the target data as
\[\hat\bbeta^{(0)}_{\rm Trans-Lasso}= \hat\bw-\hat\bdelta,\]
where $\hat\bw$ is the LASSO estimator using the data in the informative set $\bigcup_{k \in \mathcal{J} }\mathcal{S}^{(k)} $ and
\[
\hat{\bdelta} = \underset{\bm\delta \in \mathbb{R}^p}{\rm argmin} \, \frac{1}{2n_0}\sum_{i=1}^{n_0}\left\{y_i^{(0)} - (\bm x_i^{(0)})^{\top} (\hat\bw-\bdelta)\right\}^2+\lambda_\delta \| \bdelta\|_1.
\]
 The idea underpinning Trans-Lasso is that for the data in source set $\mathcal{J}$, the population version of the regression coefficient of $y$ on $\bm x$ is a linear combination of $\bbeta^{(k)}, k \in \mathcal{J}$, and thus is approximately sparse. In addition,  the difference between $\bbeta^{(0)}$ and this population regression coefficient is also approximately sparse. Provided $\mathcal{J}$ is known, the oracle Trans-Lasso estimator in \cite{li2021}
 can achieve the following convergence rate
 \[ \| \hat\bbeta^{(0)}_{\rm Trans-Lasso} - \bbeta^{(0)}\|^2 = O_p\left( \frac{s_0\log p}{n_0+n_{\mathcal{J}}} +\frac{s_0\log p}{n_0} \wedge h \sqrt{\frac{\log p}{n_0}} \wedge h^2\right),
 \]
 where $a \wedge b$ is the minimum of $a$ and $b$ and $n_{\mathcal{J}}=\sum_{k \in \mathcal{J}} n_k$ is the sample size of the informative set. Immediately we see that we require $h\ll s_0\sqrt{\log p/n_0}$ and $n_0 \ll n_{\mathcal{J}}$ for the oracle Trans-Lasso to have improved rate over the LASSO estimator. Note that in Trans-Lasso, the observations in a single source dataset $\mathcal{S}^{(k)}$ is either all included in or left out. This strategy immediately raises the following fundamental question:
\begin{center}
\it{Does all-in-or-all-out make the best use of source data in transferring knowledge?}
\end{center}

\subsection{Importance weighted transfer learning}\label{sec:iw}
An idea to use all the observations in the source data is via importance weighting by noting \citep{cochran2007sampling,fishman2013monte}
  \begin{equation}\label{eq: importance weighting}
  \bbeta^{(0)} = \underset{\bm\beta \in \mathbb{R}^p}{\rm argmin} \, {\mathbb{E}}_{\bm z \sim f_k}
  \left\{
  \frac{f_0(\bm z )}{f_k(\bm z)}
  (y - \bm x^{\top} \bbeta)^2
  \right\},
  \end{equation}
where $f_k$ is the pdf of $\bm z$ in the $k$th source. Note that by this notation, we allow $f_k$ to be different for different $k$. To estimate $\bbeta^{(0)}$, a simple approach is to formulate a weighted least-squares loss function aggregating all the observations in the target and source datasets, aided by adding an $\ell_1$ penalty on the estimand to encourage sparsity. This simple approach may use all the observations in all the available data as opposed to all-in-or-all-out.

The simple importance weighted approach, however, does not work unless for trivial cases, because specifying the weights $f_0/f_k$ requires the knowledge of  the unknown $f_0$ and $f_k$, the joint distribution of $y$ and $\bm x$ in each data set. One way to proceed is to estimate these two unknowns via density estimation, which will suffer from the curse of dimensionality since $f_0$ and $f_k$ are both $(p+1)$-dimensional. The major thrust of this paper is a novel importance weighted transfer learning approach that only requires the estimation of one-dimensional densities. The foundation of our approach rests on the following important property.

 \bep\label{prop1}
  For $1 \le k \le K$, it holds that
  \begin{equation}
  \bbeta^{(0)} = \underset{\bm\beta \in \mathbb{R}^p}{\rm argmin} \, \mathbb{E}_{\bm z \sim f_k}
  \left\{
  \omega^{(k)} \cdot
  (y - \bm x^{\top} \bbeta)^2
  \right\}, \label{eq:riw-tl}
  \end{equation}
  with  weights $\omega^{(k)}$ defined as
\begin{equation}\label{weight}
\omega^{(k)} =\omega^{(k)}(\bm z)=
\frac{f_{\epsilon}(y-\bm{x}^\top\bbeta^{(0)})}{f_{\epsilon^{(k)}}(y-\bm{x}^\top\bbeta^{(k)})},
\end{equation}
where $f_\epsilon$ is the pdf of a random variable $\epsilon$ independent of the data in source $k$, satisfying $\mathbb{E}(\epsilon)=0$ for reasons stated  in Section \ref{sec:f}, and $f_{\epsilon^{(k)}}$ is the pdf of $\epsilon^{(k)}$ in the $k$th source data.
 \eep
Note that in this proposition, we can make $\epsilon$ depend on $k$ but we have omitted $k$ for simplicity.  A few remarks are in place.
 \begin{enumerate}[$(a)$~~]
\item  For the proposition to hold, we do not require the marginal distribution of $\bm x^{(k)}$,  the population version of the predictor in source $k$, to be the same across $k$. That is, our framework adapts to what is known as \textit{covariate shift} in the literature.

 \item In the denominator of the weights, if $\epsilon^{(k)}$ follows Gaussian,  estimating $f_{\epsilon^{(k)}}$ is rather trivial because we just need to estimate the variance of $\epsilon^{(k)}$ or the residuals $y_i^{(k)}-(\bm{x}_i^{(k)})^\top\bbeta^{(k)}$ in practice.  Alternatively, without any distributional assumption on $f_{\epsilon^{(k)}}$, we can estimate the density of $\epsilon^{(k)}$ nonparametrically based on these residuals via univariate density estimation. The latter case is more robust to the misspecification of the distribution of $\epsilon^{(k)}$ and is somewhat more interesting from a methodological and theoretical perspective.  Thus, the rest of the paper focuses on this case, while we examine the performance when $f_{\epsilon^{(k)}}$ is assumed Gaussian numerically.

\item In the numerator of the weights, $\epsilon$ can be any random variable independent of the data as long as it has a mean zero. This opens the door to a wide variety of options. We will discuss two instances for specifying the distribution of $\epsilon$: the first by symmetrizing $f_{\epsilon^{(k)}}$ and the second by assuming a uniform distribution.

 \item For the unknown
 $\bbeta^{(0)}$ and $\bbeta^{(k)}$ needed in specifying the importance weights, we can replace them with their penalized regression estimators or a Trans-Lasso estimator.
 \end{enumerate}
 With $\omega^{(k)}(\bm z^{(k)}_i)$ estimated using the estimated quantities discussed, we then formulate a penalized loss function combining all the observations in the target and source data. Because our approach utilizes residuals at its core, we will refer to it as Residual Importance Weighted Transfer Learning or RIW-TL for short.

\vspace{2mm}
\noindent{\textbf{Difference between LASSO, oracle Trans-Lasso and oracle RIW-TL.}} Before we discuss in more detail our methodology, including, for example, how to deal with weights that may be unbounded, we first highlight the difference between the oracle Trans-Lasso in  \cite{li2021} and the oracle RIW-TL in a simple setting where there is a single source $K=1$. As a reminder, the oracle Trans-Lasso is achieved when its tuning parameters are chosen optimally assuming a known $\mathcal{J}$, while for the oracle RIW-TL, $\omega^{(k)}(\bm z^{(k)}_i)$ are assumed known and properly trimmed for boundedness with optimally chosen tuning parameters in various places. We provide the convergence rate of an oracle by calculating the squared $\ell_2$ norm of the difference between an oracle estimator and $\bbeta^{(0)}$.

To appreciate these oracles under different regimes, we divide the range of the contrast $\bdelta^{(1)}=\bbeta^{(1)}-\bbeta^{(0)}$ into three scenarios with the rates provided in Table \ref{tab:0}. In the first regime when $\bbeta^{(0)}$ and $\bbeta^{(1)}$ are close, both approaches produce oracle estimators as if we had a sample size $n_0+n_1$. That is, we effectively use all the samples in the target and source. In regime (ii) when $\bdelta^{(1)}$ becomes larger but its $\ell_1$ norm is smaller than $s_0 \sqrt{\log p/n_0}$, oracle RIW-TL retains the same rate while oracle Trans-Lasso has a slower rate. How much slower depends on the interplay between the ambient dimension $p$, true dimension $s_0$ and the two sample sizes $n_0$ and $n_1$.  In the difficult regime (iii) where $\bdelta^{(1)}$ is much larger than $s_0 \sqrt{\log p/n_0}$ in $\ell_1$ norm, the oracle Trans-Lasso may produce an estimator inferior to the LASSO estimator if the source data chosen is uninformative, while in the best case scenario it has the same
convergence rate as the LASSO estimator using only the target data. In the inferior case, using uninformative source data causes what is referred to as negative transfer for the target problem \citep{torrey2010}. In all the settings, our oracle RIW-TL has a faster rate than the LASSO estimator using only the target data.

 \begin{table}[!hbt]
\caption{\label{tab:0}Comparison of the convergence rates of two oracles under three regimes with a single source having sample size $n_1$. }
\begin{tabular*}{\columnwidth}{@{\extracolsep\fill}lcc@{\extracolsep\fill}}
\toprule
{\bf Regime}  & {\bf Oracle RIW-TL}  & {\bf Oracle Trans-Lasso}\\
\midrule
($i$)~$\|\bdelta^{(1)}\|_1 \lesssim \dfrac{s_0 \sqrt{n_0 \log p}}{n_0 + n_1} $    & $ \dfrac{s_0 \log p}{n_0+n_1}$ & $\dfrac{s_0 \log p}{n_0+n_1}$ \\[1ex]
($ii$)~$\dfrac{s_0 \sqrt{n_0 \log p}}{n_0 + n_1} \ll \|\bdelta^{(1)}\|_1 \ll s_0 \sqrt{\dfrac{ \log p}{ n_0}} $ &  $\dfrac{s_0 \log p}{n_0+n_1}$ &$ \sqrt{\dfrac{\log p}{n_0}} \left\|\bdelta^{(1)}\right\|_1 \wedge \left\|\bdelta^{(1)}\right\|_1^2\wedge \dfrac{s_0\log p}{n_0}
$ \\[1ex]
($iii$)~$\|\bdelta^{(1)}\|_1 \gtrsim s_0 \sqrt{\dfrac{ \log p}{ n_0}}$&  $\dfrac{s_0 \log p}{n_0 + n_1/\max\{{\|\bdelta^{(1)}\|},1\}}$  & $\sqrt{\dfrac{\log p}{n_0}} \left\|\bdelta^{(1)}\right\|_1 \wedge \dfrac{s_0\log p}{n_0}$ \\[1ex]
\bottomrule
\end{tabular*}
\end{table}

\vspace{-2mm}
\subsection{Contributions}
Our methodological contribution is to propose residual importance weighted transfer learning (RIW-TL) as a new framework for transferring information from source data to target data. This marks a paradigm shift from selecting in the existing all-in-or-all-out framework \citep[cf.]{li2021,tian2022} to weighting. By weighting each individual observation in the source data via a density ratio, RIW-TL potentially makes better and more effective use of the observations in the source, as demonstrated theoretically and empirically in this paper. We show that the rate of oracle RIW-TL is superior in the difficult regime when a diverse range of source data is available, some of which may have large contrasts between $\bbeta^{(k)}$ and $\bbeta^{(0)}$. Moreover, RIW-TL is oblivious to the homogeneity of the design, not requiring the marginal distributions of the predictors in the source to be similar, unlike its all-in-or-all-out competitors. These facts make RIW-TL appealing in reality, especially because we seldom know whether knowledge learned in other problems or which other problems' knowledge can be transferred.

When the weights in \eqref{weight} needed for the oracle are unknown, we propose a cross-fitting approach to implement the oracle RIW-TL. We investigate the nonparametric estimation of the weights in \eqref{eq:riw-tl} via kernel density estimation \citep{Silverman1978} or parametric estimation in numerical study by assuming Gaussianity of the errors. When $n_0$, the sample size of the target data, is small, it is found that the rate of RIW-TL via cross-fitting is inferior to that of the oracle. To bring the rate closer to the oracle, we further explore the specification of $\epsilon$ as a uniform random variable with unknown endpoints and show that the rates will be of the same order when the sources are close to the target on average.

In specifying and estimating the weights in RIW-TL, without extra care, however, the weights $f_\epsilon/f_{\epsilon^{(k)}}$ may become unbounded. As a result, a sample version of \eqref{eq:riw-tl} will be dominated by those observations with large weights, causing a large variance in the resulting estimator. This is reminiscent of a wider issue for importance weighting-based methods, for example, when used in Monte Carlo computing \citep{tokdar2010importance}. To overcome the challenge, we employ further sample selection so that we can focus on those observations whose weights defined in \eqref{weight} are bounded in a suitable way. This on the other hand brings challenges in parameter estimation which we must handle with delicate analysis.

\vspace{-2mm}
\subsection{Literature review}
Transfer learning is a modern technique for applying the insight gained in related problems to a new context. In many relatively data-scarce applications including pattern drug sensitivity prediction, biological imaging diagnosis, and natural language processing, for example, it has found many uses \citep[cf.]{torrey2010,shin2016, Wang2018deep}. On the other hand, high-dimensional data analysis has undergone rapid development thanks to a series of important breakthroughs made in statistics with the development of LASSO \citep{Tibshirani1996Regression} and SCAD \citep{fan2001variable}, among many others.

Exploring the interface between these two areas, \cite{bastani2021} studied transfer learning with a single source for high-dimensional generalized linear models (GLMs). \cite{li2021} proposed the Trans-Lasso method in high-dimensional linear regression with multiple source data, showing that their estimators can be minimax optimal under certain conditions. \cite{tian2022} proposed a method for transfer learning in high-dimensional GLMs. \cite{Li2023} developed an approach for high-dimensional GLMs with knowledge transfer that estimates the target parameter and coefficients difference jointly, while \cite{Zhang2022} considered transfer learning for high-dimensional quantile regression. The literature on transfer learning for high-dimensional data is rapidly expanding recently. We refer the reader to its use in federated learning, functional data analysis and many others   \cite[etc.]{Li2021federal, Zhou2022, Lin2022, Gu2022}. The majority of the methods taken by these papers select source data to enhance estimation for the target data at the source level, that is,  data in a single source are either selected altogether or discarded as a whole. To identify those useful sources,   most methods aggregate candidate sets of informative sources, sometimes by applying model averaging to avoid negative transfer \citep[cf.]{li2021} or identifying informative sources by minimizing prediction error in the target \citep[cf.]{tian2022}.

An alternative to the all-in-or-all-out approach is importance weighting \citep[cf.]{Sugiyama2007,fang2020}. As can be seen in Section \ref{sec:iw}, a key step to use importance weighting is to estimate the importance weight ${f_0(\bm z )}/{f_k(\bm z)}$. For small or moderate dimensional problems, there are well-established works mostly for estimating the weights based on nonparametric regression. For example, \cite{Scholkopf2007} proposed kernel mean matching method  and \cite{Kanamori2009} modeled the importance weights using a set of function basis. \cite{Shimodaira2000} proposed the empirical risk minimization method with an exponentially-flattened parameter to reduce the effects of high variance of the estimated importance weights, while
\cite{Yamada2011} further explored the fluctuation problem and proposed a relative importance weighting method. A recent review on this topic can be found in  \cite{lu2022rethinking}. These methods nevertheless are not applicable in the high-dimensional data setting in this paper.

\vspace{-2mm}
 \subsection{Organization of the paper and notations}
 The rest of this paper is organized as follows. In Section 2, we study the oracle RIW-TL  for a high-dimensional linear model. In Section 3, we propose an approach to estimate this oracle via kernel density estimation and cross-fitting, and study its properties.  In Section 4, we propose another variant of this approach by adopting different weighting and sample selection with properties provided. Extensive simulation and real data analysis are reported in Sections 5 and 6. A brief discussion is presented in Section 7. All the technical details and additional numerical results can be found in the supplementary materials. The code implementing RIW-TL is freely available on \url{https://github.com/RIW-TL/Transfer-learning}.

{\bf Notations}. For any two positive  series $\{a_n\}$ and $\{b_n\}$, we use $a_n\gtrsim b_n$ or {$b_n \lesssim a_n$} to mean $\lim_{n\to \infty} a_n/b_n>C>0$ for some constant $C$,  and $a_n\asymp b_n$ to mean they have the same rate. We use either $a_n\ll b_n$ or $a_n=o(b_n)$ when $a_n/b_n\to 0$, and  $a_n=O(b_n)$ when $\lim\limits_n a_n/b_n<C<\infty$ for some constant $C$. And $o_p(\cdot)$ and $O_p(\cdot)$ are defined similarly for the stochastic version.
 For any vector $\bm v=(v_1, \cdots, v_p)^\top\in\mR^p$, $\|\bm v\|_1$, $\|\bm v\|$ and $\|\bm v\|_\infty$ denote its $\ell_1,\ell_2$ and $\ell_\infty$ norm respectively.  For a set $\mathcal{A}$, we use $|\mathcal{A}|$ to denote its cardinality. We define $[p]:=\{1, \cdots, p\}$ for any integer $p$. For any set $\mathcal{A} \subseteq [p]$, we use $\bm v_{\mathcal{A}}$ to denote a subvector of $\bm v$ with entries in $\mathcal{A}$. For any nonempty set $\mathcal{A}$, we denote $\bI_{\mathcal{A}}$ as the diagonal matrix with its $i$th diagonal element being indicator ${\rm I
}(i \in \mathcal{A})$.
 Finally, we use $C, C_1, \cdots, c_1,\cdots$ to denote absolute constants which may vary from line to line depending on the context.

 As a reminder, we have used $\mathcal{S}^{(0)}$ with sample size $n_0$ and $\mathcal{S}^{(k)}$ with sample size $n_k$ to denote the target and the $k$th source data respectively.
 The notations $f_k$ and $f_{\epsilon^{(k)}}$ are used respectively to denote the joint distribution of $(y, \bm x^\top)^\top$ in $\mathcal{S}^{(k)}$  (or $(y^{(k)}, (\bm x^{(k)})^\top)^\top$), and the distribution of $\epsilon^{(k)}$. {Without loss of generality, we assume that $\bm x^{(k)}_i$ are all centered such that $\mathbb{E}(\bm x_i^{(k)})=0$ with the covariance matrix denoted as $\bSig^{(k)} = (\sigma_{ij}^{(k)})$,} which is positive definite.

 \section{Residual Importance Weighting}\label{sec:riw}
 We first discuss the choice of $\omega^{(k)}$ in \eqref{weight}. Write
 $$
y_i^{(k)}=(\bm{x}_{i}^{(k)})^\top\bbeta^{(k)}+\epsilon^{(k)}_i=(\bm{x}_{i}^{(k)})^\top\bbeta^{(0)}+
\{(\bm{x}_{i}^{(k)})^\top(\bbeta^{(k)}-\bbeta^{(0)})+\epsilon_i^{(k)}\},
 $$
 where we will denote
 \begin{equation}\label{eta}
\eta_i^{(k)}:=(\bm{x}_i^{(k)})^{\top}(\bbeta^{(k)} - \bbeta^{(0)})=(\bm{x}_i^{(k)})^{\top}\bm \delta^{(k)}.
\end{equation}
 From Proposition \ref{prop1}, the importance weight of the observation $\bm z_i^{(k)}$  is simply
  \begin{equation}\label{weight2}
  \omega_i^{(k)} = \frac{f_{\epsilon}\bigl(y_i^{(k)}-(\bm{x}_i^{(k)})^\top\bbeta^{(0)}\bigr)}{f_{\epsilon^{(k)}}\bigl(y_i^{(k)}-(\bm{x}_i^{(k)})^\top\bbeta^{(k)}\bigr)} =
   \frac{f_{\epsilon}(\epsilon_i^{(k)} + \eta_i^{(k)})}{f_{\epsilon^{(k)}}(\epsilon_i^{(k)})}.
  \end{equation}
The quantity $\eta_i^{(k)}$ measures  the contrast between $\bbeta^{(k)}$ and $\bbeta^{(0)}$ adjusted by an individual predictor. Since the denominator in \eqref{eta} can be close to zero, one difficulty associated with $\omega_i^{(k)}$ is that it may become unbounded. To avoid this and for technical reasons, we employ sample selection by choosing those observations that have bounded weights in the following set
\begin{equation}\label{weight-constraint}
\mathcal{I}_k = \{i \in [n_k]: |\epsilon_i^{(k)} + \eta_i^{(k)}| \le A,~ |\eta_i^{(k)}| \leq M \}, \qquad k = 1,\cdots,K,
\end{equation}
where $A$ and $M$ are two tuning parameters. This particular choice of sample selection is employed to ensure that both the numerator and the denominator in (\ref{weight2}) are bounded under some conditions.
We discuss the selection of $A$ and $M$ via cross-validation later.
We remark that $\mathcal{I}_k$ is not the only sample selection that one can use. An alternative is, for example, $\{i \in [n_k]: |\epsilon_i^{(k)}| \le A, |\eta_i^{(k)}| \leq M\}$ that also constraints the weights. We find that the choice of  $\mathcal{I}_k$ matters, leading to different theoretical properties. {More discussion can be found in Section \ref{sec:Properties of RIW-TL}, for example, in Remark \ref{remark2}}, and this particular alternative subset will be discussed in Remark \ref{remark1} in Section \ref{sec:f}.

To study in detail the magnitude of $\omega_i^{(k)}$, we
 make the following assumptions on $f_{\epsilon^{(k)}}$.

\begin{enumerate}[]
\item
 \textbf{Condition 1.} For $k=1,\cdots, K$, $f_{\epsilon^{(k)}}$ is bounded away from $\infty$ satisfying $f_{\epsilon^{(k)}}(t)>0$ for any finite $t$, and has bounded first derivative.
\end{enumerate}
This assumption allows $\inf_{t\in\mR} f_{\epsilon^{(k)}}(t)=0$ and   is quite mild. It holds for many commonly encountered distributions such as Gaussian and the $t$ distribution. Under this assumption,
the following proposition shows that $f_{\epsilon^{(k)}}(\epsilon_i^{(k)})$ and $\omega_i^{(k)}$ are both bounded for all $i \in \mathcal{I}_k$.

\bep\label{prop2}
Under Condition 1, for $k = 1,\cdots,K$ and all $i \in \mathcal{I}_k$, it holds that
  ($i$)  $f_{\epsilon^{(k)}}(\epsilon_i^{(k)})$ is bounded away from 0 strictly (i.e. larger than a positive constant), and   ($ii$)
    $\omega_i^{(k)}$ is bounded away from 0  and $\infty$ {strictly}.
\eep

With the weights all bounded, we are now ready to define the oracle RIW-TL assuming known weights.
Denote  $\mathcal{I}_0 = [n_0]$ and let $\omega_i^{(0)} = 1$ for all the observations in the target data. The oracle RIW-TL estimator is found as
 \begin{equation}\label{ora}
\tilde \bbeta_{ora}^{(0)} = \underset{\bm\beta \in \mathbb{R}^p}{\rm argmin}~ \frac{1}{2(n_0+\sum_{k=1}^K n_k)}
\left\{
\sum_{k = 0}^K \sum_{i\in \mathcal{I}_k}
 \omega_i^{(k)} \left(y_i^{(k)} - (\bm x_i^{(k)})^\top \bbeta\right)^2
\right\} + \lambda \|\bbeta\|_1,
\end{equation}
where  $\lambda$ is a tuning parameter. That is, we define the oracle RIW-TL as the solution to a penalized weighted least-squares problem using the observations in the target data, as well as all those observations having bounded weights in the $K$ sources.

\subsection{The choice of $f_\epsilon$ in \eqref{weight2}}\label{sec:f}
We remark that $\bbeta^{(0)}$ is not the population minimizer of the expectation of the loss function in \eqref{ora} when $\lambda=0$. That is, our formulation incurs some bias even in the best case scenario. Although it will be interesting to characterize this bias explicitly, it turns out that it is largely irrelevant for the purposes of establishing a rate for the oracle RIW-TL estimator due to our proof strategy. In particular, to establish the rate of $\tilde \bbeta_{ora}^{(0)}$, we will make use of the basic inequality that the loss function in \eqref{ora} evaluated at $\tilde \bbeta_{ora}^{(0)}$ is no greater than it evaluated at $\bbeta^{(0)}$, since the former minimizes the loss function. This is the proof strategy widely used for proving the rate of LASSO type estimators \citep[cf.]{Bickel2009}. To establish this rate for the LASSO estimator defined in \eqref{eq:lasso} for example, a critical step is to provide an upper bound of the largest element for $\sum_{i=1}^{n_0} \bm x_i^{(0)} \epsilon_i^{(0)}$ with high probability.

We now discuss informally our strategy to make a similar term stochastically small.
Recall that for a random vector $(y_i^{(k)},(\bm{x}_{i}^{(k)})^{\top})^{\top}$ from source $k$, we have denoted
$$
y_i^{(k)}=(\bm{x}_{i}^{(k)})^\top\bbeta^{(0)}+(\eta_i^{(k)}+\epsilon_i^{(k)}),
 $$
 in which we can view $\epsilon_i^{(k)} + \eta_i^{(k)}$ as the random error.
In the context of the oracle RIW-TL estimator defined in \eqref{ora}, to use the idea of the basic inequality, we need to bound the  term $\sum_{i=1}^{n_0} \bm x_i^{(0)} \epsilon_i^{(0)}$ as in the LASSO estimator case, and additional $K$ terms with the $k$th term having the form
\begin{equation}\label{eq:sum}
\sum_{i=1}^{n_k}  \xi_i^{(k)}\bm x_i^{(k)},~~ \xi_i^{(k)}:=(\eta_i^{(k)}+\epsilon_i^{(k)}) \omega^{(k)}_i {\rm I}(i \in \mathcal{I}_k),
\end{equation}
which arises due to the formulation in \eqref{ora}.
Note that $\xi_i^{(k)}$ is a bounded random variable for given $\bm x_i^{(k)}$, since  $\omega^{(k)}_i $ is bounded by Proposition \ref{prop2}.
The expectation of $\xi_i^{(k)}$ is seen as
\[ \text{Bias}(f_\epsilon, \mathcal{I}_k):=
\mathbb{E}_{\epsilon_i^{(k)}} \xi_i^{(k)}
 = \int_{i \in \mathcal{I}_k} (\epsilon_i^{(k)} + \eta_i^{(k)}) f_{\epsilon}(\epsilon_i^{(k)} + \eta_i^{(k)}) d \epsilon_i^{(k)}= \int_{-A}^A t f_{\epsilon}(t) dt.
 \]
 If we take $f_\epsilon$ to be a symmetric function,
 $\xi_i^{(k)}$ will have mean zero for fixed $\bm x_i^{(k)}$. Thus, \eqref{eq:sum} is just a sum of bounded mean zero random variables that can be further bounded stochastically. This is why we require the random error $\epsilon$ satisfies $\mathbb E(\epsilon)=0$. Since the expectation in the above expression
 plays an important role in the choice of $f_\epsilon$ and sample selection $\mathcal{I}_k$, we have used notation $\text{Bias}(f_\epsilon, \mathcal{I}_k)$ for its value.

 Although $f_\epsilon$ can be any symmetric pdf, in this paper, we discuss two choices where $f_\epsilon(t)=[f_{\epsilon^{(k)}}(t)+f_{\epsilon^{(k)}}(-t)]/2$ or when $f_\epsilon$ is the pdf of a uniform distribution. These choices are somewhat more interesting than when $f_\epsilon$ is from a Gaussian distribution, while we only present numerical results for the latter.

\begin{remark}\label{remark1}
 One can  consider another  set  $\mathcal{I}'_k=\{i \in [n_k]: |\epsilon_i^{(k)}| \le A, |\eta_i^{(k)}| \leq M\}$ for which  the bias becomes
\[\text{Bias}(f_\epsilon, \mathcal{I}'_k)=
\int_{-A}^A (\epsilon_i^{(k)} + \eta_i^{(k)}) f_{\epsilon}(\epsilon_i^{(k)} + \eta_i^{(k)}) d \epsilon_i^{(k)} = \int_{-A + \eta_i^{(k)}}^{A + \eta_i^{(k)}} t f_{\epsilon}(t) dt = \int_{A - \eta_i^{(k)}}^{A + \eta_i^{(k)}} t f_{\epsilon}(t) dt, \]
 which is not zero usually.
When $\eta_i^{(k)}$ is small, the bias approximately equals  $2 A f_{\epsilon}(A) \eta_i^{(k)}$. However,
when  $\bm\delta^{(k)}$'s are small on average, taking $f_{\epsilon}$ as the density of a uniform distribution and $\mathcal{I}'_k$ for sample selection has certain advantages as we show in Section \ref{sec:riw-tl-u}.
\end{remark}

\subsection{The effective sample size}
The formulation of the oracle RIW-TL implies that its effective sample size is $n_0+\sum_{k=1}^{K} |\mathcal{I}_k|$. That is, the source data have contributed additional $\sum_{k=1}^{K} |\mathcal{I}_k|$ observations  to the estimation of $\bbeta^{(0)}$ as compared to the LASSO estimator. Noting $\mathcal{I}_k$ is a random set, we quantify the gain in the sample size by evaluating the expectation of $|\mathcal{I}_k|$.

To simplify exposition, we denote $n_{\mathcal{I}_k} = |\mathcal{I}_k|$ as the cardinality of  $\mathcal{I}_k$ and write $n_{\mathcal{I}} = \sum_{k = 1}^K n_{\mathcal{I}_k}$ as the total sample size used for transfer from the source data. Recall the notation $\bm \delta^{(k)} = \bbeta^{(k)} - \bbeta^{(0)}$. Denote $\eta^{(k)} = (\bm x^{(k)})^{\top}(\bbeta^{(k)} - \bbeta^{(0)})=(\bm x^{(k)})^{\top}\bm\delta^{(k)}$ to be consistent with the notion $\eta^{(k)}_i$ in \eqref{eta},
  and \[d^2_k=(\bdelta^{(k)})^{\top} \bSig^{(k)} \bdelta^{(k)}\] as the quadratic contrast weighted by the covariance of $\bm x^{(k)}$.  We have the following results on the effective sample size $n_{\mathcal{I}_k}$ from the $k$th source.

\bep\label{prop3}
If $\epsilon_i^{(k)}$ is distributed as $N(0,1)$ and $\bm x_i^{(k)}$ follows $N(\bm 0, \bSig^{(k)})$, for $k=1,\cdots, K$, it holds that
$$
 \frac{n_k  d_k}{d_k^2 + 1} \left\{ 1 - {\rm exp}\left( - \frac{d_k^2 + 1}{d_k^2} \varphi^2 \right)
\right\} \lesssim \mE(n_{\mathcal{I}_k}) \leq n_k,
$$
where $\varphi = {\rm min}\{A,M\}<\infty$.
\eep

Proposition \ref{prop3} explicitly states the dependence of the effective sample size $\mE(n_{\mathcal{I}_k})$ on $d_k$. When $d_k = O(1)$, the expected contribution in terms of the sample size from source $k$   is of the order $n_k$ since in this case   $\mE(n_{\mathcal{I}_k}) \asymp n_k$. When $d_k$ diverges as $ n_k\to\infty$, the expected contribution from source $k$  to the effective sample size has an order no less than $n_k/d_k$ since  $n_k/d_k  \lesssim   \mE(n_{\mathcal{I}_k}) \leq n_k$. The expected contribution in these two cases can be summarized as $\mE(n_{\mathcal{I}_k}) \gtrsim n_k/\max\{d_k,1\}$. When  the eigenvalues of $\bSig^{(k)}$ are upper bounded, we can equivalently write $\mE(n_{\mathcal{I}_k}) \gtrsim n_k/\max\{\|\bm\delta^{(k)}\|,1\}$.

Using the results in this proposition, we have the following corollary regarding the sample usage rate (SUR) defined as
\[\rho_{\mathcal{I}} := \frac{n_0 + \mE(n_{\mathcal{I}})}{n_0 + \sum_{k=1}^K n_k},\]
the ratio between the effective sample size used by RIW-TL and the total sample size of all the data. The larger the sample usage rate is, the more observations in the source data are used in knowledge transfer.
\begin{corollary}\label{coro1}
Under the conditions of  Proposition \ref{prop3},   the sample usage rate $\rho_{\mathcal{I}}$ is bounded as
 $$ \frac{n_0+\sum_{k=1}^K n_k/d_k}{n_0+\sum_{k=1}^K n_k}\lesssim  \rho_{\mathcal{I}}\lesssim  1.$$
\end{corollary}
Immediately we see that when $d_k$'s are of the order $O(1)$, the effective sample size of RIW-TL is of the same order as the total sample size in that $\rho_{\mathcal{I}} \asymp 1$.

Corollary \ref{coro1} only provides some idea about the order of the effective sample size. In practice, it is often useful to know exactly how probable an observation in the source data gets transferred for the estimation of $\bbeta^{(0)}$. Clearly, characterizing this probability for a general setup is impossible and hence we focus on simple simulation to illustrate the main point. Towards this, recall from Section \ref{sec:tl} that Trans-Lasso retains observation in source $k$ if $\|\bdelta^{(k)}\|_1 \le h$, while our RIW-TL retains those observations in this source if they belong to $\mathcal{I}_k$ requiring in essence $|\eta_i^{(k)}|=|(\bm x_i^{(k)})^\top \bdelta^{(k)}|$ to be small. Note $|\eta_i^{(k)}|\le \|\bm\delta^{(k)}\|_1\|\bm{x}_i^{(k)}\|_\infty$. Thus below we examine the dependence of the probability of an observation in the source data being used for the target problem on $\|\bm\delta^{(k)}\|_1$ and $\|\bm{x}_i^{(k)}\|_\infty$.
  \begin{figure}[htb]
  \centering
  \includegraphics[scale=0.5]{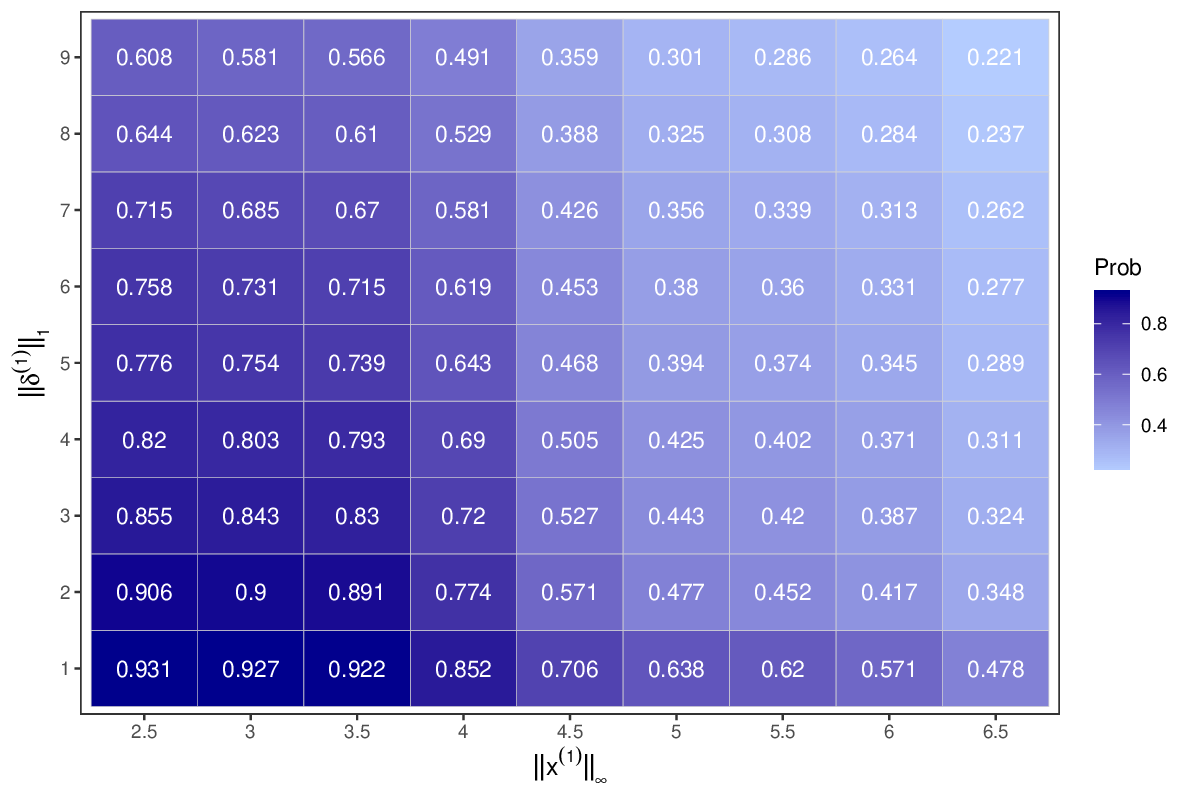}
  \caption{The probability an observation in the source data is included in RIW-TL.} \label{fig1}
  \end{figure}

  In the simulation, we take the number of source data as  $K=1$. We generate $\bm z_i^{(0)}$ and $\bm z_i^{(1)}$ respectively from the linear  models with  $\epsilon_i^{(0)}$ and $\epsilon_i^{(1)}$ following $i.i.d.$ standard normal distribution. The predictors are generated such that $\bm x_i^{(0)}\sim N(0,\bSig)$ and $\bm x_i^{(1)} \sim (1/3)N(-4\mathbf{1}_p,\bSig)+(1/3)N(\mathbf{0}_p,\bSig)+(1/3)N(2\mathbf{1}_p,\bSig)$ where $\bSig=(\sigma_{ij})$ with $\sigma_{ij}=0.5^{|i-j|}$. Here $\mathbf{1}_p$ is a $p$-dimensional vector of ones and $\mathbf{0}_p$ a vector of zeros.
 We set  $\bbeta^{(0)} = (\bm 1_{5}^{\top},\bm 0_{95}^{\top})^{\top}$ and $\bbeta^{(1)}= (\bm 1_{5}^{\top},  {0.2}\bm 1_l^{\top}, \bm 0_{95 - l}^{\top})^{\top}$ with $l\in  \{[5,45],5\}$ such that   $\|\bm{\delta}^{(1)}\|_1=0.2l\in\{[1,9],1\}$, where $\{[a,b],c\}$ denotes the  grid points from $a$ to $b$ with step length $c$.
  We split the range of $\|\bm x^{(1)}\|_\infty$ and $\|\bm{\delta}^{(1)}\|_1$ into $9$ disjoint  intervals as can be seen in Figure \ref{fig1}. For the constraint set in \eqref{weight-constraint}, we set $A=3/2$ and $M=3$ such that  $\mathcal{I}_1 =\{i \in [n_1]: |\epsilon_i^{(1)} + \eta_i^{(1)}| \leq 3/2,~ |\eta_i^{(1)}| \le 3\}$. We generate 100 replicates to compute the frequencies of an observation belonging to $\mathcal{I}_1$ as a function of {$\|\bm\delta^{(1)}\|_1$ and $\|\bm{x}_i^{(1)}\|_\infty$}, discretized on a two-dimensional grid as seen in Figure \ref{fig1}. These frequencies are estimates of the probabilities of an observation getting transferred.    We can see that , given $\|\bm x_i^{(1)}\|_\infty$ (or $\|\bm\delta^{(1)}\|_1$),   the probability is a decreasing function of $\|\bm\delta^{(1)}\|_1$ (or $\|\bm x_i^{(1)}\|_\infty$). For all the grids, this probability is never smaller than 0.2. The probability matrix in Figure \ref{fig1} should be compared to that of Trans-Lasso where the probabilities are always one for the observations in source $k$ if $\|\bdelta^{(k)}\|$ is small, and zero otherwise.

  \subsection{Properties of the oracle RIW-TL}
For $k = 1,\cdots,K$, let $\bX^{(0)}=(\bm x_1^{(0)},\cdots, \bm x_{n_0}^{(0)})^\top \in \mathbb{R}^{n_0 \times p}$ and $\bX^{(k)}=(\bm x_1^{(k)},\cdots, \bm x_{n_k}^{(k)})^\top \in \mathbb{R}^{n_k \times p}$ be the design matrices of the target data and source $k$ data, respectively.
We study the properties of the oracle estimator defined in \eqref{ora} under the following conditions.

\begin{enumerate}[]
 \item
 \textbf{Condition 2.} For $0\le k\le  K$, $\epsilon_i^{(k)}$'s are $i.i.d.$ sub-Gaussian variables with mean zero; that is, there exists  some constant $\kappa > 0$ such that
$\underset{0\le k\le K}{\rm max} \mathbb{E}\{{\rm exp}(u \epsilon_i^{(k)})\} \leq {\rm exp} (u^2 \kappa^2/2)$
 for all $u \in \mathbb{R}$. We denote $\var(\epsilon_i^{(k)})=(\sigma^{(k)})^2<\infty$ { for all $0 \leq k \leq K$.}

 \item \textbf{Condition 3.} Let $\mathcal{H}_0={\rm supp}(\bbeta^{(0)})$ be the support set of $\bbeta_0$. There exist some positive constants $\phi_1$ and $\phi_2$ such that
$$
\underset{\bm v \in \mathcal{E}(\mathcal{H}_0,3)}{\rm inf}
\frac{\bm v^{\top} (\bX^{(0)})^{\top} \bX^{(0)} \bm v}{n_0 \|\bm v\|^2} \geq \phi_1^2 \quad \text{and} \quad
\underset{\bm v \in \mathcal{E}(\mathcal{H}_0,3)}{\rm inf}
\frac{\bm v^{\top}
\left[\sum\limits_{k = 1}^K (\bX^{(k)})^{\top} \bI_{\mathcal{I}_k} \bX^{(k)}\right] \bm v} {n_{\mathcal{I}}\|\bm v\|^2} \geq \phi_2^2,
$$
where $\mathcal{E}(\mathcal{H}_0,3) = \{\bm v \in \mathbb{R}^p :\|\bm v_{\mathcal{H}_0^c}\|_1 \leq 3\|\bm v_{\mathcal{H}_0}\|_1\}$.
\end{enumerate}

The sub-Gaussian condition on $\epsilon^{(k)}$ in Condition 2 is common in high-dimensional data analysis.   Condition 3 is a variant of the restricted eigenvalue (RE) condition commonly assumed in the literature \citep{Bickel2009,sun2012}. When $n_{\mathcal{I}}$ is large, Condition 3 is more likely to hold.
Since the effective sample size $n_{\mathcal{I}_k}$ may be small for some sources,  instead of imposing RE condition separately on $ (\bX^{(k)})^{\top} \bI_{\mathcal{I}_k} \bX^{(k)}$ for each $k$, we assume that RE condition holds on the aggregated design matrix $\sum_{k = 1}^K (\bX^{(k)})^{\top} \bI_{\mathcal{I}_k} \bX^{(k)}$. We have the following rate for the oracle RIW-TL estimator.

\begin{theorem}[Convergence rate of oracle RIW-TL]\label{the1}
 Assume that Conditions 1-3 are satisfied and that  $\mathcal{I}_k$ and $\omega_i^{(k)}$ are all known.
 Let $\lambda \asymp \rho_{\mathcal{I}} \sqrt{{\rm log}\,p / (n_0 + \mathbb{E}(n_{\mathcal{I}}))}$ where $\rho_{\mathcal{I}}$ is the sample usage rate defined before Corollary \ref{coro1}.
 For the oracle estimator $\tilde \bbeta_{ora}^{(0)}$ defined in \eqref{ora}, if $\min\limits_{0\le k\le K} n_k\to \infty$, it follows that \begin{equation}\label{rate1}
\|\tilde \bbeta_{ora}^{(0)} - \bbeta^{(0)}\|^2
 = O_p \left(
 \frac{s_0 {\rm log}\,p}{n_0 + \mathbb{E}(n_{\mathcal{I}})}
\right).
\end{equation}
\end{theorem}
Theorem \ref{the1} shows that the convergence rate of RIW-TL in squared $\ell_2$ norm is inversely proportional to $n_0 + \mathbb{E}(n_{\mathcal{I}})$, rather than $n_0$ in the LASSO case. When $\mE(n_{\mathcal{I}}) \gg n_0$, the convergence rate is much faster than LASSO.  As we have argued in Proposition \ref{prop3},  $\mE(n_{\mathcal{I}}) \gtrsim \sum_{k=1}^K n_k/\max\{d_k,1\}$,   where $d_k^2 ={(\bdelta^{(k)})^{\top} \bSig^{(k)} \bdelta^{(k)}}$. Thus if $d_k=O(1)$ uniformly over $k$, we can achieve a rate of the order $\|\hat \bbeta_{ora}^{(0)} - \bbeta^{(0)}\|^2=O_p\left({s_0 {\rm log}\,p}/(n_0 + \sum_{k=1}^{K} n_k)\right)$ which basically makes full use of the $K$ source data. On the other hand, when $d_k$ diverges, $\mathbb{E}(n_{\mathcal{I}})$ is of the order $\sum_{k=1}^K (n_k/d_k)$, which can still be larger than $n_0$, for example, when $n_k\gg n_0 d_k/K$ for all $k$. In this scenario, the oracle RIW-TL still enjoys a faster convergence rate than LASSO. We emphasize again that the convergence rate in the theorem depends on $d_k$ which is allowed to diverge but does not require $\bbeta^{(k)}, k=1,\cdots, K$, to be sparse.

Based on this theorem, we have compared the oracle RIW-TL and the oracle Trans-Lasso  in Table \ref{tab:0} in a simple case when there is only one source data. Both methods have the  same rate  $s_0{\rm log}\,p/(n_0 + n_1)$
when $\|\bm\delta^{(1)}\|_1 =O(s_0\sqrt{n_0 {\rm log}\,p}/(n_0+n_1))$. When the magnitude of $\bm\delta^{(1)}$ becomes larger but $\|\bm\delta^{(1)}\|_1 = o(s_0\sqrt{{\rm log}\,p/n_0})$, the oracle Trans-Lasso estimator has a  convergence rate faster than the LASSO but slower than $s_0 {\rm log}\,p/(n_1 + n_0)$ that is the rate of the oracle RIW-TL.
As $\|\bm\delta^{(1)}\|_1\gtrsim s_0\sqrt{{\rm log}\,p/n_0}$, the convergence rate of oracle Trans-Lasso is at best the same as LASSO, while the oracle RIW-TL can still be better than LASSO if $d_1=o(n_1)$.

\vspace{2mm}
\noindent\textbf{Oblivious to covariate heterogeneity.}
We note that the theorem and the theory in the sequel apply where $\bm x^{(k)}$ follow different distributions for different $k$. In contrast,  for example in \cite{li2021}, it is required that the degree of heterogeneity in the design matrices of the sources is small or moderate. A similar assumption is made in \cite{tian2022} for transfer learning in generalized linear model, among others. One fundamental reason RIW-TL is design-oblivious is that it operates on importance weights dependent only on the conditional distribution of $y$ given $\bm x$.

\section{RIW-TL in Practice} \label{sec:riw-tl in practice}
We have discussed the properties of the oracle RIW-TL which relies on knowing $\mathcal{I}_k$ in sample selection and the weights $\omega_i^{(k)}$ for selected observations. In practice, they are seldom provided. This section provides a data-driven approach to estimate these two sets of unknown quantities.

Recall the definition $\mathcal{I}_k$ in \eqref{weight-constraint} involving constraints on $\eta_i^{(k)}+\epsilon_i^{(k)}$ and $\eta_i^{(k)}$.  Since $y_i^{(k)}-(\bm x_i^{(k)})^\top \bbeta^{(0)}=\eta_i^{(k)}+\epsilon_i^{(k)}$, we can estimate the latter as $y_i^{(k)}-(\bm x_i^{(k)}) ^\top \bbeta^{(0)}$ by plugging in a preliminary estimate of $\bbeta^{(0)}$. For $\eta_i^{(k)}=(\bm x_i^{(k)}) ^\top(\bbeta^{(k)}-\bbeta^{(0)})$, we can plug in this preliminary estimate of $\bbeta^{(0)}$ and a preliminary estimate of $\bbeta^{(k)}$. For the weight in \eqref{weight2}, we replace $\epsilon_i^{(k)}$ by the residual $y_i^{(k)}-(\bm x_i^{(k)})^\top \bbeta^{(k)}$ in which $\bbeta^{(k)}$ is replaced by its preliminary estimator. Then for the denominator, if $f_{\epsilon^{(k)}}$ is from a Gaussian distribution, we just need to estimate the variance for it to be fully specified, which is  {relatively easy} but will be prone to the misspecification of the distribution. Therefore in this section, we tackle the more challenging, and perhaps more interesting case, where $f_{\epsilon^{(k)}}$ is fully nonparametric. As it is a one-dimensional density, we can estimate $f_{\epsilon^{(k)}}$ via any density estimator for which the literature is vast.  For the numerator $f_\epsilon$, we take $f_\epsilon(t)=[f_{\epsilon^{(k)}}(t)+f_{\epsilon^{(k)}}(-t)]/2$ and thus $f_\epsilon$ is source dependent. In this paper, we use kernel density estimation for estimating $f_{\epsilon^{(k)}}$
and SCAD for obtaining preliminary estimates of $\bbeta^{(0)}$ and $\bbeta^{(k)}$.

It turns out, however, that plugging these estimates directly into \eqref{ora} for estimating $\bbeta^{(0)}$
incurs non-ignorable bias. In high-dimensional linear regression,  a similar phenomenon was pointed out by \cite{fan2012variance} that using a plugging-in estimator to estimate the variance of the errors will underestimate. As a remedy, they proposed refitting which motivated our cross-fitting procedure below. Before presenting the detailed algorithm, we provide a schematic summary of our algorithm in Figure  \ref{diagram}. Splitting each source into three subsets, in Step 1 we obtain preliminary estimates of $\bbeta^{(k)}$ on the first subset. This allows us to obtain kernel density estimate of the residuals based on the data in the second subset as in Step 2, and subsequently construct the sample selection sets and obtain the weights for the data in the third subset, as in Step 3. In Step 4, we obtain one RIW-TL estimate and this process is averaged over three permutations of the three subsets as in Step 5.
\begin{figure}[!htbp]
  \centering
  \includegraphics[width=5in]{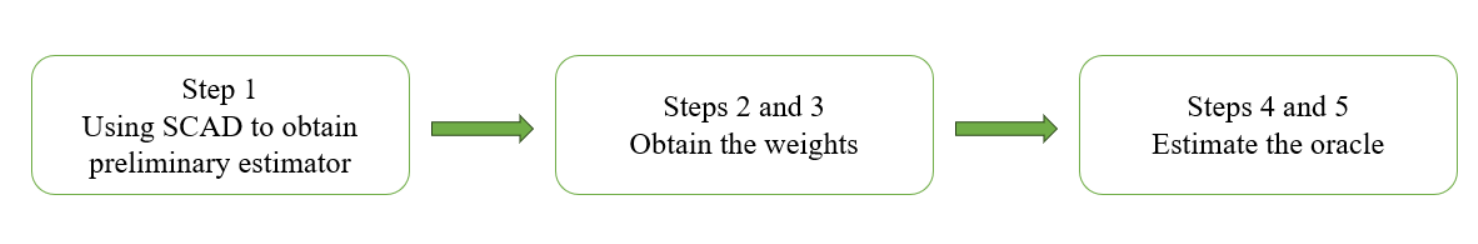}
  \caption{The diagram of the estimation procedure.}\label{diagram}
\end{figure}

Formally, we randomly split $[n_k]$, the index set of the observations in the $k$th source, into three (roughly) equally-sized subsets. We denote these subsets as $\mathcal{D}_{kj}$, $j=1,2,3$. The algorithm is formally presented in Algorithm \ref{algo1}.

\begin{algorithm}[h]
\caption{{\bf  Cross-fitting Algorithm for RIW-TL with Kernel Density Estimation}}\label{algo1}
\begin{algorithmic}
\Require Target data $\mathcal{S}^{(0)}$ and $K$ source data $\{\mathcal{S}^{(k)}\}_{k = 1}^K$.
\Ensure $\hat\bbeta^{(0)}$.
\State {\it Step 1:}
Estimate $\bbeta^{(k)}$ as $\tilde \bbeta^{(k)}$ using observations with indices in $\mathcal{D}_{k1}$ by SCAD for {$k=0,\cdots, K$.}

\State {\it Step 2:}
Employ kernel density estimation on the residuals in $\mathcal{D}_{k2}$ to estimate $f_{\epsilon^{(k)}}$.
For any $t$,  $f_{\epsilon^{(k)}}(t)$ is estimated as
     \begin{equation*}
     \hat f_{\epsilon^{(k)}}(t) = \frac{1}{|\mathcal{D}_{k2}|b_k} \sum\limits_{i\in \mathcal{D}_{k2}} K\left(\frac{|t - \hat\epsilon^{(k)}_{i}|}{b_k} \right),
     \end{equation*}
where $K(\cdot)$ is the kernel function with bandwidth $b_k$ for source $k$ and $\hat\epsilon^{(k)}_{i} = y_i^{(k)} -  (\bm x_i^{(k)})^{\top} \tilde \bbeta^{(k)}$  for  $i \in \mathcal{D}_{k2}$.

\State {\it Step 3:}
 Estimate the weights  $\omega_i^{(k)}$ for  $i\in \mathcal{D}_{k3}$ as
      \begin{equation}\label{hat-weight}
      \hat \omega_i^{(k)} = \frac{\hat f_{\epsilon}(y_i^{(k)}-(\bm x_i^{(k)})^\top\tilde\bbeta^{(0)})}{\hat f_{\epsilon^{(k)}}(y_i^{(k)} - (\bm x_i^{(k)})^\top\tilde\bbeta^{(k)})},
      \end{equation}
where  $\hat f_{\epsilon}(t) = [\hat f_{\epsilon^{(k)}}(t) + \hat f_{\epsilon^{(k)}}(-t)]/2$ is the estimator of $f_{\epsilon}(t)$.
The estimated subset for $\mathcal{D}_{k3}$ is denoted as
    $$
    \hat {\mathcal{I}}_{k3} = \{i \in \mathcal{D}_{k3}: |y_i^{(k)}-(\bm x_i^{(k)})^\top\tilde\bbeta^{(0)}| \leq A,~ |\hat \eta_i^{(k)}| \leq M \},
    $$
    where $\hat \eta_i^{(k)}= (\bm x_i^{(k)})^{\top}(\tilde \bbeta^{(k)} - \tilde \bbeta^{(0)})$  for  $i \in \mathcal{D}_{k3}$.

\State {\it Step 4:}
     With the notations $\hat {\mathcal{I}}_{03} = \mathcal{D}_{03}$ and $\omega_i^{(0)} = 1$ for $i \in \hat {\mathcal{I}}_{03}$.
         The final estimator $\hat\bbeta_1^{(0)}$ can be obtained as
\begin{equation}\label{est-unknowW}
\hat \bbeta_1^{(0)} = \underset{\bm\beta \in \mathbb{R}^p}{\rm argmin}~ \frac{1}{|\mathcal{D}_{03}|+\sum_{k=1}^K |\mathcal{D}_{k3}|}
\left\{
\sum\limits_{k = 0}^K \sum\limits_{i \in \mathcal{D}_{k3}} {\rm I}\{i \in \hat {\mathcal{I}}_{k3}\}
 \hat \omega_i^{(k)} \left(y_i^{(k)} - (\bm x_i^{(k)})^\top \bbeta\right)^2
\right\} + \lambda\|\bbeta\|_1.
\end{equation}

\State {\it Step 5:}
Change the order of  $\mathcal{D}_{k1}$, $\mathcal{D}_{k2}$ and $\mathcal{D}_{k3}$ as in the steps above. Denote the estimator by replacing $(\mathcal{D}_{k1}, \mathcal{D}_{k2}, \mathcal{D}_{k3})$ with $(\mathcal{D}_{k2}, \mathcal{D}_{k3}, \mathcal{D}_{k1})$ as $\hat \bbeta_2^{(0)}$, and the estimator by replacing $(\mathcal{D}_{k1}, \mathcal{D}_{k2},\mathcal{D}_{k3})$ with  $(\mathcal{D}_{k3},\mathcal{D}_{k1}, \mathcal{D}_{k2})$ as $\hat \bbeta_3^{(0)}$. The final estimator based on this cross-fitting scheme is taken as
\begin{equation}\label{hat-ker-beta}
\hat \bbeta^{(0)} = (\hat \bbeta_1^{(0)} + \hat \bbeta_2^{(0)} + \hat \bbeta_3^{(0)})/3.
\end{equation}
\end{algorithmic}
\end{algorithm}

For $b_k$, the bandwidth parameter, and $A$ and $M$, the two factors in the constraint defining $\mathcal{I}_k$, we choose them via $J$-fold cross validation that picks their optimal combination to minimize the prediction error of the response on the hold-out set. Of course, this would involve choosing $K+2$ parameters on a grid which is computationally intensive. As a compromise, we set $b_k$ to be the same as long as the smoothness of $f_{\epsilon^{(k)}}$ are roughly the same, and $M = 2A$ in our implementation.

\subsection{Properties of RIW-TL}\label{sec:Properties of RIW-TL}
To study the properties of RIW-TL obtained via {cross fitting}, {for each $j = 1,2,3$,} we need to study those of the selected sample sets $\hat {\mathcal{I}}_{kj}$ which are seen to estimate
 $$
     {\mathcal{I}}_{kj} = \{i \in \mathcal{D}_{kj}: |y_i^{(k)}-(\bm x_i^{(k)})^\top\bbeta^{(0)}| \leq A,~ |\eta_i^{(k)}| \leq M \},
 $$
and the properties of the estimators of the corresponding weights in Step 3. Our strategy to establish the properties of $\hat {\mathcal{I}}_{kj}$ is to bound this set in a suitable way. Towards this, we define
$$
\mathcal{I}_{kj}^{-} = \{i \in \mathcal{D}_{kj}: |\epsilon_i^{(k)} + \eta_i^{(k)}| \leq A - \alpha_n,~|\eta_i^{(k)}| \leq M - \alpha_n\},
$$
$$
\mathcal{I}_{kj}^{+} = \{i \in \mathcal{D}_{kj}: |\epsilon_i^{(k)} + \eta_i^{(k)}| \leq A + \alpha_n,~|\eta_i^{(k)}| \leq M + \alpha_n\},
$$
that serve as  ``lower bound" and   ``upper bound", respectively. The parameter $\alpha_n$ in the two sets  is taken as
\begin{equation}\label{alpha_n}
\alpha_n = 2 \underset{1\le k\le K}{\rm max} \psi_k \underset{0\le k\le K}{\rm max} \|\tilde \bbeta^{(k)} - \bbeta^{(k)}\|_1~~\text{with}~~\psi_k = \underset{i \in \mathcal{D}_k}{\rm max}~ \|\bm x_i^{(k)}\|_{\infty},
\end{equation}
which, as will be discussed later, has an order  $\alpha_n=o_p(1)$ under mild conditions.  {Note that in our algorithm, $\tilde\bbeta^{(k)}$ is constructed based on the data with indices outside  $\mathcal{D}_{kj}$ for all $k$.  Hence given $\bm x_{i}^{(k)}$'s,  $\alpha_n$ can be viewed as a constant when    $\hat{\mathcal{I}}_{kj}$ is under consideration.}
Denote
 $\mathcal{I}_k^{-} = \cup_{j=1}^3 \mathcal{I}_{kj}^{-}$ and $n_{\mathcal{I}_k^{-}} = \Sigma_{j = 1}^3 n_{\mathcal{I}_{kj}^{-}}$ with $n_{\mathcal{I}_{kj}^{-}} = |\mathcal{I}_{kj}^{-}|$, and define  $(\mathcal{I}_k^{+}, n_{\mathcal{I}_k^{+}}, n_{\mathcal{I}_{kj}^{+}})$ likewise.
 To  establish  the properties of RIW-TL estimator, we assume the following conditions.

\begin{enumerate}[]
\item
 \textbf{Condition 4.} For each $j = 1, 2, 3$,  the RE Condition 3 holds with $\mathcal{I}_k$  replaced by $\mathcal{I}_{kj}^{-}$ for all $k$, and $n_{\mathcal{I}}$ replaced by $\Sigma_{k = 1}^K n_{\mathcal{I}_{kj}^{-}}$.

\item
 \textbf{Condition 5.} For $k=0, 1, \cdots, K$, the initial estimator $\tilde \bbeta^{(k)}$ of  $\bbeta^{(k)}$ is consistent in the sense that the $\ell_1$ rate of convergence satisfies $\|\tilde \bbeta^{(k)} - \bbeta^{(k)}\|_1 = O_p(\gamma_k)$ where $\gamma_k = o(1)$.

\item
 \textbf{Condition 6.} The kernel $K(\cdot)$ is symmetric about zero and has bounded continuous derivative of the first order. The  bandwidths $b_k$ satisfy $b_k = o(1)$ and $b_k^2 \gg \psi_k \gamma_k$ for $k = 1,\cdots,K$.

\item
 \textbf{Condition 7.} The maximum element of $\bm \Sigma^{(k)}$ satisfies $\max\limits_{1 \le i,j \le p} \sigma_{ij}^{(k)}\leq C < \infty$ uniformly over $k=0,\cdots, K$.
\end{enumerate}

As shown in the following Proposition \ref{prop4}, we have  $\mathcal{I}_{kj}^{-} \subseteq \hat {\mathcal{I}}_{kj}$  in probability tending to 1.  Condition 4 ensures  that the restricted eigenvalue condition can be applied on the subset $ \mathcal{I}_{kj}^{-}$.
  In Condition 5,   $\bbeta^{(k)}$ can be  exactly or approximately sparse.
 When { $\bbeta^{(k)}$ is exactly sparse with $s_k= |{\rm supp}(\bbeta^{(k)})|$} and $\tilde {\bm \beta}^{(k)}$ is the {LASSO} or SCAD estimator, we have $\gamma_k = s_k \sqrt{{\rm log}\,p / n_k}$, and when  $\tilde {\bm \beta}^{(k)}$ is the debiased {LASSO} estimator \citep{Zhang2014}, one   has   $\gamma_k =  s_k\sqrt{1/n_k}$. For the order of $\alpha_n$, note {that $\psi_k = O_p(\sqrt{{\rm log}\,(n_k p)})$}
  when $\bm x_i^{(k)} \sim N_p(\bm 0,\bSig^{(k)})$. Consequently, it holds that
  \begin{equation}\label{upp-alpha-n}
  \alpha_n \lesssim \max\limits_{0\le k\le K} \gamma_k \max\limits_{1\le k\le K}\sqrt{{\rm log}\,(n_k p)}
  \end{equation}
  in probability. Combining wit  the order of $\gamma_k$, one can see that the condition  $\alpha_n= o_p(1)$ is mild.

In Condition 6, the quantity {$\psi_k \gamma_k$} is from $|\hat\epsilon_i^{(k)}-\epsilon_i^{(k)}|$, which depends on the error of the estimator $\tilde {\bm \beta}^{(k)}$.   Suppose that { $\psi_k \gamma_k \sim n_k^{-c_1}$.} When  $ 2/5 < c_1\le 1/2$, the commonly used optimal bandwidth $b_k \asymp n_k^{-\frac{1}{5}}$ can be adopted; when  $0< c_1 \leq 2/5$,   under-smoothing  is required. In other words, if $\tilde {\bm \beta}^{(k)}$ is a good estimator, the commonly used optimal bandwidth is still applicable;  otherwise under-smoothing is necessary. We first characterize $\hat {\mathcal{I}}_{kj}$ as an estimator of $ {\mathcal{I}}_{kj}$.

\bep\label{prop4}
Under Conditions 1 and 5, for $k = 1,\cdots,K$ and  $j = 1,2,3$, we have the following conclusions: ($i$)
  {$\mathcal{I}_{kj}^-\subseteq \hat{\mathcal{I}}_{kj} \subseteq \mathcal{I}_{kj}^{+}$}; ($ii$) $ \mE(n_{\mathcal{I}_k^{-}}) \asymp \mE(n_{\mathcal{I}_{kj}^-})$ and
    $n_{\mathcal{I}_k} \asymp \mE(n_{\mathcal{I}_k}) \asymp \mE(n_{\hat{\mathcal{I}}_{kj}}) \asymp \mE(n_{\mathcal{I}_{kj}^-})$ with probability tending to 1.
\eep

 Proposition \ref{prop4} implies that $\mE(n_{\hat{\mathcal{I}}_{kj}})$ has the same order as $\mE(n_{\mathcal{I}_k})$, which enables us to obtain the explicit convergence rate in Theorem \ref{the2}.

\begin{lemma}\label{lemma1}
Suppose that Conditions 1, 5 and 6 are satisfied. For {$j = 1,2,3$},  as $\min\limits_{0\le k\le K} n_k\to \infty$, it holds that
$$
\max\limits_{1\le k\le K}\max\limits_{i \in \hat {\mathcal{I}}_{kj}} |\hat \omega_i^{(k)}/ \omega_i^{(k)} - 1| =
O_p (\alpha_n + q_n),
$$
where $\alpha_n$ is defined in \eqref{alpha_n} and $q_n = u_n + v_n$ with ${u_n = \underset{1 \leq k \leq K}{\rm max}  (\psi_k \gamma_k / b_k^2)=o(1)}$, and $v_n = \underset{1 \leq k \leq K}{\rm max}
\{b_k^2 + ({\rm log}\, n_k /(n_k b_k))^{1/2}\} = o(1)$ which is the error in kernel density estimation.
\end{lemma}

Lemma \ref{lemma1} shows that the estimated weights converge to the true weights for all $i \in \hat {\mathcal{I}}_{kj}$.

\begin{remark}\label{remark2}
There is a subtle issue associated with the bias for the estimated subsets $\hat {\mathcal{I}}_{kj}$'s. For the true subset $\mathcal{I}_k$ in \eqref{weight-constraint}, recall that the bias  $\text{Bias}(f_\epsilon, \mathcal{I}_k)$  can be  eliminated when $f_{\epsilon}$ is symmetric. However, for the  estimated subset $\hat {\mathcal{I}}_{kj}$,  the bias $\text{Bias}(f_\epsilon, \hat {\mathcal{I}}_{kj})$ is no longer zero  due to estimation errors. This bias is small nevertheless. To see this,
take the subset $\hat{\mathcal{I}}_{k3}$ as an example for illustration.    When $f_{\epsilon}$ is a symmetric function,  the bias  becomes
$$
\text{Bias}(f_\epsilon, \hat {\mathcal{I}}_{k3}) = \int_{ i\in \hat {\mathcal{I}}_{k3}}  (\epsilon_i^{(k)} + \eta_i^{(k)}) f_{\epsilon}(\epsilon_i^{(k)} + \eta_i^{(k)}) d \epsilon_i^{(k)} =  \int_{A - r_i^{(k)0}}^{A + r_i^{(k)0}} t f_{\epsilon}(t) dt \approx 2 A f_{\epsilon}(A) r_i^{(k)0},
$$
   where  $r_i^{(k)0} = (\bm x_i^{(k)})^{\top}(\tilde \bbeta^{(0)} - \bbeta^{(0)})$.
 One can see that this bias depends on $r_i^{(k)0}$, implying that the estimated error from  $\tilde \bbeta^{(0)}$ will appear in the final convergence rate. But the constant $A f_{\epsilon}(A)$ can be small (i.e. equivalently when $A$ is large). Particularly, when $f_{\epsilon}$ is the density of the standard normal distribution, we see that $A f_{\epsilon}(A) r_i^{(k)0} \asymp A {\rm exp}(-A^2/2) r_i^{(k)0}$,
which is small for large A.
\end{remark}

Combining Lemma \ref{lemma1} and the consistency of $\hat{\mathcal{I}}_{kj}$ in Proposition \ref{prop4}, the convergence of the RIW-TL estimator can be proved.
 Let $\hat {\mathcal{I}} = \cup_{k=1}^K \hat {\mathcal{I}}_k$ and  define $\mu_{\rm max} =  \Sigma_{k=1}^K \Sigma_{i \in \mathcal{D}_{k}} \|\bm x_i^{(k)}\|_{\infty} / \Sigma_{k = 1}^K n_k$, which is the empirical version of $\Sigma_{k = 1}^K \pi_k \mE(\|\bm x^{(k)}\|_\infty)$ with $\pi_k = n_k / \Sigma_{k = 1}^K n_k$. Note that
when $\bm x_i^{(k)}$ is Gaussian, we have $\mu_{\rm max} = O_p(\sqrt{{\rm log}\,p})$. When the predictors are bounded in $\ell_\infty$ norm, we have $\mu_{\rm max} = O(1)$. Bounded predictors are commonly seen in image data \citep{Shorten2019} and gene expression data \citep{Vinas2022}, among many others.  Theorem \ref{the2} below establishes the convergence rate of the RIW-TL estimator computed by the cross-fitting algorithm.

\begin{theorem}[Convergence rate of RIW-TL]\label{the2}
Suppose that Conditions 1-7 are satisfied. Let $\lambda = 2 (\lambda_n^{(1)} + \lambda_{n}^{(2)})$ with
   $$
 \lambda_n^{(1)} = C_1 \rho_{\mathcal{I}} \sqrt{\log p / (n_0 + \mE(n_{\mathcal{I}}))},~~~~
 \lambda_{n}^{(2)}= C_2
 [\mu_{\rm max} (q_n + \alpha_n) + C_{f_{\epsilon},A}  \gamma_0],
 $$
 where $\rho_{\mathcal{I}}$ is the sample usage rate discussed in Corollary \ref{coro1} and $C_{f_{\epsilon},A} = \max\limits_{a\in [A - \theta_0,A + \theta_0] }a f_{\epsilon}(a)$ for some $\theta_0$ sufficiently small. 
 As $\min\limits_{0\le k\le K}n_{k}\to \infty$, it follows that
\begin{equation}\label{final-rate}
\|\hat \bbeta^{(0)} - \bbeta^{(0)}\|_2^2
 = O_p \left[
 \frac{s_0 {\rm log}\,p}{n_0 + \mathbb{E}(n_{\mathcal{I}})} +
 \rho_{\mathcal{I}}^{-2}
s_0 \{\mu_{\rm max} (q_n + \alpha_n) + {C_{f_{\epsilon},A}}  \gamma_0\}^2
\right],
\end{equation}
where $\alpha_n$ is defined in \eqref{alpha_n}, $q_n$ is defined in Lemma \ref{lemma1}, and $\gamma_0$ is the convergence rate of the initial estimator $\tilde {\bm\beta}^{(0)}$.
\end{theorem}

  The rate in this theorem involves two terms. The first term is the same as the rate of the oracle RIW-TL in Theorem \ref{the1}.  The second term involves the estimation errors from estimating the weights $\omega_i^{(k)}$ and selecting the sample ${\mathcal{I}}_{kj}$. Specifically,  $q_n+\alpha_n$ comes from the error of the weights estimator as in Lemma \ref{lemma1}  and $\gamma_0$ comes from  that of $\hat{\mathcal{I}}_{kj}$. {It is worth pointing out that  the distances between the target and sources in terms of $d_k$ affect only the quantities $\rho_{\mathcal{I}}$ and $\mE(n_{\mathcal{I}})$. Particularly, when $d_k=O(1)$ for all $k\ge 1$, it holds that  $\mE(n_{\mathcal{I}})\asymp \sum_{k=1}^K n_k$ and $\rho_{\mathcal{I}}\asymp1$.
To obtain a  more explicit expression of the second term on the right hand side of  \eqref{final-rate}, we consider the special cases   in the following  Corollary \ref{coro-2}. Recall  that $\gamma_k$ is the rate  of  $\|\tilde\bbeta^{(k)}-\bbeta^{(k)}\|_1$ for $k=0,\cdots,K$.

 \begin{corollary}\label{coro-2}
  Assume that the conditions in Theorem \ref{the2} and  ($i$) the conditions in Corollary \ref{coro1} hold with
all the eigenvalues of $\bSig^{(k)}$ bounded away from 0 and $\infty$;
 ($ii$) the sample sizes of sources are much larger than the target such that $\max\limits_{1\le k\le K}\gamma_k=o(\gamma_0)$; ($iii$)  $s_0\asymp 1$; and ($iv$) $\|\bm\delta^{(k)}\|=O(1)$ for all $k\ge 1$. Then we have
 $$
 \|\hat \bbeta^{(0)} - \bbeta^{(0)}\|_2^2
 = O_p \left\{
 \frac{s_0 {\rm log}\,p}{n_0 + \sum_{k=1}^K n_k  } +  \gamma_0^2 \log(p) \log (n_1 p)\right\}.
 $$
 \end{corollary}

In Theorem \ref{the2} and Corollary \ref{coro-2}, we allow $\|\bm\delta^{(k)}\|\nrightarrow 0$ for all $k\ge 1$ thus allowing the target and the sources to be not close. The price paid, however, is that the rate obtained is conservative,  due to some technical reasons, in the sense that the error rate seems no better than  $\gamma_0^2$. Nevertheless, our numerical experience suggests that RIW-TL performs better than LASSO, which makes sense {since in the bound, the constant $C_{f_{\epsilon},A}$ can be small when $A$ is large.} Theoretically, to obtain a rate better than that of  LASSO,  we may use Trans-Lasso as an initial estimator. In the next section, we investigate an alternative version of RIW-TL when $f_{\epsilon}$ is taken as the density function of a uniform distribution. In this case, an average distance between the sources and the target rather than $\gamma_0$  appears in the error rate. As we will see,  this choice of $f_{\epsilon}$ is well suited for the relatively easy transfer learning scenario {where most of the sources are informative.}

\vspace{2mm}
\noindent{\textbf{Difference between Trans-Lasso and RIW-TL.}}
To appreciate the results in Theorem \ref{the2} and Corollary \ref{coro-2} further, it is instructive to make a comparison with the results in \cite{li2021}. For the latter, to estimate the oracle Trans-Lasso, an aggregation procedure called Trans-Lasso is proposed to combine a collection of candidate estimators, each of which is based on an estimate of the informative source set $\mathcal{J}$ as defined in Section \ref{sec:tl}. In particular, \cite{li2021} showed that Trans-Lasso achieves the following convergence rate
 \[\text{Trans-Lasso Rate:}~~
\frac{s_0 {\rm log}\,p}{n_0+n_{\mathcal{J}}} + h\sqrt{\frac{{\rm log}\,p} {n_0}} \wedge h^2+\frac{\log K}{ n_0},\]
where $h$ satisfies $h\ge \max_{k\in \mathcal{J}}\|\bm\delta^{(k)}\|_1$.
{When $n_{\mathcal{J}}\gg n_0$, the Trans-Lasso rate  is determined by $ h\sqrt{ \log p / n_0} $ if $h\gtrsim \sqrt{\log p/n_0}$, and $h^2+\log K/n_0$ otherwise. Under the settings in  Corollary \ref{coro-2} when $\tilde\bbeta_0$ is the LASSO estimator, $\gamma_0^2$ has the order  $s_0^2\log p/n_0$. Thus, if
 $ h\sqrt{ \log p / n_0}\ge \gamma_0^2 \log(p) \log (n_1 p)$ (equivalently  $h\gtrsim n_0^{-1/2} s_0^2 (\log p)^{3/2}(\log n_1p)$), RIW-TL is better than Trans-Lasso, demonstrating the advantage of RIW-TL when the distances between sources and the target is moderate or large. However, when $h$ is small, the error rate of Trans-Lasso can be smaller than that of  RIW-TL.

\bigskip
In view of the above comparison, we consider a special case where both $n_0$ and $h$ are small (e.g. $h\lesssim \sqrt{\log p/n_0}$).  From the discussion above,  Trans-Lasso has rate $h^2+n_0^{-1}\log K$, which is equivalent to $n_0^{-1}$  up to some logarithmic terms; the same is true for  RIW-TL. These results hold even if $n_k(k\ge 1)$ are very large. This may be a limitation for transfer learning in small target data because increasing the number of auxiliary samples will not improve the convergence rate, even when sources are close to the target.} Hence both methods are conservative in this setting, substantially inferior to their oracle versions respectively that have a rate inversely proportional to $n_0+\sum_{k=1}^K n_k$ in the best case scenario.  In the following section, we design a variant of RIW-TL by taking  $f_{\epsilon}$ as the density of a  uniform distribution and show that it achieves the same convergence rate as the oracle under this particular setting.

\section{An Alternative RIW-TL}\label{sec:riw-tl-u}
In this section,  we develop an alternative version of RIW-TL, which will be referred to as RIW-TL-U, by taking $\epsilon$ from a uniform distribution; that is, $\epsilon \sim U[-T, T]$ such that  $f_\epsilon(t)=(2T)^{-1} {\rm I}(|t|\le T)$, where $T$ is a tuning parameter. The denominator is still estimated via kernel density estimation.
Similar to what we have provided on RIW-TL, we discuss the rate of the oracle RIW-TL-U assuming known weights and sample selection sets, before proceeding to discuss the rate of RIW-TL-U when they are estimated.

Before delving into technical details, we quickly summarize our findings. It turns out that the convergence rate of RIW-TL-U is independent of   the convergence rate of $\tilde {\bm\beta}^{(0)}$ as long as $\tilde {\bm\beta}^{(0)}$  is consistent. The reason for this is that for the numerator in \eqref{weight2}, we only need a consistent estimator of the endpoints, thus relaxing the dependence on $n_0$. The price to pay is that a term on the average bias appears in the convergence rate, which fortunately can be small when $\bm\delta^{(k)}$'s are small on average. Under this setting, we show that RIW-TL can achieve the rate of its oracle version.

The development of this section largely follows the last two sections. Since we use a different distribution for $f_\epsilon$, we will introduce a new set of notations to be consistent with this choice. First we write the weight associated with the data $\bm z_i^{(k)}$ as
\begin{equation}\label{weights-U}
\omega_{i,T}^{(k)} = \frac{f_{\epsilon}(\epsilon_i^{(k)} + \eta_i^{(k)})}{f_{\epsilon^{(k)}}(\epsilon_i^{(k)})} = \frac{1}{2T f_{\epsilon^{(k)}}(\epsilon_i^{(k)})} {\rm I}(|\epsilon_i^{(k)} + \eta_i^{(k)}| \leq T),
\end{equation}
while for sample selection, we look at the observations in the following subset
\begin{equation}\label{weight-constraint-2}
\mathcal{I}'_k = \{i \in \mathcal{D}_k: |\epsilon_i^{(k)}| \le A,~ |\eta_i^{(k)}| \leq M \}, \qquad k = 1,\cdots,K,
\end{equation}
where $A$ and $M$ are  parameters satisfying $A + M \leq T$. Define $n_{\mathcal{I}'_k} = |\mathcal{I}'_k|$ and $n_{\mathcal{I}'}=\sum_{k=1}^K n_{\mathcal{I}'_k}$. First, we have the following result regarding the sample usage rate.

\bep\label{prop5}
If $\epsilon_i^{(k)}$ is distributed as $N(0,1)$ and $\bm x_i^{(k)}$ follows $N(\bm 0, \bSig^{(k)})$ for $k=1,\cdots, K$,
it holds that $\mE(n_{\mathcal{I}'_k}) \gtrsim n_k/\max\{d_k,1\}$ for $k=1,\cdots, K$.
Furthermore, the conclusions of Corollary \ref{coro1} holds for the sample usage rate defined as $\rho_{\mathcal{I}'}=(n_0 + \mathbb{E}(n_{\mathcal{I}'}))/(n_0 + \sum_{k=1}^K n_k)$.
\eep

Similar to  $\tilde {\bm \beta}_{ora}^{(0)}$,  we define the oracle RIW-TL-U estimator $\tilde \bbeta_{T,ora}^{(0)}$ by replacing $(\mathcal{I}_k, \omega_i^{(k)})$ with $(\mathcal{I}'_k, \omega_{i,T}^{(k)})$ in \eqref{ora}, of which the convergence rate is presented in the following theorem.   Define  $h_{\rm ave}=\sum_{k = 1}^K \pi_k \|\bbeta^{(k)} - \bbeta^{(0)}\|_1 $ with $\pi_k = n_k / \sum_{j=1}^K n_j$, which measures the average distance between the target and sources.

\begin{theorem}\label{the3}
 Assume that Conditions 2 and 3 hold. For $k = 1,\cdots,K$, suppose that $\mathcal{I}'_k$ are observed and the weights $\omega_{i,T}^{(k)}$ are known for all $i \in \mathcal{I}'_k$. Let $\lambda = 2 (\lambda_n^{(1)} + \lambda_{n}^{(2)})$ with
   $$
 \lambda_n^{(1)} = C_1 \rho_{\mathcal{I}'} \sqrt{\log p / (n_0 + \mE(n_{\mathcal{I}'}))},~~
 \lambda_{n}^{(2)}= C_2 C_{A,T} h_{\rm ave},
 $$
 for some constants $C_1, C_2 > 0$, where  $C_{A,T} = A/T<1$. As $\min\limits_{0\le k\le K} n_k\to \infty$, it follows that
\begin{equation}\label{oracel-rate-U}
\|\tilde \bbeta_{T,ora}^{(0)} - \bbeta^{(0)}\|^2
 = O_p \left\{
\frac{s_0 {\rm log}\,p}{n_0 + \mathbb{E}(n_{\mathcal{I}'})} +
\rho_{\mathcal{I}'}^{-2} s_0
 h_{\rm ave}^2
\right\}.
\end{equation}
\end{theorem}

Under the conditions of  Proposition  \ref{prop5} and Corollary \ref{coro1}, it holds that $\rho_{\mathcal{I}'}^{-1}\asymp  1$, when $d_k$'s are of order $O(1)$.
  Compared with the results of oracle RIW-TL in    Theorem \ref{the1},  the oracle RIW-TL-U $\tilde \bbeta_{T,ora}^{(0)} $ in  Theorem \ref{the3} has an additional term that  depends on $h_{\rm ave}$.    When $h_{\rm ave}\lesssim \sqrt{\log p / (n_0+\mathbb{E}(n_{\mathcal{I}'}))}$, the oracle RIW-TL-U is of the same order as that of the oracle RIW-TL; otherwise it is inferior to oracle RIW-TL.

When weights are unknown and  $\mathcal{I}'_k$'s are unobserved, similar to Section \ref{sec:riw-tl in practice}, we modify the cross-fitting procedure therein by estimating the weights as
$$
\hat \omega_{i,T}^{(k)} = \frac{1}{2T \hat f_{\epsilon^{(k)}}(\hat \epsilon_i^{(k)})} {\rm I}(|\hat\epsilon_i^{(k)} + \hat\eta_i^{(k)}| \leq T),~~~i \in \mathcal{D}_{k3}.
$$
The estimated sample selection sets are denoted as
$$
\hat {\mathcal{I}}'_{k3} = \{i\in  \mathcal{D}_{k3}: |\hat\epsilon_i^{(k)}| \leq A,~ |\hat \eta_i^{(k)}| \leq M \},
$$
where $\hat\epsilon_i^{(k)} = y_i^{(k)}-(\bm x_i^{(k)})^{\top} \tilde \bbeta^{(k)}$ and $A,M$ are parameters such that $A + M - 2 \theta_0 \leq T$ with a sufficient small constant $\theta_0 > 0$. The requirement $A + M - 2\theta_0 \leq T$ is to guarantee that both the weights $\omega_{i,T}^{(k)}$ and its estimators $\hat \omega_{i,T}^{(k)}$ are nonzero for all $i \in \hat{ \mathcal{I}}'_{k3}$ in  probability tending to 1.

Similar to the definition of $\hat {\bm \beta}_1^{(0)}$ in \eqref{est-unknowW}, we define the estimator $\hat {\bm \beta}_{T,1}^{(0)}$ by replacing $(\hat {\mathcal{I}}_{k3}, \hat \omega_i^{(k)})$ with $(\hat {\mathcal{I}}'_{k3}, \omega_{i,T}^{(k)})$. In the same manner, we can obtain the estimators $\hat {\bm \beta}_{T,2}^{(0)}$ and $\hat {\bm \beta}_{T,3}^{(0)}$. Finally, the RIW-TL-U estimator is defined  as $\hat {\bm \beta}_T^{(0)} = (\hat {\bm \beta}_{T,1}^{(0)} + \hat {\bm \beta}_{T,2}^{(0)} + \hat {\bm \beta}_{T,3}^{(0)})/3$ and the convergence rate of $\hat {\bm \beta}_T^{(0)}$ is established in the following.

\begin{lemma}\label{lemma2}
Under the Conditions 5 and 6, for each $j = 1,2,3$, as $\min_{1\le k\le K} n_k\to \infty$, it holds that
$
\max\limits_{1 \le k \le K} \max\limits_{i \in \hat {\mathcal{I}}'_{kj}} |\hat \omega_{i,T}^{(k)}/ \omega_{i,T}^{(k)} - 1| =
O_p (q_n),
$
where $q_n$ is defined in Lemma \ref{lemma1}.
\end{lemma}

\vspace{-3mm}
 Lemma \ref{lemma2} shows the convergence rate of $\hat \omega_{i,T}^{(k)}$ for all $i \in \hat {\mathcal{I}}'_{kj}$. Compared with Lemma \ref{lemma1}, this convergence rate does not involve  $\gamma_0$, the $\ell_1$ convergence rate of $\tilde \bbeta^{(0)}$, which helps to improve the convergence rate of RIW-TL-U estimator in Theorem \ref{the4}.

\begin{theorem}\label{the4}
Suppose that Conditions 2-7 are satisfied. Let $\lambda = 2 (\lambda_n^{(1)} + \lambda_{n}^{(2)})$ with
   $$
 \lambda_n^{(1)} = C_1 \rho_{\mathcal{I}'} \sqrt{\log p / (n_0 + \mE(n_{\mathcal{I}'}))},~~~~
 \lambda_{n}^{(2)}= C_2 \{\mu_{\rm max} q_n + C_{A,T}(\gamma_{\rm max} + h_{\rm ave})\},
 $$
 where $C_1, C_2 > 0$ are some constants, and $\gamma_{\rm max} = \underset{1 \le k \le K}{\rm max} \gamma_k$ with $\gamma_k$ the convergence rate of $\tilde {\bm \beta}^{(k)}$. As $\min\limits_{0\le k\le K} n_k\to \infty$, it follows that
$$
\|\hat \bbeta_T^{(0)} - \bbeta^{(0)}\|_2^2
 = O_p \left\{
 \frac{s_0 {\rm log}\,p}{n_0 + \mathbb{E}(n_{\mathcal{I}'})} +
 \rho_{\mathcal{I}'}^{-2}
 s_0 \left(\mu_{\rm max} q_n +  \gamma_{\rm max} + h_{\rm ave}\right)^2
\right\}.
$$
\end{theorem}

Theorem \ref{the4} establishes the convergence rate of RIW-TL-U estimator. Similar to the discussion after Theorem \ref{the3}, we have $\rho_{\mathcal{I}'}^{-1}\asymp1$ under conditions of Proposition \ref{prop5}. Different from the results in \eqref{oracel-rate-U} for the  oracle RIW-TL-U,  we have two additional  terms in  $\mu_{\rm max} q_n$ and $\gamma_{\rm max}$, where the former is from the estimation errors of the weights $\hat \omega_{i,T}^{(k)}$ and the latter from that of the subsets $\hat {\mathcal{I}}'_{kj}$'s.
 Both $\mu_{\rm max} q_n$ and $\gamma_{\rm max}$,  depending only on the sample sizes $n_k(k\ge 1)$ of the sources, can be small when $n_k\to \infty (k\ge 1)$. Thus  $h_{\rm ave}$ can be the dominating term in $\mu_{\rm max} q_n +  \gamma_{\rm max} + h_{\rm ave}$
when the sample sizes of sources are sufficiently large. Consequently,   the convergence rate of RIW-TL-U estimator is of the same order as its oracle version.  Compared with the RIW-TL estimator in Theorem \ref{the2} that depends on $\gamma_0$, the estimator $\hat \bbeta_T^{(0)}$ is free from  $\gamma_0$ and can be better than both  RIW-TL and Trans-Lasso when  $h_{\rm ave}$, the average distance between the target and the sources, is small.

\section{Simulation}\label{sec:simulation}
In this section, we compare RIW-TL to Trans-Lasso and  LASSO numerically in synthetic data. For our framework, we implement RIW-TL as discussed in Section \ref{sec:riw-tl in practice}, RIW-TL-U as discussed in Section \ref{sec:riw-tl-u}, and a variant of RIW-TL named RIW-TL-P in which $\epsilon^{(k)}$ is assumed Gaussian with its variance estimated by the sample counterpart and $f_\epsilon$ symmetrizing the estimated $f_{\epsilon^{(k)}}$. In our setup, we let $p = 200$, $n_0 = 150$, $K = 10$ and $n_k = 600$ for $k = 1,\cdots,K$. The errors $\epsilon_i^{(0)}$ and $\epsilon_i^{(k)}$ are all generated independently from the standard normal distribution. The response predictor pairs are all generated from the linear model. To simulate interesting scenarios in which the extent of the informativeness of the source data can vary, we consider the following configurations for generating $\bm x_i^{(k)}$ and $\bbeta^{(k)}$ for $k = 0,\cdots,K$.

\textbf{Covariates.} The covariates $\bm x_i^{(0)}$ for the target data are generated from $N_p(\bm 0,\bSig)$ with $\bSig = (0.5^{|l - l^{'}|})_{l,l^{'} = 1}^p$. For $\bm x_i^{(k)}$, $k=1,\cdots,K$, we either generate from $N_p(\bm 0,\bSig)$, or the multivariate central $t$ distribution with covariance matrix $\bSig$ and five degrees of freedom. Clearly, for the former scenario, the marginal distribution of the covariates {for both the sources and the target} are the same, while the conditional distribution
of $y$ given $\bm x$ is different between the sources and the target if $\bbeta^{(k)} \neq \bbeta^{(0)}$. This is sometimes referred to as \textit{posterior shift} in the literature. For the latter scenario, both the marginal distribution of $\bm x$ and the conditional distribution of $y$ given $\bm x$ are different between the sources and target. This is often referred to as \textit{full distribution shift}. We will examine these two kinds of shifts in separate plots.

\textbf{Coefficients.} For the target data, we set $\bbeta^{(0)} = (\bm 1_{s_0}^{\top}, \bm 0_{p - s_0}^{\top})^{\top}$ and $s_0 = 10$. For generating $\bbeta^{(k)}$ in the source data, we vary their closeness to $\bbeta^{(0)}$ in terms of how many entries differ and how much the entries differ. Inspired by \cite{li2021}, we consider the following scenarios. Recall $[n] = \{1,\cdots,n\}$. Define an index set $\mathcal{B} = [m_{\mathcal{B}}]$ where $m_{\mathcal{B}} = 0,2,4,6,8$ or $10$. Given $\mathcal{B}$, for $1 \leq k \leq K$, we specify $\bbeta^{(k)}$ as follows.
    If $k \in \mathcal{B}$, let
  $$
 \beta^{(k)}_j = \left\{
  \begin{array}{ll}
  \beta_j^{(0)} - 0.5, &  j \in T_k;\\
   \beta_j^{(0)}, & \mbox{otherwise,}
        \end{array}
        \right.
 $$
 where $T_k$ is a random subset of $\{s_0 + 1,\cdots,p\}$ with $|T_k|= d$, and  the values of $d$ will vary according to the specification later;
  if $k \notin \mathcal{B}$, let
  $$
 \beta^{(k)}_j = \left\{
  \begin{array}{ll}
  \beta_j^{(0)} -  1,    &  j \in [s_0];\\
  \beta_j^{(0)} - 0.5, &  j \in  U_k;\\
   \beta_j^{(0)}, & \mbox{otherwise,}
        \end{array}
        \right.
 $$
 where ${U_k}$ is a random subset of $\{s_0 + 1,\cdots,p\}$ with ${|U_k|}= 2 s_0.$

\bigskip

In the above configuration, the magnitude of the difference between different entries between $\bbeta^{(0)}$ and $\bbeta^{(k)}$ is either $0.5$ or $1$ and thus is fixed, while how many entries differ depends on the values of $m_{\mathcal{B}}$ and $d$. A larger value of $d$ corresponds to a larger difference between $\bbeta^{(0)}$ and $\bbeta^{(k)}$ in general. {On the other hand, those $\bbeta^{(k)}$ with $k \in \mathcal{B}$ are relatively close to $\bbeta^{(0)}$. The larger $m_{\mathcal{B}}$ is, the more sources are relatively close to the target}. To reflect the generality of our method, we further consider the case where the magnitude is random and the case where $\epsilon_i^{(0)}$ and $\epsilon_i^{(k)}$'s follow different distributions with the simulation results shown in the supplementary materials to save space. We remark that the results under these setting are qualitatively similar to what is seen in this section.

 For the tuning parameters appearing in the constraints in $\mathcal{I}_k$ of RIW-TL, we let $M = 2 A$ and tune $M$ over a step size $0.5$ grid on the interval $[1,3]$. We take all $b_k$'s in kernel density estimation as the same and tune it on \{0.1,0.2,0.3\} via $5$-fold cross-validation. For RIW-TL-U, {we set  $M=2A=2T/3$ and select $M$ by cross-validation from the interval $[1,3]$ with a step length 0.5}.
 For measuring the performance of each approach, we compute the sum of squared estimation errors (SSE) and sample usage rate (SUR). Depending on the values of $d$ and $m_{\mathcal{B}}$, we consider two kinds of simulations to show the performance of various methods. For the first one, with $d$ taking values  $4$ or $8$, we observe the behavior of the resulting estimators varying $m_{\mathcal{B}} \in \{0,2,4,6,8, 10$\}. For the second one, we take $m_{\mathcal{B}} = 4$ or $8$ as an example and tune $d$ over a step size 4 grid on $[0,32]$. For each simulation configuration, we repeat the experiment 200 times and report the average SSE and SUR. The results for posterior shift are shown in Figure \ref{fig22} and those for the full distribution shift are in Figure \ref{fig33}.  From the simulation results, we can draw the following conclusions.

   \begin{figure}[h]
\centering
  \includegraphics[scale=0.5]{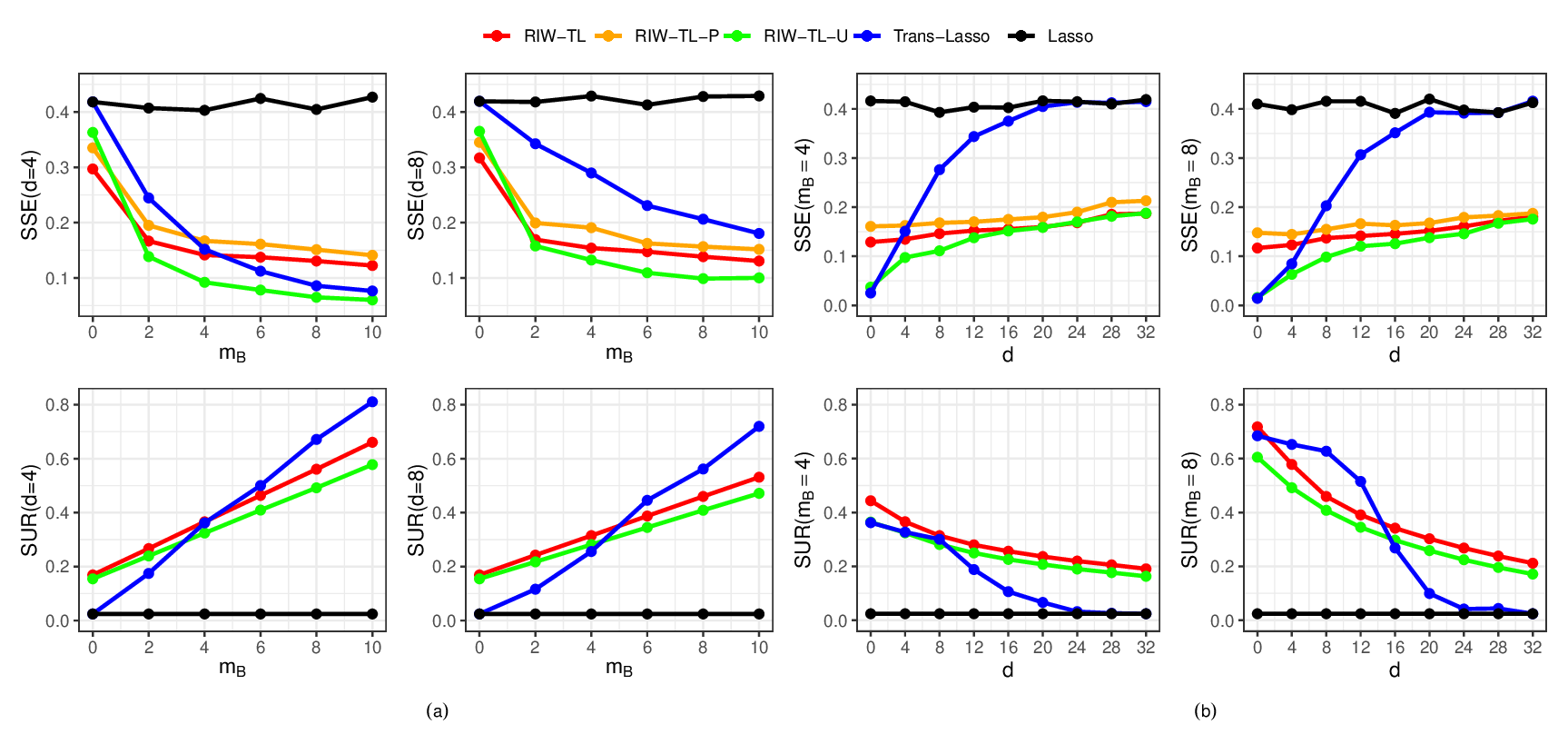}
  \caption{Posterior shift: (a) The estimation errors (the first row) and sample usage rates (the second row) versus $m_{\mathcal{B}}$ for different $d$. (b) The estimation errors (the first row) and sample usage rates (the second  row) versus $d$ for different $m_{\mathcal{B}}$. Note in the SUR plots, RIW-TL-P lines are invisible because they overlap with the corresponding RIW-TL lines. }
  \label{fig22}
\end{figure}

\begin{figure}[!htbp]
\centering
  \includegraphics[scale=0.5]{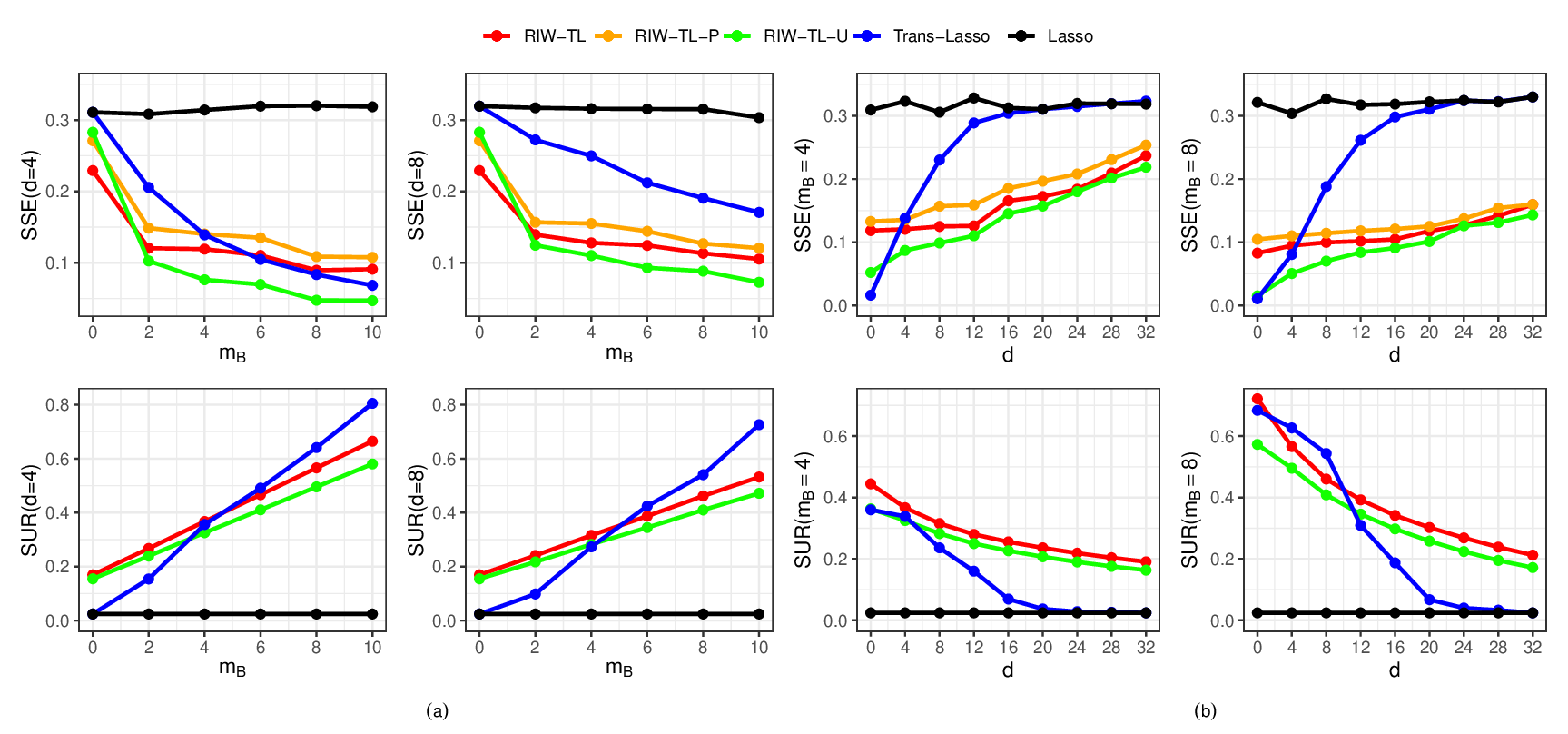}
  \caption{Full distribution shift: (a) The estimation errors (the first row) and sample usage rates (the second  row) versus $m_{\mathcal{B}}$ for different $d$. (b) The estimation errors (the first  row) and sample usage rates (the second row) versus $d$ for different $m_{\mathcal{B}}$. Note in the SUR plots, RIW-TL-P lines are invisible because they overlap with the corresponding RIW-TL lines.}
  \label{fig33}
\end{figure}

\noindent{\textbf{LASSO is inferior to any other transfer learning-based method.}}
From  Figures \ref{fig22} and \ref{fig33}, one can see that its SSE is the worst and its SUR is the smallest in all the cases because no source data is utilized. By borrowing strength from the source data,  the transfer learning methods perform better than LASSO.  The performance of transfer learning-based methods improves in general when $m_{\mathcal{B}}$ increases or $d$ decreases. This is because more source data are similar to the target data.

\noindent{\textbf{RIW-TL type estimators  have advantages over Trans-Lasso especially when $d$ is large or $m_{\mathcal{B}}$ is small.}}  A smaller  $m_{\mathcal{B}}$ or a larger $d$ implies that fewer sources are close to the target because in these cases, $\|\bm\delta^{(k)}\|_1$'s are larger in general.  From the first and the third rows in Figures \ref{fig22} and \ref{fig33}, it is seen that the improvement of  RIW-TL  and its variant (RIW-TL-P, RIW-TL-U)   over Trans-Lasso is significant especially when $d$ {is} large or $m_{\mathcal{B}}$ is small, which coincides with our theory that our method allows effective use of sample even when   $\|\bm\delta^{(k)}\|_1$'s are large.  In addition, we observe that   RIW-TL-U has better performance than Trans-Lasso, RIW-TL and RIW-TL-P. This may be due to that RIW-TL-U requires much weaker conditions on the initial estimators.

\noindent{{\textbf{RIW-TL type estimators make more effective use of the samples in difficult cases.}} It is seen that the SURs of RIW-TL type of methods and Trans-Lasso are increasing functions of $m_{\mathcal{B}}$ (see the second rows of Figures \ref{fig22} and \ref{fig33}) and are decreasing functions of $d$  (see the fourth rows of  Figures \ref{fig22} and \ref{fig33}). When $m_{\mathcal{B}}$ is smaller or $d$ is larger corresponding to the difficult case where the difference between the sources and the target is larger, we can see that RIW-TL type of methods make more effective use of the samples by having larger SURs in general.

\noindent{{\textbf{RIW-TL can be competitive in easy cases where Trans-Lasso is a desired approach.}} When the sources and the target are close, the all-in-or-all-out Trans-Lasso approach will be able to use the source data effectively. This is confirmed by the SURs in Figures \ref{fig22} and \ref{fig33} when either $m_{\mathcal{B}}$ is large or $d$ is small. For these scenarios, we can see that RIW-TL and RIW-TL-P  are slightly inferior to Trans-Lasso as expected because Trans-Lasso uses more data that are close to the target. Even for these scenarios, we can see that  RIW-TL-U still performs similarly to or better than Trans-Lasso, though its SUR is smaller generally. As argued in our theory, this is due to the fact that the average distance $h_{\rm ave}$ is small in these settings.   However, when $d$ is moderate   (e.g. $d = 8$), although  RIW-TL and RIW-TL-U have smaller values of SUR, they have smaller estimation errors than the Trans-Lasso. This is probably due to that Trans-Lasso selects the whole sample of sources that are identified as informative, among which some observations may be far from the target and thus are less useful for transfer.

\noindent{{\textbf{Assuming a parametric distribution for $f_{\epsilon^{(k)}}$ does not gain much even if $\epsilon^{(k)}$ follows this distribution.}} In the simulation presented here and in the supplementary materials, we see that RIW-TL-P underperforms RIW-TL and RIW-TL-U in general even when the error distributions of $\epsilon^{(k)}, k=0,1,\cdots,K$, are Gaussian. This is probably due to that univariate density estimation employed for RIW-TL and RIW-TL-U is fairly accurate or that taking $f_\epsilon$ by symmetrizing $f_{\epsilon^{(k)}}$ is not the best option when the latter is Gaussian. In view of this, we recommend the use of RIW-TL or RIW-TL-U in practice.

\section{Real Data Analysis}
In this section, we apply our method to the Genotype-Tissue Expression (GTEx) data available on \href{https://gtexportal.org/}{https://gtexportal.org/}. Our decision to analyze this dataset is motivated by \cite{li2021} where Trans-Lasso is shown to improve the prediction performance of LASSO.

This dataset includes 38,187 gene expression levels from 49 tissues of 838 human donors. We are interested in predicting the expression level of gene JAM2 (Junctional adhesion molecule B) in brain tissue using other central nervous system (CNS) genes. Mutations in JAM2 have been found to cause primary familial brain classification \citep{Cen2020, Schottlaender2020}. JAM2 is expressed in 49 tissues in our data sets and we use 40 tissues that have more than 150 measurements on JAM2. The association between JAM2 and other CNS genes in each of the 9 brain tissues will be treated as the target model successively; the details of the target tissues are found in Table \ref{Tiss-test}. That is, we consider 9 target models separately, using the other 31 tissues as the source models in which the total sample size of the source data is 12,386. The covariates in use are the genes that are in the enriched MODULE\_137 (\href{https://www.gsea-msigdb.org/gsea/msigdb/cards/MODULE\_137.html}
 {https://www.gsea-msigdb.org/gsea/msigdb/cards/MODULE\_137.html}) and there are no missing values in all of the 40 tissues. The final covariates include a total of 1089 genes.

 \begin{table}[h]
\caption{The list of 9 brain tissues, their abbreviations and
sample sizes.}\label{Tiss-test}
\begin{tabular*}{\columnwidth}{@{\extracolsep\fill}cllc@{\extracolsep\fill}}
\toprule
{\bf ID} & {\bf Names} & {\bf Abbreviation} & {\bf Sample size} \\
\midrule
1 & Caudate\_basal\_ganglia            & C.B.ganglia    & 194 \\
2 & Cerebellar\_Hemisphere             & C.hemisphere   & 175 \\

3 & Cerebellum                         & Cerebellum     & 209 \\

4 & Cortex                             & Cortex         & 205 \\

5 & Frontal\_Cortex\_BA9               & F.cortex       & 175 \\

6 & Hippocampus                        & Hippocampus    & 165 \\

7 & Hypothalamus                       & Hypothalamus   & 170 \\

8 & Nucleus\_accumbens\_basal\_ganglia & N.A.B. ganglia & 202 \\

9 & Putamen\_basal\_ganglia            & P.B.ganglia    & 170 \\
\bottomrule
\end{tabular*}
\end{table}

 After obtaining the initial estimator  $\tilde\bbeta^{(k)}$ for each source using SCAD, we test the normality of the residuals for each source tissue. From the $p$-values   presented in Supplementary materials, we find that the $p$-values of most of the source tissues are less than 0.05, suggesting that the normality assumption on $\epsilon^{(k)} (k \geq 1)$ is not appropriate. Consequently, we only consider RIW-TL and RIW-TL-U for our methods. The same tuning scheme in Section \ref{sec:simulation} is used here. To obtain a reliable comparison, we randomly split the data in the target tissue into a training set and a testing one with a ratio 7:3 and compute the average prediction errors on the testing set over 100 replications. We examine the relative prediction error (RPE) of a given method over that of LASSO. Clearly, the smaller this error is, the better an approach is. The results are shown in Figure \ref{fig6} together with the sample usage rate of each method.
\begin{figure}[!htbp]
\centering
\includegraphics[scale=0.5]{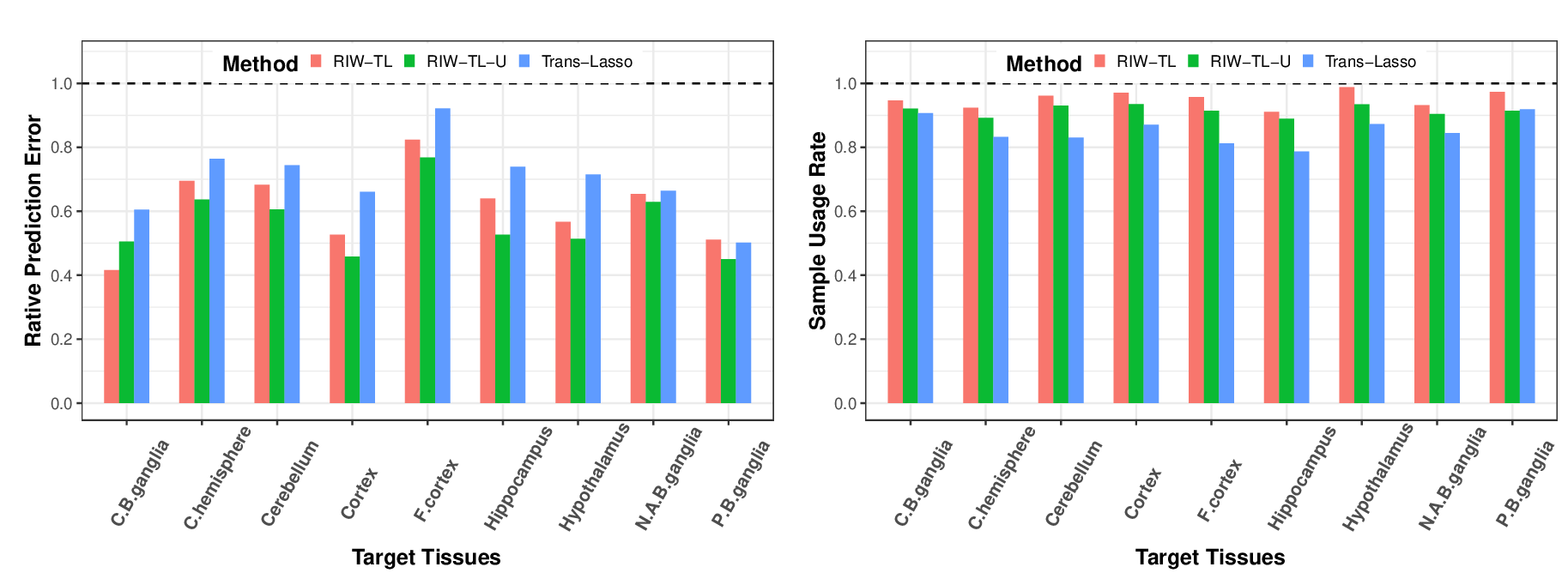}
\caption{Left: Relative prediction errors.
Right: Sample usage rates.
}\label{fig6}
\end{figure}

In Figure \ref{fig6}, it is seen that RIW-TL, RIW-TL-U and Trans-Lasso all perform better than LASSO across different target tissues as all the relative prediction errors are smaller than one. Define the relative gain of a method as the difference between its prediction error and that of LASSO then divide by the latter. We find that the average gain of Trans-Lasso is 29.8\%, while that of RIW-TL is 38.7\% and RIW-TL-U is 43.4\%. On the other hand, it is seen that the sample usage rates of RIW-TL and RIW-TL-U are close to one, higher than those of Trans-Lasso,  suggesting that they use more source data, and hence leading to smaller prediction errors.

Since all the methods are based on LASSO which does variable selection, we compare them in terms of the selected variables in relation to {LASSO estimator $\hat\bbeta_{\rm Lasso} = (\hat\beta_{L,1},\cdots,\hat\beta_{L,p})^{\top}$}. Towards this, for any estimator $\check\bbeta$,
define
$$
{\rm S} = \frac{1}{p} \# \{j: \check \beta_j \neq 0\},~~
{\rm PR} = \frac{\# \{j: \check \beta_j \neq 0, \hat \beta_{L,j} \neq 0\}}{\# \{j: \hat \beta_{L,j} \neq 0\}},~~
{\rm NR} = \frac{\# \{j: \check \beta_j \neq 0, \hat \beta_{L,j} = 0\}}{\# \{j: \hat \beta_{L,j} = 0\}}.
$$
Here S stands for the sparsity rate indicating the proportion of variables selected as nonzero. PR is the positive rate documenting the percentage of nonzero coefficients of the LASSO  estimator also estimated as nonzero by a method. NR is the negative rate, the percentage of zero coefficients of the LASSO estimate that are otherwise estimated as nonzero. The results are summarized in Table \ref{table3} together with the relative prediction error defined previously. We can see that both RIW-TL, RIW-TL-U and Trans-Lasso select more variables than LASSO. This makes sense intuitively since the sample size is boosted in these transfer learning methods by including data in the source data. As a return, the penalized likelihood is able to identify more variables.  Comparing the two RIW-TL methods (RIW-TL and RIW-TL-U) and Trans-Lasso, we find that in general our methods produce less sparse models with larger positive rates and negative rates. Overall, compared to LASSO and Trans-Lasso, the two residual importance weighted methods have smaller relative prediction errors for all the targets, with RIW-TL-U having the smallest errors except for C.B.ganglia.

\begin{table*}[t]
\caption{Further results on variable selection and the relative prediction error (RPE) in data analysis.}\label{table3}
\tabcolsep=0pt
\begin{tabular*}{\textwidth}{@{\extracolsep{\fill}}lcccccccccccc@{\extracolsep{\fill}}}
\toprule%
\textbf{Method} & \textbf{S} & \textbf{PR} & \textbf{NR} &\textbf{RPE}
                & \textbf{S} & \textbf{PR} & \textbf{NR} &\textbf{RPE}
                & \textbf{S} & \textbf{PR} & \textbf{NR} &\textbf{RPE}\\
\midrule
& \multicolumn{4}{@{}c@{}}{C.B.gangli}
& \multicolumn{4}{@{}c@{}}{C.hemisphere}
& \multicolumn{4}{@{}c@{}}{Cerebellu}\\
RIW-TL      & 0.312 & 0.794 & 0.468 & 0.416   & 0.252 & 0.379 & 0.248 & 0.695   & 0.812 & 0.880 & 0.810 & 0.283 \\
RIW-TL-U    & 0.378 & 0.676 & 0.313 & 0.505   & 0.429 & 0.690 & 0.422 & 0.637   & 0.813 & 0.840 & 0.812 & 0.206 \\
Trans-Lasso & 0.144 & 0.365 & 0.140 & 0.605   & 0.121 & 0.138 & 0.121 & 0.764   & 0.124 & 0.321 & 0.119 & 0.344 \\
Lasso       & 0.031 & 1     & 0     & 1       & 0.027 & 1     & 0     & 1       & 0.023 & 1     & 0     & 1     \\
& \multicolumn{4}{@{}c@{}}{Cortex}
& \multicolumn{4}{@{}c@{}}{F.cortex}
& \multicolumn{4}{@{}c@{}}{Hippocampus}\\
RIW-TL      & 0.253 & 0.619 & 0.245 & 0.527   & 0.565 & 0.663 & 0.557 & 0.824   & 0.511 & 0.643 & 0.507 & 0.640 \\
RIW-TL-U    & 0.529 & 0.619 & 0.527 & 0.458   & 0.590 & 0.687 & 0.582 & 0.768   & 0.594 & 0.714 & 0.591 & 0.527 \\
Trans-Lasso & 0.129 & 0.571 & 0.120 & 0.661   & 0.142 & 0.181 & 0.139 & 0.922   & 0.096 & 0.214 & 0.092 & 0.739 \\
Lasso       & 0.019 & 1     & 0     & 1       & 0.076 & 1     & 0     & 1       & 0.026 & 1     & 0     & 1      \\

& \multicolumn{4}{@{}c@{}}{Hypothalamus}
& \multicolumn{4}{@{}c@{}}{N.A.B.ganglia}
& \multicolumn{4}{@{}c@{}}{P.B.gang}\\
RIW-TL      & 0.619 & 0.516 & 0.622 & 0.567   & 0.307 & 0.357 & 0.304 & 0.654   & 0.496 & 0.567 & 0.494 & 0.511 \\
RIW-TL-U    & 0.775 & 0.613 & 0.780 & 0.514   & 0.694 & 0.750 & 0.797 & 0.629   & 0.587 & 0.767 & 0.788 & 0.450 \\
Trans-Lasso & 0.123 & 0.226 & 0.120 & 0.715   & 0.111 & 0.196 & 0.106 & 0.664   & 0.120 & 0.200 & 0.118 & 0.502 \\
Lasso       & 0.028 & 1     & 0     & 1       & 0.051 & 1     & 0     & 1       & 0.028 & 1     & 0     & 1     \\
\bottomrule
\end{tabular*}
\end{table*}

\section{Discussion}
We have proposed a novel transfer learning method named RIW-TL for high-dimensional linear regression, based on weighting residuals by their importance. Our method possesses several attractive features.  First, it uses information adaptively at the individual observation level rather than the source level popular in the literature. We show that the oracle version of RIW-TL gives a better convergence rate than its competitors. Second, the practical version of RIW-TL only requires the estimation of univariate densities thus avoiding the curse of dimensionality needed in naive importance weighting. Third, RIW-TL treats posterior shift and full distribution shift in a unified manner as the needed densities for defining the weights are the conditional distributions of the response given the predictors. Thus, it overcomes the shortcomings of competing approaches in the literature \cite[cf.]{li2021,tian2022} where the predictors across sources should follow more or less a homogenous design. Last, we demonstrate via numerical experiments that RIW-TL achieves better performance than its competitors such as Trans-Lasso, especially in difficult cases where the sources are not close to the target. We have focused on the situation where all $\bbeta^{(k)}, k=0,\cdots, K$, are sparse, but the results are readily applicable when they are approximately so in the sense of \cite{zhang2008sparsity}.

There are many avenues to extend the current framework to other settings. Firstly, we can study this approach for quantile regression to accommodate the heterogeneity and heavy tailedness in the source and target domains. Secondly, thanks to the fact that our importance weights is constructed as a one-dimensional conditional density ratio, a similar analysis may be applied to nonparametric or semi-parametric regression models. Thirdly, we can extend the idea to study generalized linear models where the probability ratio is a natural candidate for importance weighting. Finally, it will be interesting to explore how the ideas developed in this paper can be used for semi-supervised and unsupervised learning. These and other generalizations of RIW-TL will be discussed in future work.

\bibliographystyle{elsarticle-harv}
\bibliography{reference}

\clearpage
{\setstretch{1.5}
\begin{center}
{\Large \bf Supplements to ``Residual Importance Weighted Transfer Learning
For High-dimensional Linear Regression"}
\end{center}
}

This supplementary material contains four parts. The proofs of the propositions, lemmas and theorems where the weights are assumed to be known are presented in Section \ref{S.1}. Section \ref{S.2} contains the proofs where the weights are estimated by kernel density estimation. Section \ref{S.3} includes the proofs of RIW-TL when $f_{\epsilon}$ is from a uniform distribution. Section \ref{S.4} includes additional numerical results.

\setcounter{equation}{0}
\renewcommand\thesection{S.1}
\renewcommand{\theequation}{A.\arabic{equation}}
\section{Proofs in Section \ref{sec:riw}}\label{S.1}
\subsection{Proof of Proposition \ref{prop1}}
 Denote a generic random vector of $\bm z$ in source $k$ as $(\bm x^{(k)}, y^{(k)})$ and, with some abuse of notation,  the conditional distribution of $y$ given $\bm x$ in source $k$ as $f_k(y|{\bm x})$.   Note
  \begin{eqnarray*}
 && \mathbb{E}_{\bm z \sim f_k}
  \left\{
  \omega^{(k)} \cdot
  (y - \bm x^{\top} \bbeta)^2
  \right\}  = \mathbb{E}_{\bm x^{(k)}} \left[\mathbb{E}_{y^{(k)}|\bm x^{(k)}}
  \left\{
  \omega^{(k)} \cdot
  (y^{(k)} - (\bm x^{(k)})^{\top} \bbeta)^2
  \right\}\right]\\
  &=& \mathbb{E}_{\bm x^{(k)}} \left\{ \int (y - (\bm x^{(k)})^{\top} \bbeta)^2 \frac{f_{\epsilon}(y - ({\bm x}^{(k)})^{\top}\bbeta^{(0)})}{f_{\epsilon^{(k)}}(y - ({\bm x}^{(k)})^{\top} \bbeta^{(k)})} f_k(y|{\bm x}^{(k)}) dy \right\}\\
  &=&\mathbb{E}_{\bm x^{(k)}}  \left\{ \int (y - (\bm x^{(k)})^{\top} \bbeta)^2f_{\epsilon}(y - ({\bm x}^{(k)})^{\top}\bbeta^{(0)}) dy\right\}\\
   &=& \mathbb{E}_{\bm x^{(k)}} \left\{ \mathbb{E}_{\epsilon}  (y - (\bm x^{(k)})^{\top} \bbeta)^2\right\},
  \end{eqnarray*}
 where the second equality holds because the conditional distribution of $y$ given $\bm x$ in source $k$ satisfies
 $f_k(y|{\bm x}^{(k)})=f_{\epsilon^{(k)}}(y - ({\bm x}^{(k)})^{\top} \bbeta^{(k)})$ due to the linear model assumption, and in the last equality $\epsilon$ is a random variable satisfying $y=(\bm x^{(k)})^{\top} \bbeta^{(0)}+\epsilon$. Since by assumption $\mathbb{E} (\epsilon)=0$, the expectation $\mathbb{E}_{\epsilon}  (y - (\bm x^{(k)})^{\top} \bbeta)^2$ is minimized at $\bbeta^{(0)}$ for any given $\bm x^{(k)}$ and so is the total expectation above. Proposition \ref{prop1} is proved.

\subsection{Proof of Proposition \ref{prop2}}
 The conclusion $(i)$ is obvious. Recall $\omega_i^{(k)} = f_{\epsilon}(\epsilon_i^{(k)} + \eta_i^{(k)})/f_{\epsilon^{(k)}}(\epsilon_i^{(k)})$ and $f_{\epsilon}(t)=[f_{\epsilon^{(k)}}(t)+f_{\epsilon^{(k)}}(-t)]/2$. By the definition of $\mathcal{I}_k$, that is, $\mathcal{I}_k = \{i: |\epsilon_i^{(k)} + \eta_i^{(k)}| \leq A, |\eta_i^{(k)}| \leq M\}$, the conclusion of $(ii)$ is proved.


\subsection{Proof of Proposition \ref{prop3}}
Denote by $f_{(\epsilon_i^{(k)}, \eta_i^{(k)})}(t_1,t_2)$ the joint density function of $(\epsilon_i^{(k)}, \eta_i^{(k)})$. When $\epsilon_i^{(k)}$ is distributed as $N(0,1)$ and $\bm x_i^{(k)}$ follows $N(\bm 0, \bSig^{(k)})$,  $\eta_i^{(k)}$ is  distributed as $N(0,d_k^2)$ and we have
$$
f_{(\epsilon_i^{(k)}, \eta_i^{(k)})}(t_1,t_2) = \frac{1}{2 \pi d_k} \exp
\left(
- \frac{t_1^2}{2} - \frac{t_2^2}{2 d_k^2}
\right).
$$
Let $X = \epsilon_i^{(k)} + \eta_i^{(k)}$, $Y = \eta_i^{(k)}$ and $g(x,y)$ be the joint density of $(X,Y)$. It is easy to see that
$$
g(x,y) = f_{(\epsilon_i^{(k)}, \eta_i^{(k)})}(x - y,y) = \frac{1}{2 \pi d_k} \exp
\left\{
- \frac{(x - y)^2}{2} - \frac{y^2}{2 d_k^2}
\right\}.
$$
With the notation  $C_{A,M} = \exp(-AM)$,  we have
\begin{eqnarray}\nonumber
\mathbb{P}(|\epsilon_i^{(k)} + \eta_i^{(k)}| \le A,~ |\eta_i^{(k)}| \leq M ) & = &
\frac{1}{2 \pi d_k}
\int_{-A}^A \int_{-M}^M \exp
\left\{
- \frac{(x - y)^2}{2} - \frac{y^2}{2 d_k^2}
\right\} dx dy\\\nonumber
& = &
\frac{1}{2 \pi d_k}
\int_{-A}^A \int_{-M}^M \exp
\left\{
- \frac{d_k^2 x^2 + (d_k^2 + 1)y^2}{2 d_k^2} + xy
\right\} dx dy\\\nonumber
&\geq&
C_{A,M} \frac{1}{2 \pi d_k}
\int_{-A}^A \int_{-M}^M \exp
\left\{
- \frac{(d_k^2 + 1)(x^2 + y^2)}{2 d_k^2}
\right\} dx dy.
\end{eqnarray}
By the polar coordinates transformation, it holds that
\begin{eqnarray}\nonumber
\mathbb{P}(|\epsilon_i^{(k)} + \eta_i^{(k)}| \le A,~ |\eta_i^{(k)}| \leq M )
&\geq&
 C_{A,M} \frac{1}{2 \pi d_k}
\int_{-A}^A \int_{-M}^M \exp
\left\{
- \frac{(d_k^2 + 1)(x^2 + y^2)}{2 d_k^2}
\right\} dx dy\\\nonumber
& \geq &  C_{A,M} \frac{1}{d_k} \int_{0}^{\varphi} \exp\left(- \frac{d_k^2 + 1}{2 d_k^2} r^2 \right) r~ dr\\\nonumber
& = &
C_{A,M} \frac{d_k}{d_k^2 + 1} \left\{
1 - \exp\left(
-\frac{d_k^2 + 1}{2 d_k^2} \varphi^2
\right)
\right\},
\end{eqnarray}
where $\varphi = {\rm min}\{A,M\} < \infty$.
 This completes the  proof of Proposition \ref{prop3}.

\subsection{Proof of Corollary \ref{coro1}}
Recall that $\rho_{\mathcal{I}} = (n_0 + \mathbb{E}(n_{\mathcal{I}}))/(n_0 + \Sigma_{k = 1}^K n_k)$. By the conclusion of Proposition \ref{prop3}, that is, $n_k/d_k \lesssim \mE(n_{\mathcal{I}_k}) \leq n_k$, we complete the proof.

\vspace{5mm}
 For  $0 \le k\le K$, write $\bX^{(k)}=(\bm x_1^{(k)},\cdots, \bm x_{n_k}^{(k)})^\top$ as the design matrix, $\bm Y^{(k)} = (y_1^{(1)},\cdots,y_{n_k}^{(k)})^{\top} \in \mathbb{R}^{n_k}$ as the corresponding response vector, and $\bm \epsilon^{(k)} = (\epsilon_1^{(k)},\cdots,\epsilon_{n_k}^{(k)})^{\top}$ as the error vector. For any nonempty set $\mathcal{B}$, recall $\bI_{\mathcal{B}}$ is a diagonal matrix with its $i$th diagonal element being indicator ${\rm I}(i \in \mathcal{B})$. For easy reference, recall the following conditions.

\subsection{Proof of Theorem \ref{the1}}
It is sufficient to show the conclusion for $K = 1$ where $\mathcal{I} = \mathcal{I}_1$.
 Denote by $\bW = {\rm diag}\{\omega_1^{(1)},\cdots,\omega_{n_1}^{(1)}\}$ the diagonal weights matrix.
Let
$$
\bm Y =
\begin{pmatrix}
 \bm Y^{(0)} \\
\bI_{\mathcal{I}_1} \bW^{1/2} \bm Y^{(1)}
\end{pmatrix} \quad \text{and}
\quad \bX =
\begin{pmatrix}
 \bX^{(0)} \\
\bI_{\mathcal{I}_1} \bW^{1/2} \bX^{(1)}
\end{pmatrix}.
$$
Since $\tilde \bbeta_{ora}^{(0)}$ is the minimizer, it holds that
\begin{equation}\label{f1}
\frac{1}{2(n_0 + n_1)} \|\bm Y - \bX \tilde \bbeta_{ora}^{(0)}\|^2 + \lambda \|\tilde \bbeta_{ora}^{(0)}\|_1
\leq
\frac{1}{2(n_0 + n_1)} \|\bm Y -  \bX \bbeta^{(0)}\|^2 + \lambda \|\bbeta^{(0)}\|_1.
\end{equation}
Denote by $\tilde \bdelta = \tilde \bbeta_{ora}^{(0)} - \bbeta^{(0)}$. By simplifying (\ref{f1}), we obtain the following basic inequality:
\begin{eqnarray}\nonumber
&& \frac{1}{2(n_0 + n_1)} \tilde \bdelta^{\T} [(\bX^{(0)})^{\T} \bX^{(0)} + (\bX^{(1)})^{\T} \bI_{\mathcal{I}_1} \bW \bX^{(1)}] \tilde \bdelta + \lambda \|\tilde \bbeta_{ora}^{(0)}\|_1\\\label{f3}
&\leq&
\frac{1}{n_0 + n_1} [(\bm \epsilon^{(0)})^{\T} \bX^{(0)} + (\bm \epsilon^{(1)} + \bm \eta^{(1)})^{\T} \bI_{\mathcal{I}_1} \bW \bX^{(1)}] \tilde \bdelta + \lambda \|\bbeta^{(0)}\|_1 : =\Lambda,
\end{eqnarray}
\noindent
where $\bm \eta^{(1)} =  \bX^{(1)} (\bbeta^{(1)} - \bbeta^{(0)}) = (\eta_1^{(1)},\cdots,\eta_{n_1}^{(1)})^{\top}$. For $\Lambda$ in (\ref{f3}), it follows that
$$
|\Lambda| \leq \frac{1}{n_0 + n_1}\|(\bm \epsilon^{(0)})^{\T} \bX^{(0)} +(\bm \epsilon^{(1)} + \bm \eta^{(1)})^{\T} \bI_{\mathcal{I}_1} \bW \bX^{(1)}\|_{\infty} \|\tilde \bdelta\|_1 + \lambda \|\bbeta^{(0)}\|_1.
$$
Given $t>0$, we have
\begin{eqnarray}\nonumber
&& {\mathbb{P}}\left( \frac{1}{n_0 + n_1}\|(\bm \epsilon^{(0)})^{\T} \bX^{(0)} + (\bm \epsilon^{(1)} + \bm \eta^{(1)})^{\T} \bI_{\mathcal{I}_1} \bW \bX^{(1)}\|_{\infty} > t \right)\\\nonumber
&\leq& \sum\limits_{j=1}^p \, {\mathbb{P}}\left(
\frac{1}{n_0 + n_1}
\left| \sum\limits_{i=1}^{n_0} \epsilon_i^{(0)} x_{ij}^{(0)} + \sum\limits_{i = 1}^{n_1} (\epsilon_i^{(1)} + \eta_i^{(1)}) \omega_i^{(1)} x_{ij}^{(1)} {\rm I}( i \in \mathcal{I}_1)\right| > t \right)\\\label{f4}
& \leq & p~ \underset{1 \leq j \leq p}{\rm max}\, {\mathbb{P}} \left( \frac{1}{n_0 + n_1}
\left| \sum\limits_{i=1}^{n_0} \epsilon_i^{(0)} x_{ij}^{(0)} + \sum\limits_{i = 1}^{n_1} (\epsilon_i^{(1)} + \eta_i^{(1)}) \omega_i^{(1)} x_{ij}^{(1)} {\rm I}( i \in \mathcal{I}_1)\right| > t \right).
\end{eqnarray}
Recall $\mathcal{I}_1 = \{i \in \mathcal{D}_1: |\epsilon_i^{(1)} + \eta_i^{(1)}| \leq A,~|\eta_i^{(1)}| \leq M\}$. By the symmetry of $f_{\epsilon}$, for any fixed $\eta_i^{(1)}$, we have
\begin{eqnarray}\nonumber
\mathbb{E}[(\epsilon_i^{(1)} + \eta_i^{(1)}) \omega_i^{(1)} {\rm I}(i \in \mathcal{I}_1)]
&=& \int_{\epsilon_i^{(1)} \in \mathcal{I}_1|_{\epsilon_i^{(1)}}} (\epsilon_i^{(1)} + \eta_i^{(1)}) \frac{f_{\epsilon}(\epsilon_i^{(1)} + \eta_i^{(1)})}{f_{\epsilon^{(1)}}(\epsilon_i^{(1)})} f_{\epsilon^{(1)}}(\epsilon_i^{(1)}) d \epsilon_i^{(1)}\\\label{f5}
&=& \int_{-A}^A \xi_i f_{\epsilon}(\xi_i) d \xi_i = 0,
\end{eqnarray}
where $\xi_i = \epsilon_i^{(1)} + \eta_i^{(1)}$ and $\mathcal{I}_1|_{\epsilon_i^{(1)}} = \{\epsilon_i^{(1)}:  |\epsilon_i^{(1)} + \eta_i^{(1)}| \le A,~ |\eta_i^{(1)}| \leq M \}$. By the definition of $\mathcal{I}_1$ and Proposition \ref{prop2}, we have that $(\epsilon_i^{(1)} + \eta_i^{(1)}) \omega_i^{(1)} {\rm I}(i \in \mathcal{I}_1)$ is bounded and consequently is a sub-Gaussian variable, of which the sub-gaussian parameter is denoted as $\sigma_{\omega}$.

For $j = 1,\cdots,p$, without loss of generality, assume that $\Sigma_{i = 1}^{n_0} (x_{ij}^{(0)})^2/n_0 = 1$ and $\Sigma_{i \in \mathcal{I}_1} (x_{ij}^{(1)})^2 / n_{\mathcal{I}_1}= 1$. Since $\epsilon_i^{(0)}$ is distributed as sub-Gaussian with parameter $\kappa$ in Condition 2, and the target and the sources are independent, letting
$$
t = 2 \,{\rm max}\{\kappa, \sigma_{\omega}\} \tilde \rho_{\mathcal{I}_1}
\sqrt{(\tilde t^2 + 2 \,{\rm log}\,p) / (n_0 + n_{\mathcal{I}_1})}~~\text{with}~~
\tilde \rho_{\mathcal{I}_1} = (n_0 + n_{\mathcal{I}_1})/(n_0 + n_1),
$$
for any $\tilde t > 0$, we have
\begin{eqnarray}\nonumber
 && {\mathbb{P}}\left( \frac{1}{n_0 + n_1}\|(\bm \epsilon^{(0)})^{\T} \bX^{(0)} + (\bm \epsilon^{(1)} + \bm \eta^{(1)})^{\T} \bI_{\mathcal{I}_1} \bW \bX^{(1)}\|_{\infty} > t \right)\\\nonumber
 & = &
 {\mathbb{P}}\left( \frac{1}{n_0 + n_{\mathcal{I}_1}}\|(\bm \epsilon^{(0)})^{\T} \bX^{(0)} + (\bm \epsilon^{(1)} + \bm \eta^{(1)})^{\top} \bI_{\mathcal{I}_1} \bW \bX^{(1)}\|_{\infty} >
t / \tilde \rho_{\mathcal{I}_1}
  \right)\\\label{prob-compute}
 & \leq & 2 \, {\rm exp}\left(-\tilde t^2/2 \right).
\end{eqnarray}
Let
$$
\lambda \asymp \tilde \rho_{\mathcal{I}_1} \sqrt{{\rm log}\,p / (n_0 + n_{\mathcal{I}_1})}~~\text{with}~~\tilde \rho_{\mathcal{I}_1} = (n_0 + n_{\mathcal{I}_1})/(n_0 + n_1).
$$
Then it holds in probability that
\begin{equation}\label{lambda-value}
\frac{1}{n_0 + n_1}\|(\bm \epsilon^{(0)})^{\T} \bX^{(0)} + (\bm \epsilon^{(1)} + \bm \eta^{(1)})^{\T} \bI_{\mathcal{I}_1} \bW \bX^{(1)}\|_{\infty} \leq \frac{1}{2} \lambda.
\end{equation}
Feeding the result of \eqref{lambda-value} into (\ref{f3}), we have
\begin{eqnarray}\nonumber
&& \frac{1}{2(n_0 + n_1)} \tilde \bdelta^{\T} [(\bX^{(0)})^{\T} \bX^{(0)} + (\bX^{(1)})^{\T} \bI_{\mathcal{I}_1} \bW \bX^{(1)}] \tilde \bdelta \\\nonumber
&\leq&
 \lambda \|\tilde \bbeta_{ora}^{(0)} - \bbeta^{(0)}\|_1/2 + \lambda \|\bbeta^{(0)}\|_1 - \lambda \|\tilde \bbeta_{ora}^{(0)}\|_1\\\nonumber
&\leq&
 \lambda (\| (\tilde \bbeta_{ora}^{(0)} - \bbeta^{(0)})_{\mathcal{H}_0}\|_1 + \|(\tilde \bbeta_{ora}^{(0)})_{\mathcal{H}_0^c}\|_1)/2 + \lambda(\| (\tilde \bbeta_{ora}^{(0)} - \bbeta^{(0)})_{\mathcal{H}_0}\|_1 - \|(\tilde \bbeta_{ora}^{(0)})_{\mathcal{H}_0^c}\|_1)\\\label{f7}
& = & 3 \lambda \| (\tilde \bbeta_{ora}^{(0)} - \bbeta^{(0)})_{\mathcal{H}_0}\|_1/2 - \lambda \|(\tilde \bbeta_{ora}^{(0)})_{\mathcal{H}_0^c}\|_1/2.
\end{eqnarray}
We see $\tilde \bdelta = \tilde \bbeta_{ora}^{(0)} - \bbeta^{(0)} \in \mathcal{E}(\mathcal{H}_0,3)$. By the restricted eigenvalue Condition 3, it follows that
\begin{eqnarray}\nonumber
&&\frac{1}{n_0 + n_1} \tilde \bdelta^{\T} [(\bX^{(0)})^{\T} \bX^{(0)} + (\bX^{(1)})^{\T} \bI_{\mathcal{I}_1} \bW \bX^{(1)}] \tilde \bdelta \\\nonumber
&\geq&
\frac{1}{n_0 + n_1} [ \tilde \bdelta^{\T} (\bX^{(0)})^{\T} \bX^{(0)}\tilde \bdelta  + c \cdot \tilde \bdelta^{\top} (\bX^{(1)})^{\T} \bI_{\mathcal{I}_1} \bX^{(1)} \tilde \bdelta] \\\nonumber
&\geq&
\frac{1}{n_0 + n_1} (n_0  \phi_1^2 + n_{\mathcal{I}_1} c \phi_2^2)\|\tilde \bdelta\|^2 \\\label{f8}
&\geq&
\frac{n_0 + n_{\mathcal{I}_1}}{n_0 + n_1} \phi^2 \|\tilde \bdelta\|^2 : = \tilde \rho_{\mathcal{I}_1} \phi^2 \|\tilde \bdelta\|^2,
\end{eqnarray}
where $\phi^2 = {\rm min}\{\phi_1^2, c \phi_2^2\}$ and $\omega_i^{(1)} \geq c > 0$ for all $i \in \mathcal{I}_1$ by Proposition \ref{prop2}. Back to (\ref{f7}), by adding $\lambda \|(\tilde \bbeta_{ora}^{(0)} - \bbeta^{(0)})_{\mathcal{H}_0}\|_1/ 2$ to both sides of this inequality, it follows that
\begin{equation}\label{ff9}
\frac{1}{n_0 + n_1} \tilde \bdelta^{\T} [(\bX^{(0)})^{\T} \bX^{(0)} + (\bX^{(1)})^{\T} \bI_{\mathcal{I}_1} \bW \bX^{(1)}] \tilde \bdelta  +  \lambda \|\tilde \bbeta_{ora}^{(0)} - \bbeta^{(0)}\|_1
\leq
 4 \lambda \| (\tilde \bbeta_{ora}^{(0)} - \bbeta^{(0)})_{\mathcal{H}_0}\|_1.
\end{equation}
For  $4\lambda \| (\tilde \bbeta_{ora}^{(0)} - \bbeta^{(0)})_{\mathcal{H}_0}\|_2$ in (\ref{ff9}), it holds that
\begin{eqnarray}\nonumber
4\lambda  \| (\tilde \bbeta_{ora}^{(0)} - \bbeta^{(0)})_{\mathcal{H}_0}\|_1
&\leq&
4\lambda \sqrt{s_0} \| (\tilde \bbeta_{ora}^{(0)} - \bbeta^{(0)})_{\mathcal{H}_0}\|_2\\\nonumber
&\leq&
4\lambda \sqrt{s_0} \|\tilde \bbeta_{ora}^{(0)} - \bbeta^{(0)}\|_2\\\nonumber
&\leq&
\frac{4 \lambda \sqrt{s_0}}{\phi}
\sqrt{\frac{1}{\tilde \rho_{\mathcal{I}_1}}}
\sqrt{\frac{\tilde \bdelta^{\top} [(\bX^{(0)})^{\T} \bX^{(0)} + (\bX^{(1)})^{\T} \bI_{\mathcal{I}_1} \bW \bX^{(1)}] \tilde \bdelta}{n_0 + n_1}}\\\label{f9}
&\leq& \frac{\tilde \bdelta^{\T} [(\bX^{(0)})^{\T} \bX^{(0)} + (\bX^{(1)})^{\top} \bI_{\mathcal{I}_1} \bW \bX^{(1)}] \tilde \bdelta}{2(n_0 + n_1)} +  \frac{8 \lambda^2 s_0}{\tilde \rho_{\mathcal{I}_1} \phi^2},
\end{eqnarray}
where the first inequality follows from the Cauchy-Schwarz inequality, the third uses (\ref{f8}) and the last uses the inequality $4 a b \leq a^2/2 + 8 b^2$. Inserting (\ref{f9}) into (\ref{ff9}), we have
$$
\frac{1}{2} \tilde \rho_{\mathcal{I}_1} \phi^2 \|\tilde \bdelta\|^2 \leq \frac{1}{2(n_0 + n_1)} \tilde \bdelta^{\top} [(\bX^{(0)})^{\top} \bX^{(0)} + (\bX^{(1)})^{\top} \bI_{\mathcal{I}_1} \bW \bX^{(1)}] \tilde \bdelta
 \leq
 \frac{8 \lambda^2 s_0}{\tilde \rho_{\mathcal{I}_1} \phi^2},
$$
where the first inequality uses (\ref{f8}). Then we have
\begin{equation}\label{oracle-rate}
\|\tilde \bdelta\|^2 \leq
\frac{16 \lambda^2 s_0}{\tilde \rho_{\mathcal{I}_1}^2 \phi^4}.
\end{equation}
Recall $\lambda \asymp \tilde \rho_{\mathcal{I}_1} \sqrt{{\rm log}\,p / (n_0 + n_{\mathcal{I}_1})}$ with $\tilde \rho_{\mathcal{I}_1} = (n_0 + n_{\mathcal{I}_1})/(n_0 + n_1)$. Then under the event
\begin{equation}\label{event1}
\{n_{\mathcal{I}_1} \asymp \mathbb{E}(n_{\mathcal{I}_1})\},
\end{equation}
it follows that
$$
\lambda \asymp \rho_{\mathcal{I}_1} \sqrt{{\rm log}\,p / (n_0 + \mathbb{E}(n_{\mathcal{I}_1}))} ~~
\text{with}~~\rho_{\mathcal{I}_1} = (n_0 + \mathbb{E}(n_{\mathcal{I}_1}))/(n_0 + n_1).
$$
Then (\ref{oracle-rate}) can be rewritten as
$$
\|\tilde \bdelta\|^2
 =  O_p \left(
\frac{s_0 {\rm log}\,p}{n_0 + \mathbb{E}(n_{\mathcal{I}_1})}
\right).
$$

It remains to prove that the probability of the event in (\ref{event1}) is tending to 1. Denote by $p^{(1)} = \mathbb{P}(|\epsilon^{(1)} + \eta^{(1)}| \leq A,~|\eta^{(1)}| \leq M)$. It follows that
 $$
 \mE(n_{\mathcal{I}_1}) = n_1 \mE({\rm I}\{i \in \mathcal{I}_1\}) = n_1 p^{(1)}
 $$
 and
 $$
\var(n_{\mathcal{I}_1}) = n_1 [ \mE^2({\rm I}\{i \in \mathcal{I}_1\}) -
 \mE({\rm I}\{i \in \mathcal{I}_1\})^2 ] = n_1 p^{(1)}(1 - p^{(1)}).
 $$
We see that $\mE(n_{\mathcal{I}_1})$ has an order no less than $\sqrt{\var(n_{\mathcal{I}_1})}$. Since $n_{\mathcal{I}_1} -  \mathbb{E}(n_{\mathcal{I}_1}) = O_p \left(\sqrt{\var(n_{\mathcal{I}_1})}\right)$, it holds in probability that
 \begin{equation}\label{nE}
 n_{\mathcal{I}_1} \asymp \mathbb{E}(n_{\mathcal{I}_1}).
\end{equation}
 Thus, the proof of Theorem \ref{the1} is completed.

\renewcommand\thesection{S.2}
\section{Proofs in Section \ref{sec:riw-tl in practice}}\label{S.2}
\renewcommand{\theequation}{B.\arabic{equation}}
\setcounter{equation}{0}
\subsection{Proof of Proposition \ref{prop4}}
We prove the proposition when $k = 1$ and $j = 3$. The other cases can be proved in the same manner.
First, we prove ($i$). Recall
$$
\mathcal{I}_{13}^{-} = \{i \in \mathcal{D}_{13}: |\epsilon_i^{(1)} + \eta_i^{(1)}| \leq A - \alpha_n,~|\eta_i^{(1)}| \leq M - \alpha_n\}
$$
and
\begin{equation}\label{I_13_add}
\hat {\mathcal{I}}_{13} = \{i \in \mathcal{D}_{13}: |\epsilon^{(1)} + \eta_i^{(1)} - r_i^{(1)0}| \leq A,~ |\hat \eta_i^{(1)}| \leq M \},
\end{equation}
where
\begin{equation}\label{alpha_n_1}
\alpha_n = 2 \underset{1\le k\le K}{\rm max} \psi_k \underset{0\le k\le K}{\rm max} \|\tilde \bbeta^{(k)} - \bbeta^{(k)}\|_1~~\text{with}~~\psi_k = \underset{i \in \mathcal{D}_k}{\rm max}~ \|\bm x_i^{(k)}\|_{\infty},
\end{equation}
and $r_i^{(1)0} = (\bm x_i^{(1)})^{\top} (\tilde \bbeta^{(0)} - \bbeta^{(0)})$. For any $i \in \mathcal{I}_{13}^{-}$, it holds that
\begin{eqnarray}\nonumber
 |\epsilon_i^{(1)} + \eta_i^{(1)} - r_i^{(1)0}| & \leq & |\epsilon_i^{(1)} + \eta_i^{(1)}| + |r_i^{(1)0}|\\\nonumber
 & = & |\epsilon_i^{(1)} + \eta_i^{(1)}| + |(\bm x_i^{(1)})^{\top} (\tilde \bbeta^{(0)} - \bbeta^{(0)})|\\\nonumber
 & \leq & |\epsilon_i^{(1)} + \eta_i^{(1)}| + \|\bm x_i^{(1)}\|_{\infty} \|\tilde \bbeta^{(0)} - \bbeta^{(0)}\|_1\\\nonumber
 & \leq & A - \alpha_n + \|\bm x_i^{(1)}\|_{\infty} \|\tilde \bbeta^{(0)} - \bbeta^{(0)}\|_1.
 \end{eqnarray}
 By the definition of $\alpha_n$ in \eqref{alpha_n_1}, it follows that
$$
 \underset{1 \leq i \leq n_1}{\rm max}\|\bm x_i^{(1)}\|_{\infty} \|\tilde \bbeta^{(0)} - \bbeta^{(0)}\|_1 \leq \alpha_n.
$$
Then it follows that, for any $i \in \mathcal{I}_{13}^{-}$,
  $$
  |\epsilon_i^{(1)} + \eta_i^{(1)} - r_i^{(1)0}| \leq
  |\epsilon_i^{(1)} + \eta_i^{(1)}| + \|\bm x_i^{(1)}\|_{\infty} \|\tilde \bbeta^{(0)} - \bbeta^{(0)}\|_1 \leq
  A - \alpha_n + \alpha_n = A.
  $$
Further, for any $i \in \mathcal{I}_{13}^{-}$, it holds that
 \begin{eqnarray}\nonumber
 |\hat \eta_i^{(1)}| &\leq& |\eta_i^{(1)}| + |\hat \eta_i^{(1)} - \eta_i^{(1)}|\\\nonumber
 & = & |\eta_i^{(1)}| + |(\bm x_i^{(1)})^{\top} (\tilde \bbeta^{(1)} - \bbeta^{(1)}) - (\bm x_i^{(1)})^{\top} (\tilde \bbeta^{(0)} - \bbeta^{(0)})|\\\nonumber
  &\leq& |\eta_i^{(1)}|+ \|\bm x_i^{(1)}\|_{\infty} \|\tilde \bbeta^{(1)} - \bbeta^{(1)}\|_1 + \|\bm x_i^{(1)}\|_{\infty} \|\tilde \bbeta^{(0)} - \bbeta^{(0)}\|_1\\\nonumber
  &\leq& M - \alpha_n + \alpha_n = M,
 \end{eqnarray}
 which leads to $\mathcal{I}_{13}^{-} \subseteq \hat {\mathcal{I}}_{13}$. Similarly, by the definition of $\mathcal{I}_{13}^{+}$, we have
 $\hat {\mathcal{I}}_{13} \subseteq \mathcal{I}_{13}^{+}$. Then we have that $\mathcal{I}_{13}^- \subseteq \hat{\mathcal{I}}_{13} \subseteq \mathcal{I}_{13}^{+}$.

 Next, we prove ($ii$). Due to the splitting technique, $n_{\mathcal{I}_{1j}^{-}}$ has the same order for $j = 1,2,3$. Therefore, by the definition of $n_{\mathcal{I}_1^{-}}$, that is, $n_{\mathcal{I}_1^{-}} = \Sigma_{j = 1}^3 n_{\mathcal{I}_{1j}^{-}}$, we have
   $$
   \mE(n_{\mathcal{I}_1^{-}}) \asymp \mE(n_{\mathcal{I}_{13}^-}).
   $$
Without loss of generality, assume that $n_1$ can be evenly split into three parts such that   $\tilde n_1 = n_1 / 3$. Recall the definition of $\mE(n_{\mathcal{I}_1})$ and $\mE(n_{\mathcal{I}_{13}^{-}})$, that is,
 $$
 \mE(n_{\mathcal{I}_1}) = n_1 \mathbb{P}(|\epsilon^{(1)} + \eta^{(1)}| \leq A,~|\eta^{(1)}| \leq M)
 $$
 and
  $$
 \mE(n_{\mathcal{I}_{13}^{-}}) = \tilde n_1 \mathbb{P}(|\epsilon^{(1)} + \eta^{(1)}| \leq A - \alpha_n,~|\eta^{(1)}| \leq M - \alpha_n),
 $$
 where $\tilde n_1 = n_1 /3$.
Since $\alpha_n = o_p(1)$, so we can derive that
 $$
 \mE(n_{\mathcal{I}_{13}^-}) \asymp \mE(n_{\mathcal{I}_1}).
 $$
 Similarly, we can derive that
  $$
  \mE(n_{\mathcal{I}_{13}^+}) \asymp \mE(n_{\mathcal{I}_1}).
  $$
  Finally, by combining the result that $\mathcal{I}_{13}^-\subseteq \hat {\mathcal{I}}_{13} \subseteq \mathcal{I}_{13}^{+}$ and (\ref{nE}), it holds with probability tending to one that
   $$
   n_{\mathcal{I}_1} \asymp \mE(n_{\mathcal{I}_1}) \asymp \mE(n_{\hat{\mathcal{I}}_{13}}) \asymp \mE(n_{\mathcal{I}_{13}^-}).
   $$
 Proposition \ref{prop4} is proved.

\vspace{5mm}
To simplify the proof of Lemma \ref{lemma1}, we give the following Lemma \ref{lemma0}.

\setcounter{lemma}{-1}
\begin{lemma}\label{lemma0}
Suppose that Conditions 1, 5 and 6 are satisfied, for $k = 1,\cdots,K$, it holds that
$$
\underset{t}{\rm sup} \,
|\hat f_{\epsilon^{(k)}}(t) - f_{\epsilon^{(k)}}(t)| = O_p (q_n),
$$
where $q_n = u_n + v_n$ with $u_n = \underset{1 \leq k \leq K}{\rm max} \psi_k \gamma_k / b_k^2$ and $v_n = \underset{1 \leq k \leq K}{\rm max}
\{b_k^2 + ({\rm log}\, n_k /(n_k b_k))^{1/2}\}$ being the error of kernel estimation.
\end{lemma}

\begin{proof}
For $k = 0,1,\cdots,K$, let
$$
\tilde f_{\epsilon^{(k)}}(t) = \frac{1}{n_k b_k} \sum\limits_{i=1}^{n_k} K\left( \frac{t - \epsilon_i^{(k)}}{b_k} \right)
$$
be the traditional kernel density estimator of $f_{\epsilon^{(k)}}(\cdot)$. From \cite{Silverman1978}, we know
\begin{equation}\label{B1}
\underset{t}{\rm sup} |\tilde f_{\epsilon^{(k)}}(t) - f_{\epsilon^{(k)}}(t)| = O_p \left\{ b_k^2 + \left( \frac{{\rm log}\, n_k}{n_k b_k}\right)^{1/2} \right\} = O_p (v_n),
\end{equation}
where $v_n = \underset{1 \leq k \leq K}{\rm max}\{b_k^2 + ({\rm log}\, n_k /(n_k b_k))^{1/2}\}$.
Further, by Taylor's expansion, for some $(\epsilon_i^{(k)})^*$ between $\epsilon_i^{(k)}$ and $\hat \epsilon_i^{(k)}$, it holds that
\begin{eqnarray}\nonumber
\underset{t}{\rm sup} |\hat f_{\epsilon^{(k)}}(t) - \tilde f_{\epsilon^{(k)}}(t)|
& = & \underset{t}{\rm sup} \frac{1}{n_k b_k} \left| \sum\limits_{i=1}^{n_k} K \left(\frac{t - \hat \epsilon_i^{(k)}}{b_k}\right) - \sum\limits_{i=1}^{n_k} K \left(\frac{t - \epsilon_i^{(k)}}{b_k}\right)
\right|\\\nonumber
&\leq&  \underset{t}{\rm sup} \frac{1}{n_k b_k^2} \sum\limits_{i=1}^{n_k} K^{\prime} \left(\frac{t - (\epsilon_i^{(k)})^*}{b_k}\right) |\hat \epsilon_i^{(k)} -  \epsilon_i^{(k)}|\\\nonumber
&=& O_p\left( \underset{i \in \mathcal{D}_k}{\rm max}~|\hat \epsilon_i^{(k)} -  \epsilon_i^{(k)}|/b_k^2 \right),
\end{eqnarray}
where $K^{\prime}(\cdot)$ denotes the first derivative of the kernel $K(\cdot)$ and is bounded according to Condition 6. By Condition 5, it holds in probability that
$$
\underset{i \in \mathcal{D}_{kj}}{\rm max}~|\hat \epsilon_i^{(k)} -  \epsilon_i^{(k)}| \leq
\psi_k \gamma_k,
$$
where $\psi_k$ is defined in \eqref{alpha_n_1}.
Further, by Condition 6, where $b_k$ satisfies $b_k = o(1)$ and $b_k^2 \gg \sqrt{{\rm log}\,(n_k p)} \gamma_k$, we have
\begin{equation}\label{B2}
\underset{t}{\rm sup} |\hat f_{\epsilon^{(k)}}(t) - \tilde f_{\epsilon^{(k)}}(t)| = O_p \left(
\psi_k \gamma_k / b_k^2
\right) = O_p (u_n)
\end{equation}
with $u_n = \underset{1 \leq k \leq K}{\rm max} \psi_k \gamma_k / b_k^2$.
Combining (\ref{B1}) and (\ref{B2}) can be found
$$
\underset{t}{\rm sup} \,
|\hat f_{\epsilon^{(k)}}(t) - f_{\epsilon^{(k)}}(t)| \leq
\underset{t}{\rm sup} |\hat f_{\epsilon^{(k)}}(t) - \tilde f_{\epsilon^{(k)}}(t)| +
\underset{t}{\rm sup} |\tilde f_{\epsilon^{(k)}}(t) - f_{\epsilon^{(k)}}(t)|
 =  O_p (q_n),
$$
where $q_n = u_n + v_n$.  Lemma \ref{lemma0} is proved.
\end{proof}

\subsection{Proof of Lemma 1}
It suffices to prove the lemma \ref{lemma1} when $K=1$. For $i = 1,\cdots,n_1$, let
$\tilde \omega_i^{(1)} = f_{\epsilon}(\epsilon_i^{(1)} + \eta_i^{(1)} - r_i^{(1)0})/f_{\epsilon}(\epsilon_i^{(1)})$.
It holds that

\begin{eqnarray}\nonumber
|\hat \omega_i^{(1)} / \omega_i^{(1)} - 1|
& = &
\left|
\frac{\hat \omega_i^{(1)}}{\tilde \omega_i^{(1)}} \cdot
\frac{\tilde \omega_i^{(1)}}{\omega_i^{(1)}} - 1
\right|\\\nonumber
&\leq&
\left|
\frac{\hat \omega_i^{(1)}}{\tilde \omega_i^{(1)}}   - 1
\right| \cdot
\left|
\frac{\tilde \omega_i^{(1)}}{\omega_i^{(1)}} - 1
\right| + \underbrace{\left|
\frac{\hat \omega_i^{(1)}}{\tilde \omega_i^{(1)}}   - 1
\right|}_{\Lambda_{i,1}} + \underbrace{\left|
\frac{\tilde \omega_i^{(1)}}{\omega_i^{(1)}} - 1
\right|}_{\Lambda_{i,2}} \\\label{B3}
&:=&
\Lambda_{i,1} \Lambda_{i,2} + \Lambda_{i,1} + \Lambda_{i,2}.
\end{eqnarray}

We  upper bound (\ref{B3}) in two steps.

{\it \textbf{Step~1}.}~ Denote by $r_i^{(1)0} = (\bm x_i^{(1)})^{\top} (\tilde \bbeta^{(0)} - \bbeta^{(0)})$ for $i = 1,\ldots,n_1$. For $\Lambda_{i,1}$, we have
\begin{eqnarray}\nonumber
\Lambda_{i,1}=|\hat \omega_i^{(1)} / \tilde \omega_i^{(1)} - 1|
& = &
\left|
\frac{\hat f_{\epsilon}(\epsilon_i^{(1)} + \eta_i^{(1)} - r_i^{(1)0} )}{f_{\epsilon}(\epsilon_i^{(1)} + \eta_i^{(1)} - r_i^{(1)0})} \cdot
\frac{f_{\epsilon^{(1)}}(\epsilon_i^{(1)})}{\hat  f_{\epsilon^{(1)}}(\hat \epsilon_i^{(1)})} - 1
\right|\\\nonumber
&\leq&
\left|
\frac{\hat f_{\epsilon}(\epsilon_i^{(1)} + \eta_i^{(1)} - r_i^{(1)0} )}{f_{\epsilon}(\epsilon_i^{(1)} + \eta_i^{(1)} - r_i^{(1)0})}  - 1
\right| \cdot
\left|
\frac{f_{\epsilon^{(1)}}(\epsilon_i^{(1)})}{\hat  f_{\epsilon^{(1)}}(\hat \epsilon_i^{(1)})} - 1
\right| \\\nonumber
&+& \underbrace{\left|
\frac{\hat f_{\epsilon}(\epsilon_i^{(1)} + \eta_i^{(1)} - r_i^{(1)0} )}{f_{\epsilon}(\epsilon_i^{(1)} + \eta_i^{(1)} - r_i^{(1)0})}  - 1
\right|}_{\Lambda_{i,11}} + \underbrace{\left|
\frac{f_{\epsilon^{(1)}}(\epsilon_i^{(1)})}{\hat  f_{\epsilon^{(1)}}(\hat \epsilon_i^{(1)})} - 1
\right|}_{\Lambda_{i,12}} \\\label{B33}
&:=&
\Lambda_{i,11} \Lambda_{i,12} + \Lambda_{i,11} + \Lambda_{i,12}.
\end{eqnarray}
For $\Lambda_{i,11}$, by Lemma \ref{lemma0} and the form of $f_{\epsilon}$, it holds that
$$
\underset{t}{\rm sup} \,
|\hat f_{\epsilon}(t) - f_{\epsilon}(t)| = O_p(q_n),
$$
where $q_n$ is defined in Lemma \ref{lemma0}. Then, combining  by Proposition \ref{prop2} and Condition 5, for any $i \in \hat {\mathcal{I}}_{1j}$, it holds that
 $$
 {\rm min}\{f_{\epsilon}(\epsilon_i^{(1)} + \eta_i^{(1)} - r_i^{(1)0}), f_{\epsilon}(\epsilon_i^{(1)} + \eta_i^{(1)}), \hat f_{\epsilon^{(1)}}(\epsilon_i^{(1)}), \hat f_{\epsilon^{(1)}}(\hat \epsilon_i^{(1)}) \}  \geq \tau_l > 0
 $$
for some constants $\tau_l$ with probability tending to one.
 Then
\begin{eqnarray}\nonumber
\Lambda_{i,11}=
\left|
\frac{\hat f_{\epsilon}(\epsilon_i^{(1)} + \eta_i^{(1)} - r_i^{(1)0} )}{f_{\epsilon}(\epsilon_i^{(1)} + \eta_i^{(1)} - r_i^{(1)0})}  - 1
\right| & = &
\left|
\frac{
\hat f_{\epsilon}(\epsilon_i^{(1)} + \eta_i^{(1)} - r_i^{(1)0} ) - f_{\epsilon}(\epsilon_i^{(1)} + \eta_i^{(1)} - r_i^{(1)0})}{f_{\epsilon}(\epsilon_i^{(1)} + \eta_i^{(1)} - r_i^{(1)0})}
\right|\\\nonumber
& \leq &
\frac{
|\hat f_{\epsilon}(\epsilon_i^{(1)} + \eta_i^{(1)} - r_i^{(1)0} ) - f_{\epsilon}(\epsilon_i^{(1)} + \eta_i^{(1)} - r_i^{(1)0})|}{\tau_l} \\\label{B4}
&=& O_p (q_n).
\end{eqnarray}

 For $\Lambda_{i,12}$, it holds that
\begin{eqnarray}\nonumber
\Lambda_{i,12} = \left|
\frac{f_{\epsilon^{(1)}}(\epsilon_i^{(1)})}{\hat  f_{\epsilon^{(1)}}(\hat \epsilon_i^{(1)})} - 1
\right|
&=&
\left|
\frac{f_{\epsilon^{(1)}}(\epsilon_i^{(1)})}{\hat f_{\epsilon^{(1)}}(\epsilon_i^{(1)})} \cdot
\frac{\hat f_{\epsilon^{(1)}}(\epsilon_i^{(1)})}{\hat  f_{\epsilon^{(1)}}(\hat \epsilon_i^{(1)})} - 1
\right|\\\nonumber
&\leq&
\left|
\frac{f_{\epsilon^{(1)}}(\epsilon_i^{(1)})}{\hat f_{\epsilon^{(1)}}(\epsilon_i^{(1)})}  - 1
\right| \cdot
\left|
\frac{\hat f_{\epsilon^{(1)}}(\epsilon_i^{(1)})}{\hat  f_{\epsilon^{(1)}}(\hat \epsilon_i^{(1)})} - 1
\right| \\\nonumber
 &+& \underbrace{\left|
\frac{f_{\epsilon^{(1)}}(\epsilon_i^{(1)})}{\hat f_{\epsilon^{(1)}}(\epsilon_i^{(1)})} - 1
\right|}_{\Lambda_{i,121}} + \underbrace{\left|
\frac{\hat f_{\epsilon^{(1)}}(\epsilon_i^{(1)})}{\hat  f_{\epsilon^{(1)}}(\hat \epsilon_i^{(1)})} - 1
\right|}_{\Lambda_{i,122}} \\\label{B5}
&:=&
\Lambda_{i,121} \Lambda_{i,122} + \Lambda_{i,121} + \Lambda_{i,122}.
\end{eqnarray}
For $\Lambda_{i,121}$, by  Lemma \ref{lemma0}, we have
\begin{equation}\label{B66}
\Lambda_{i,121} = \left|
\frac{\hat f_{\epsilon^{(1)}}(\epsilon_i^{(1)}) - f_{\epsilon^{(1)}}(\epsilon_i^{(1)})}{\hat f_{\epsilon^{(1)}}(\epsilon_i^{(1)})}
\right|  \leq
\frac{|\hat f_{\epsilon^{(1)}}(\epsilon_i^{(1)}) - f_{\epsilon^{(1)}}(\epsilon_i^{(1)})|}{\tau_l} = O_p(q_n).
\end{equation}
For $\Lambda_{i,122}$, we have
\begin{eqnarray}\nonumber
\Lambda_{i,122} &=& \left|
\frac{\hat  f_{\epsilon^{(1)}}(\hat \epsilon_i^{(1)}) - \hat f_{\epsilon^{(1)}}(\epsilon_i^{(1)})}{\hat  f_{\epsilon^{(1)}}(\hat \epsilon_i^{(1)})}
\right|\\\nonumber
& \leq &
\frac{3}{\tau_l n_1 b_1} \left|
\sum\limits_{j \in \mathcal{D}_{12}} K\left(\frac{\hat \epsilon_i^{(1)} - \hat\epsilon^{(1)}_j}{b_1} \right) -
\sum\limits_{j \in \mathcal{D}_{12}} K\left(\frac{\epsilon_i^{(1)} - \hat\epsilon^{(1)}_j}{b_1} \right)
\right|\\\nonumber
& \leq &
\frac{3}{\tau_l n_1 b_1^2} \sum\limits_{j \in \mathcal{D}_{12}} \left| K^{\prime} \left(\frac{ (\epsilon_i^{(1)})^* - \hat\epsilon^{(1)}_j }{b_1}\right) \right| |\hat \epsilon_i^{(1)} - \epsilon_i^{(1)}| \\\label{B77}
 &=& O_p\left( \underset{i \in \mathcal{D}_{13}}{\rm max}~|\hat \epsilon_i^{(1)} -  \epsilon_i^{(1)}|/b_1^2 \right) = O_p(u_n),
\end{eqnarray}
where $(\epsilon_i^{(1)})^*$ is between $\epsilon_i^{(1)}$ and $\hat \epsilon_i^{(1)}$ and $u_n$ is defined in Lemma \ref{lemma0}. Inserting (\ref{B66}) and (\ref{B77}) into (\ref{B5}), and combining with (\ref{B4}), for each $j = 1,2,3$, we have
\begin{equation}\label{Lambda_11}
\Lambda_{i,1}=\max_{i \in \hat {\mathcal{I}}_{1j}} |\hat \omega_i^{(1)}/ \tilde \omega_i^{(1)} - 1| = O_p(q_n).
\end{equation}

{\it \textbf{Step~2}.}~For $\Lambda_{i,2}$, by Conditions 1 and 5, we have
\begin{eqnarray}\nonumber
\Lambda_{i,2} = |\tilde \omega_i^{(1)}/ {\omega_i^{(1)}}  - 1|
& = &
\left|
\frac{f_{\epsilon}(\epsilon_i^{(1)} + \eta_i^{(1)} - r_i^{(1)0} ) - f_{\epsilon}(\epsilon_i^{(1)} + \eta_i^{(1)})}{f_{\epsilon}(\epsilon_i^{(1)} + \eta_i^{(1)})}
\right|\\\nonumber
&\leq&
\left|
\frac{f_{\epsilon}(\epsilon_i^{(1)} + \eta_i^{(1)} - r_i^{(1)0} ) - f_{\epsilon}(\epsilon_i^{(1)} + \eta_i^{(1)})}{\tau_l}
\right|\\\nonumber
& \leq & |f^{'}(\xi_i)||r_i^{(1)0}| /\tau_l\\\label{Lambda_22}
& \leq &  |f^{'}(\xi_i)| \|\bm x_i^{(1)}\|_{\infty} \|\tilde \bbeta^{(0)} - \bbeta^{(0)}\|_1 = O_p(\alpha_n),
\end{eqnarray}
where $\alpha_n$ is defined in \eqref{alpha_n_1} and $\xi_i$ is between the $\epsilon_i^{(1)} + \eta_i^{(1)} - r_i^{(1)0}$ and $\epsilon_i^{(1)} + \eta_i^{(1)}$. Inserting the result of (\ref{Lambda_11}) and (\ref{Lambda_22}) into (\ref{B3}), we have
$$
\underset{i \in \hat {\mathcal{I}}_{1j}}{\rm max}|\hat \omega_i^{(1)}/ \omega_i^{(1)} - 1| = O_p(q_n + \alpha_n).
$$
Lemma \ref{lemma1} is proved.

\vspace{5mm}
We introduce some notations for the following proofs. Let $\hat {\mathcal{I}} = \cup_{k=1}^K \hat {\mathcal{I}}_k$.  Recall  that $f_{\epsilon}(t)=(f_{\epsilon^{(k)}}(t)+f_{\epsilon^{(k)}}(-t))/2$ and that  $\epsilon$ is  from $f_{\epsilon}$ satisfying $\mathbb{E}(\epsilon)=0$. Then  we see that $\sigma_\epsilon^2 := \mE_{f_{\epsilon}} (X^2)$, standing for  the variance of $\epsilon$,   is bounded. Define $\mu_{\rm max} = 3 \Sigma_{k=1}^K \Sigma_{i \in \mathcal{D}_{k3}} \|\bm x_i^{(k)}\|_{\infty} / \Sigma_{k = 1}^K n_k$, which is the empirical version of $\Sigma_{k = 1}^K \pi_k \mE(\|\bm x^{(k)}\|_\infty)$ with $\pi_k = n_k / \Sigma_{k = 1}^K n_k$.

\subsection{Proof of Theorem 2}
It is sufficient to show the conclusion for $K=1$ where $\hat{\mathcal{I}}=\hat{\mathcal{I}}_1$.
Following from the definition that $\hat \bbeta^{(0)} = (\hat \bbeta_1^{(0)} + \hat \bbeta_2^{(0)} + \hat \bbeta_3^{(0)})/3$, we  prove only the convergence rate of $\hat \bbeta_1^{(0)}$ and the others are the  same.

With the notations $\hat {\mathcal{I}}_{03} = \mathcal{D}_{03}$, $\hat \omega_i^{(0)} = 1$ for all $i \in \hat {\mathcal{I}}_{03}$,  and $\tilde n_k = n_k / 3$ for $k = 0,1$. The estimator $\hat \bbeta_1^{(0)} $ is obtained by solving the following optimization problem:
$$
\hat \bbeta_1 = \underset{\bm\beta \in \mathbb{R}^p}{\rm argmin}~
 \frac{1}{\tilde n_0 + \tilde n_1}
\left\{
\sum\limits_{k = 0,1}
  \sum\limits_{i \in \mathcal{D}_{k3}}
   \left(y_i^{(k)} - (\bm x_i^{(k)})^\top \bbeta\right)^2 \hat \omega_i^{(k)} {\rm I}(i \in \hat {\mathcal{I}}_{k3})
   \right\}\notag + \lambda \|\hat \bbeta\|_1,
$$
where
\begin{align}
    \hat {\mathcal{I}}_{13} &=\{i \in \mathcal{D}_{13}:|\epsilon_i^{(1)} + \eta_i^{(1)} - r_i^{(1)0}|\le A,~ |\hat \eta_i^{(1)}| \leq M \}\notag
\end{align}
  with $r_i^{(1)0} = (\bm x_i^{(1)})^{\top}(\tilde \bbeta^{(0)} - \bbeta^{(0)})$.

  Let $\check {\mathcal{D}} = \mathcal{D}_{11} \cup \mathcal{D}_{12} \cup \{\bm x_i^{(1)}, i\in \mathcal{D}_{13}\}$. We establish the properties  by  conditioning on data $\check {\mathcal{D}}$.

{\it \textbf{Step~1}.}~We show some quantities involved are independent and bounded variables.
Recall that conditional on $\check {\mathcal{D}}$,   $\hat\eta_i^{(1)}$'s,  $\eta_i^{(1)}$'s and $r_i^{(1)0} $'s are  constants, implying
    that  ${\rm I}\{i \in \hat {\mathcal{I}}_{13}\}$ is a function of variable $\epsilon_i^{(1)}$;   consequently, ${\rm I}\{i \in \hat {\mathcal{I}}_{13}\}$ for different $(\bm x_i^{(1)},y_i^{(1)})$  are    independent  variables.
Moreover,   from the definition of $\hat\omega_i^{(1)}$, it follows that  $\hat\omega_i^{(1)}$ for different $i\in \mathcal{D}_{13}$ are independent  variables after conditioning on $\check {\mathcal{D}}$.

 Next, we show that ${\rm I}\{i \in \hat {\mathcal{I}}_{13}\} \hat \omega_i^{(1)}$ are bounded for all $i \in \mathcal{D}_{13}$. From  the proof of Lemma \ref{lemma1}, we see that $\max_{i\in \hat {\mathcal{I}}_{13}} |\hat\omega_i^{(1)}/\omega_i^{(1)}-1 |<\theta_0$ for some $\theta_0$ being  sufficiently small except on a small subset $\check {\mathcal{D}}^{1}$ of $\check {\mathcal{D}}$, satisfying $\mathbb{P}(\check {\mathcal{D}}^{1})\to 0$, as $\min\{n_0,n_1\}\to \infty$.
  By Condition 5, $\tilde {\bm\beta}_k, k=0,1$, are functions of $\check {\mathcal{D}}$ such that $\alpha_n = o_p(1)$, implying that  $\alpha_n$ is a function of $\check {\mathcal{D}}$. Then for this $\theta_0$, it holds that  $\alpha_n \le \theta_0$ except on a subset  $\check {\mathcal{D}}^{2}$ of $\check {\mathcal{D}}$ satisfying  $\mathbb{P}(\check {\mathcal{D}}^{2})\to 0$.

 Note that the following  three terms
   $$
   \max\limits_{i\in \mathcal{D}_{13}}|r_i^{(1)0}|, ~~\max_{i\in \mathcal{D}_{13}}|\hat\eta_i^{(1)}-\eta_i^{(1)}|,~~~ \max\limits_{i\in \mathcal{D}_{13}}|\hat\epsilon_i^{(1)}-\epsilon_i^{(1)}|,
   $$
   are all bounded by $\alpha_n$. Therefore,  these terms are  smaller than $\theta_0$ on the set $\check {\mathcal{D}}^c = \check {\mathcal{D}} \setminus (\check {\mathcal{D}}^1 \cup \check {\mathcal{D}}^2)$.
   Hence on the  set $\check {\mathcal{D}}^c$, we see that both $|\epsilon_i^{(1)}+\eta_i^{(1)}|$'s and $|\epsilon_i^{(1)}|$'s are bounded for $i\in \hat {\mathcal{I}}_{13}$, implying that ${\rm I}\{i \in \hat {\mathcal{I}}_{13}\} \omega_i^{(1)}$ and consequently   ${\rm I}\{i \in \hat {\mathcal{I}}_{13}\} \hat \omega_i^{(1)}$ are bounded.  In summary, on the set $\check {\mathcal{D}}^c$,  ${\rm I}\{i \in \hat {\mathcal{I}}_{13}\} \hat \omega_i^{(1)}$ for all $i\in \mathcal{D}_{13}$  are bounded and independent variables. Further, from the argument above, it is easy to see that ${\rm I}(i\in \hat {\mathcal{I}}_{13})\eta_i^{(1)}$ are bounded for $i\in \mathcal{D}_{13}$.

 {\it \textbf{Step~2.}}~
 Let
$$
\tilde {\bm Y} =
\begin{pmatrix}
 \bm Y^{(0)} \\
\bI_{\hat {\mathcal{I}}_{13}} \hat \bW^{1/2} \bm Y^{(1)}
\end{pmatrix} \quad \text{and}
\quad
\tilde \bX =
\begin{pmatrix}
 \bX^{(0)} \\
\bI_{\hat {\mathcal{I}}_{13}} \hat \bW^{1/2} \bX^{(1)}
\end{pmatrix},
$$
where $\bm Y^{(1)}=(y_i^{(1)}, i\in \mathcal{D}_{13})^{\top} \in \mathbb{R}^{|\mathcal{D}_{13}|}$ and $\bX^{(k)}=((\bm x_i^{(1)})^{\top}, i\in \mathcal{D}_{13})^{\T}\in\mathbb{R}^{|\mathcal{D}_{13}| \times p}$ for $k=0,1$, $\bI_{\hat {\mathcal{I}}_{13}} = \mathrm{diag}\{{\rm I}(i\in \hat{\mathcal{I}}_{13}), i\in \mathcal{D}_{13}\}\in \mathbb{R}^{|\mathcal{D}_{13}| \times |\mathcal{D}_{13}|}$ the diagonal matrix with the $i$th element being ${\rm I}(i\in \hat{\mathcal{I}}_{13})$ for $i\in \mathcal{D}_{13}$, $\hat \bW = {\rm diag}\{ \hat\omega_i^{(1)},~ i \in \mathcal{D}_{13}\} \in \mathbb{R}^{|\mathcal{D}_{13}| \times |\mathcal{D}_{13}|}$ being the diagonal matrix defined in the same manner.

From the fact that $\hat \bbeta_1^{(0)}$ is the minimizer, it holds that
\begin{equation}\label{B6}
\frac{1}{2(\tilde n_0 + \tilde n_1)} \|\tilde {\bm Y} - \tilde \bX \hat \bbeta_1^{(0)}\|^2 + \lambda \|\hat \bbeta_1^{(0)}\|_1
\leq
\frac{1}{2(\tilde n_0 + \tilde n_1)} \|\tilde {\bm Y} - \tilde \bX \bbeta^{(0)}\|^2  + \lambda \|\bbeta^{(0)}\|_1.
\end{equation}
 Denote by $\hat \bdelta = \hat \bbeta_1^{(0)} - \bbeta^{(0)}$. Simplifying (\ref{B6}) leads to the following basic inequality:
 \begin{eqnarray}\nonumber
&& \frac{\hat \bdelta^{\T} (\bX^{(0)})^{\T} \bX^{(0)} \hat \bdelta}{2(\tilde n_0 + \tilde n_1)}  + \frac{\hat \bdelta^{\T} (\bX^{(1)})^{\T} \bI_{\hat {\mathcal{I}}_{13}} \bW \bX^{(1)} \hat \bdelta}{2(\tilde n_0 + \tilde n_1)} + \lambda \|\hat \bbeta_1^{(0)}\|_1\\\nonumber
&\leq&
\frac{1}{\tilde n_0 + \tilde n_1}
\underbrace{\left[(\bm \epsilon^{(0)})^{\T} \bX^{(0)}   +  (\bm \epsilon^{(1)} + \bm \eta^{(1)})^{\T} \bI_{\hat {\mathcal{I}}_{13}} \hat\bW  \bX^{(1)}\right] }_{\Lambda_1} (\hat \bbeta_1^{(0)} - \bbeta^{(0)}) + \lambda \|\bbeta^{(0)}\|_1
\\\label{f14}
&-&
\underbrace{
\frac{\hat \bdelta^{\T} (\bX^{(1)})^{\T} \bI_{\hat {\mathcal{I}}_{13}} \bW(\bW^{-1} \hat \bW - \bI_{n_1}) \bX^{(1)} \hat \bdelta}{2(\tilde n_0 + \tilde n_1)}
}_{\Lambda_2}.
\end{eqnarray}

{\it \textbf{Step~2.1}.}~
We analyze the term $\|\Lambda_1\|_\infty$ on the set $\check {\mathcal{D}}^c$. Denote
$\check{\bm\epsilon}^{(1)}=(\check{\epsilon}_i^{(1)}, i\in \mathcal{D}_{13})^{\top}$ with $ \check{\epsilon}_i^{(1)}=(\epsilon_i^{(1)}+\eta_i^{(1)}) {\rm I}(i\in \hat{\mathcal{I}}_{13})\hat\omega_i^{(1)}$.
From Step 1, we see that $\check{\epsilon}_i^{(1)}$'s are independent and bounded variables on  $\check {\mathcal{D}}^c$, and consequently are sub-Gaussian variables. It follows that
$$(\bm \epsilon^{(1)} + \bm \eta^{(1)})^{\top} \bI_{\hat {\mathcal{I}}_{13}} \hat\bW  \bX^{(1)}=(\check{\bm\epsilon}^{(1)})^{\top}\bX^{(1)}= [\check{\bm\epsilon}^{(1)}-\mathbb{E}(\check{\bm\epsilon}^{(1)})]^{\top}\bX^{(1)} + [\mathbb{E}(\check{\bm\epsilon}^{(1)})]^{\top}\bX^{(1)}.$$
Then
 \begin{equation}\label{the2-lambda12}
 \Lambda_1 = \underbrace{
 (\bm \epsilon^{(0)})^{\top} \bX^{(0)} + [\check{\bm\epsilon}^{(1)}-\mathbb{E}(\check{\bm\epsilon}^{(1)})]^{\top}\bX^{(1)}
 }_{\Lambda_{11}} + \underbrace{
 [\mathbb{E}(\check{\bm\epsilon}^{(1)})]^{\top}\bX^{(1)}
 }_{\Lambda_{12}} :=\Lambda_{11} +\Lambda_{12}.
 \end{equation}
Define the two events:
 $$
    \mathcal{A}_1=\left\{\frac{1}{\tilde n_0+\tilde n_1}\|\Lambda_{11}\|_\infty \leq
    \lambda_n^{(1)}\right\}, ~~~~  \mathcal{A}_2=\left\{ \frac{1}{\tilde n_0+\tilde n_1}\|\Lambda_{12}\|_\infty\le  \lambda_n^{(2)}\right\},
$$
 where
 $$
 \lambda_n^{(1)}= C_1 \tilde \rho_{ \hat{\mathcal{I}}_{13}} \sqrt{\log p / (\tilde n_0 + n_{\hat {\mathcal{I}}_{13}})},~~~~
 \lambda_{n}^{(2)}= C_2
 [\mu_{\max} (q_n + \alpha_n) + C_{f_{\epsilon},A}  \gamma_0],
 $$
with $\tilde \rho_{\hat{\mathcal{I}}_{13}} = (\tilde n_0 + n_{\hat {\mathcal{I}}_{13}})/(\tilde n_0 + \tilde n_1)$ and $C_1, C_2$ being some positive constants. On the event $\mathcal{A}_1 \cap \mathcal{A}_2$, it holds that
$$
\frac{1}{\tilde n_0 + \tilde n_1} \|\Lambda_1\|_{\infty} \leq \frac{1}{n_0 + \tilde n_1} \|\Lambda_{11}\|_{\infty} + \frac{1}{\tilde n_0 + \tilde n_1} \|\Lambda_{12}\|_{\infty} \leq
\lambda_n^{(1)} + \lambda_n^{(2)}.
$$

{\it \textbf{Step~2.2.}}~
We prove the final conclusion on the event  $\mathcal{A}_1 \cap \mathcal{A}_2 \cap \check {\mathcal{D}}^c \cap \mathcal{T}$, where
$\mathcal{T} = \{\hat\omega_i^{(1)} \geq c_1, i \in \hat{\mathcal{I}}_{13}\}$ for some positive constants $c_1$.
On the set $\check {\mathcal{D}}^c$, we found that $\underset{i\in \hat {\mathcal{I}}_{13}}{\rm max} |\hat\omega_i^{(1)}/\omega_i^{(1)}-1 |\leq \theta_0$ for some $\theta_0$ being sufficiently small in Step 1. Then it holds that
$$
|\Lambda_2|  \le  \theta_0\frac{\hat \bdelta^{\T} (\bX^{(1)})^{\T} \bI_{\hat {\mathcal{I}}_{13}} \bW \bX^{(1)} \hat \bdelta}{2(\tilde n_0 + \tilde n_1)}.
$$
By Step 2.1, on the event $\mathcal{A}_1\cap\mathcal{A}_2$, we have $\|\Lambda_1\|_{\infty} \leq (\tilde n_0 + \tilde n_1) (\lambda_n^{(1)} + \lambda_n^{(2)})$. Let $\lambda = 2(\lambda_n^{(1)} + \lambda_n^{(2)})$. Then from \eqref{f14} we have
\begin{eqnarray}\nonumber
&& \frac{1}{2(\tilde n_0 + \tilde n_1)} \hat \bdelta^{\T} (\bX^{(0)})^{\T} \bX^{(0)} \hat \bdelta  + \frac{1}{2(\tilde n_0 + \tilde n_1)}(1-\theta_0) \hat \bdelta^{\T} (\bX^{(1)})^{\T} \bI_{\hat {\mathcal{I}}_{13}} \bW \bX^{(1)} \hat \bdelta + \lambda \|\hat \bbeta_1^{(0)}\|_1\\\nonumber
&\leq&
\frac{1}{\tilde n_0 + \tilde n_1}\|\Lambda_1\|_{\infty} \|\hat \bdelta\|_1   + \lambda \|\bbeta^{(0)}\|_1\\\nonumber
&\le& \frac{1}{2} \lambda \|\hat \bdelta\|_1 + \lambda \|\bbeta^{(0)}\|_1.
\end{eqnarray}
Similar to (\ref{f7}), on the event $\mathcal{A}_1\cap \mathcal{A}_2\cap \check {\mathcal{D}}^c$, the following inequality holds
\begin{eqnarray}\nonumber
&& \frac{1}{2(\tilde n_0 + \tilde n_1)} \hat \bdelta^{\top} [(\bX^{(0)})^{\T} \bX^{(0)} + (1-\theta_0)(\bX^{(1)})^{\T} \bI_{\hat {\mathcal{I}}_{13}} \bW \bX^{(1)}] \hat \bdelta + \frac{1}{2} \lambda \|(\hat \bbeta_1^{(0)} - \bbeta^{(0)})_{\mathcal{H}_0^c}\|_1\\\nonumber
&\leq&
\frac{3}{2} \lambda \| (\hat \bbeta_1^{(0)} - \bbeta^{(0)})_{\mathcal{H}_0}\|_1,
\end{eqnarray}
which implies that $\hat \bdelta = \hat \bbeta_1^{(0)} - \bbeta^{(0)} \in \mathcal{E}(\mathcal{H}_0,3)$. By adding $\lambda \|(\hat \bbeta_1^{(0)} - \bbeta^{(0)})_{\mathcal{H}_0}\|_1/ 2$ to both sides of this inequality, we have
\begin{eqnarray}\nonumber
&& \frac{1}{\tilde n_0 + \tilde n_1} \hat\bdelta^{\T} [(\bX^{(0)})^{\T} \bX^{(0)} +(1-\theta_0)(\bX^{(1)})^{\T} \bI_{\hat {\mathcal{I}}_{13}} \bW \bX^{(1)}] \hat \bdelta + \lambda \|\hat \bbeta_1^{(0)} - \bbeta^{(0)}\|_1\\\label{B13}
&\leq&
4 \lambda \| (\hat \bbeta_1^{(0)} - \bbeta^{(0)})_{\mathcal{H}_0}\|_1.
\end{eqnarray}
 By Condition 4, on the event $\mathcal{A}_1\cap \mathcal{A}_2\cap \check {\mathcal{D}}^c \cap \mathcal{T}$, the following holds
\begin{eqnarray}\nonumber
&&\frac{1}{\tilde n_0 + \tilde n_1} \hat \bdelta^{\T} [(\bX^{(0)})^{\T} \bX^{(0)} + (1 - \theta_0)(\bX^{(1)})^{\T} \bI_{\hat {\mathcal{I}}_{13}} \hat \bW \bX^{(1)}] \tilde \bdelta \\\nonumber
&\geq&
\frac{1}{\tilde n_0 + \tilde n_1} [ \hat \bdelta^{\T} (\bX^{(0)})^{\T} \bX^{(0)}\hat \bdelta  + c_1 (1 - \theta_0) \cdot \hat \bdelta^{\top} (\bX^{(1)})^{\T} \bI_{\hat {\mathcal{I}}_{13}} \bX^{(1)} \hat \bdelta] \\\nonumber
&\geq&
\frac{1}{\tilde n_0 + \tilde n_1} [ \hat \bdelta^{\T} (\bX^{(0)})^{\T} \bX^{(0)}\hat \bdelta  + c_1 (1 - \theta_0) \cdot \hat \bdelta^{\top} (\bX^{(1)})^{\T} \bI_{\mathcal{I}_{13}^{-}} \bX^{(1)} \hat \bdelta] \\\nonumber
&\geq&
\frac{1}{\tilde n_0 + \tilde n_1} \cdot [n_0  \phi_1^2 + n_{\mathcal{I}_{13}^{-}} c_1 (1 - \theta_0) \phi_2^2]\|\hat \bdelta\|^2 \\\label{B133}
&\geq&
\frac{\tilde n_0 + n_{\mathcal{I}_{13}^{-}}}{n_0 + \tilde n_1} \phi^2 \|\hat \bdelta\|^2,
\end{eqnarray}
where $\phi^2 = {\rm min}\{\phi_1^2, c_1 (1 - \theta_0) \phi_2^2\}$. For the term $4\lambda \| (\hat \bbeta_1^{(0)} - \bbeta^{(0)})_{\mathcal{H}_0}\|_2$ in (\ref{B13}), it holds that
\begin{eqnarray}\nonumber
&& 4\lambda  \| (\hat \bbeta_1^{(0)} - \bbeta^{(0)})_{\mathcal{H}_0}\|_1\\\nonumber
&\leq&
4\lambda \sqrt{s_0} \| (\hat \bbeta_1^{(0)} - \bbeta^{(0)})_{\mathcal{H}_0}\|_2\\\nonumber
&\leq&
4\lambda \sqrt{s_0} \|\hat \bbeta_1^{(0)} - \bbeta^{(0)}\|_2\\\nonumber
&\leq&
\frac{4 \lambda \sqrt{s_0}}{\phi} \sqrt{\frac{\tilde n_0 + \tilde n_1}{\tilde n_0 + n_{\mathcal{I}_{13}^{-}}}} \sqrt{\frac{\hat \bdelta^{\T} [(\bX^{(0)})^{\T} \bX^{(0)} + (1 - \theta_0)(\bX^{(1)})^{\T} \bI_{\hat {\mathcal{I}}_{13}} \bW \bX^{(1)}] \hat \bdelta}{\tilde n_0 + \tilde n_1}}\\\label{f9-2}
&\leq& \frac{\hat \bdelta^{\T} [(\bX^{(0)})^{\T} \bX^{(0)} + (\bX^{(1)})^{\T} \bI_{\hat {\mathcal{I}}_{13}} \bW \bX^{(1)}] \hat\bdelta}{2(\tilde n_0 + \tilde n_1)} + \frac{\tilde n_0 + \tilde n_1}{\tilde n_0 + n_{\mathcal{I}_{13}^{-}}} \frac{8 \lambda^2 s_0}{\phi^2} ,
\end{eqnarray}
where the first inequality follows from the Cauchy-Schwarz inequality, the third uses (\ref{B133}) and the last uses the inequality $4 a b \leq a^2/2 + 8 b^2$. By inserting (\ref{f9-2}) into (\ref{B13}), we have
$$
\frac{1}{2} \frac{\tilde n_0 + n_{\mathcal{I}_{13}^{-}}}{\tilde n_0 + \tilde n_1} \phi^2 \|\hat \bdelta\|^2 \leq \frac{1}{2(\tilde n_0 + \tilde n_1)} \hat \bdelta^{\T} [(\bX^{(0)})^{\T} \bX^{(0)} + (1 - \theta_0) (\bX^{(1)})^{\T} \bI_{\hat {\mathcal{I}}_{13}} \bW \bX^{(1)}] \hat \bdelta
 \leq
\frac{\tilde n_0 + \tilde n_1}{\tilde n_0 + n_{\mathcal{I}_{13}^{-}}} \frac{8 \lambda^2 s_0}{\phi^2},
$$
where the first inequality uses (\ref{B133}). Then it holds in probability that
\begin{equation}\label{rate22}
\|\hat \bdelta\|^2 \leq
\left(\frac{\tilde n_0 + \tilde n_1}{\tilde n_0 + n_{\mathcal{I}_{13}^{-}}} \right)^2
\frac{16 \lambda^2 s_0}{\phi^4}.
\end{equation}
Recall that $\lambda = 2 (\lambda_n^{(1)} + \lambda_n^{(2)})$, where
 $$
 \lambda_n^{(1)}= C_1 \tilde \rho_{\hat{\mathcal{I}}_{13}} \sqrt{\log p / (\tilde n_0 + n_{\hat {\mathcal{I}}_{13}})},~~~~
\lambda_{n}^{(2)}= C_2
 [\mu_{\max} (\alpha_n + q_n) + C_{f_{\epsilon},A}  \gamma_0]
 $$
with $\tilde \rho_{\hat{\mathcal{I}}_{13}} = (\tilde n_0 + n_{\hat {\mathcal{I}}_{13}})/(\tilde n_0 + \tilde n_1)$. In fact, by the Proposition \ref{prop4}, where it holds in probability that
$$
n_{\hat {\mathcal{I}}_{13}} \asymp \mathbb{E}(n_{\hat {\mathcal{I}}_{13}}) \asymp \mathbb{E}(n_{\hat {\mathcal{I}}_1}) \asymp \mathbb{E}(n_{\mathcal{I}_1}),
$$
and consequently that
$$
 \lambda_n^{(1)} \asymp  \rho_{\mathcal{I}_1} \sqrt{\log p / (n_0 + \mE(n_{\mathcal{I}_1}))}
 $$
with $\rho_{\mathcal{I}_1} = (n_0 + \mE(n_{\mathcal{I}_1}))/(n_0 + n_1)$. Then (\ref{rate22}) can be rewritten as
 $$
\|\hat \bbeta_1^{(0)} - \bbeta^{(0)}\|_2^2
 = O_p \left\{
 \left( \frac{n_0 + n_1}{n_0 + n_{\mathcal{I}_1^{-}}} \right)^2
s_0 (\lambda_n^{(1)} + \lambda_n^{(2)})^2
\right\}.
$$
Similar to (\ref{nE}), by Proposition \ref{prop4}, it holds in probability that
 $$
 n_{\mathcal{I}_1^{-}} \asymp \mathbb{E}(n_{\mathcal{I}_1^{-}}) \asymp \mathbb{E}(n_{\mathcal{I}_1}).
 $$
If $n_k \rightarrow \infty (k = 0,1)$, it holds  $\lambda_n^{(i)} = o(1),  i = 1,2$, and then we obtain
$$
\|\hat \bbeta_1^{(0)} - \bbeta^{(0)}\|_2^2
 = O_p \left\{
 \frac{s_0 {\rm log}\,p}{n_0 + \mathbb{E}(n_{\mathcal{I}_1})} +
\rho_{\mathcal{I}_1}^{-2}
s_0 [\mu_{\rm max} (q_n + \alpha_n) + C_{f_{\epsilon},A} \gamma_0]^2
\right\}.
$$
Similarly, we obtain the same convergence rate for $\hat \bbeta_2^{(0)}$ and $\hat \bbeta_3^{(0)}$, which leads to the conclusion desired.

{\it \textbf{Step~2.3.}}~
It remains to prove that the probabilities of events $\mathcal{A}_i$'s $(i = 1,2)$ and $\mathcal{T}$ are tending to 1. Recall
 $$
 \Lambda_1 = \underbrace{
 (\bm \epsilon^{(0)})^{\top} \bX^{(0)} + [\check{\bm\epsilon}^{(1)}-\mathbb{E}(\check{\bm\epsilon}^{(1)})]^{\top}\bX^{(1)}
 }_{\Lambda_{11}} + \underbrace{
 [\mathbb{E}(\check{\bm\epsilon}^{(1)})]^{\top}\bX^{(1)}
 }_{\Lambda_{12}} :=\Lambda_{11} +\Lambda_{12}.
 $$
The two events are
 $$
    \mathcal{A}_1=\left\{\frac{1}{\tilde n_0+\tilde n_1}\|\Lambda_{11}\|_\infty \leq  \lambda_n^{(1)}\right\}, ~~~~  \mathcal{A}_2=\left\{ \frac{1}{\tilde n_0+\tilde n_1}\|\Lambda_{12}\|_\infty\le  \lambda_n^{(2)}\right\},
$$
 where
 $$
 \lambda_n^{(1)}= C_1 \tilde \rho_{\hat{\mathcal{I}}_{13}} \sqrt{\log p / (\tilde n_0 + n_{\hat {\mathcal{I}}_{13}})},~~~~
 \lambda_{n}^{(2)}= C_2
 [\mu_{\max} (\alpha_n + q_n) + C_{f_{\epsilon},A} \gamma_0]
 $$
with $\tilde \rho_{\hat{\mathcal{I}}_{13}} = (\tilde n_0 + n_{\hat {\mathcal{I}}_{13}})/(\tilde n_0 + \tilde n_1)$ and $C_1, C_2$ are some constants.

 We firstly consider the event $\mathcal{A}_{1}$. From Step 1, we see that $\check \epsilon_i^{(1)}$'s are independent and bounded variables, and consequently are sub-Gaussian variables. Let $t = C_1 \tilde \rho_{ \hat{\mathcal{I}}_{13}} \sqrt{{\rm log}\,p /(\tilde n_0+ n_{\hat {\mathcal{I}}_{13}})}$. Since $\bm\epsilon^{(0)}$ and $\check{\bm\epsilon}^{(1)}-\mathbb{E}(\check{\bm\epsilon}^{(1)})$ are sub-Gaussian with zero mean, similar to \cite{Bickel2009}, we have
\begin{equation}\label{Lambda_1}
 {\mathbb{P}}\left(\frac{1}{\tilde n_0 + \tilde n_1}\|\Lambda_{11} \|_\infty  \leq t \right) =
  {\mathbb{P}}\left(\frac{1}{\tilde n_0 + n_{\hat {\mathcal{I}}_{13}}}\|\Lambda_{11} \|_\infty  \leq t / \tilde \rho_{ \hat{\mathcal{I}}_{13}} \right)
  > 1 -  p^{-c}\to 1
\end{equation}
for some constant $c>0$. Therefore, we prove that the probability of event $\mathcal{A}_1$  is tending to 1.

Next, we consider the event $\mathcal{A}_2$. It holds that
\begin{equation}\label{bias}
\frac{1}{\tilde n_0 + \tilde n_1} \|\Lambda_{12}\|_{\infty}
= \frac{1}{\tilde n_0 + \tilde n_1}
\left\|\sum\limits_{i\in \mathcal{D}_{13} } \bm x_i^{(1)} \mathbb{E} (\check{\epsilon}_i^{(1)})
\right\|_{\infty}
\leq \frac{1}{\tilde n_1}
\left\|\sum\limits_{i\in \mathcal{D}_{13} } \bm x_i^{(1)} \mathbb{E} (\check{\epsilon}_i^{(1)})\right\|_{\infty}.
\end{equation}
Moreover,
\begin{eqnarray}\nonumber
 \mathbb{E} (\check{\epsilon}_i^{(1)})&=& \int {\rm I}( i \in \hat{\mathcal{I}}_{13})  (\epsilon_i^{(1)} + \eta_i^{(1)}) \hat\omega_i^{(1)} f_{\epsilon^{(1)}}(\epsilon_i^{(1)}) d \epsilon_i^{(1)}\\\nonumber 
& = & \underbrace{
\int_{\epsilon_i^{(1)} \in \hat {\mathcal{I}}_{13}|_{\epsilon_i^{(1)}}}  (\epsilon_i^{(1)} + \eta_i^{(1)}) \omega_i^{(1)} f_{\epsilon^{(1)}}(\epsilon_i^{(1)})d \epsilon_i^{(1)}
}_{E_{i1}}\\\nonumber
& +&
\underbrace{
\int_{\epsilon_i^{(1)} \in \hat {\mathcal{I}}_{13}|_{\epsilon_i^{(1)}}}  (\epsilon_i^{(1)} + \eta_i^{(1)}) (\hat\omega_i^{(1)} - \omega_i^{(1)}) f_{\epsilon^{(1)}}(\epsilon_i^{(1)})d \epsilon_i^{(1)}
}_{E_{i2}} \\\nonumber
&:=&E_{i1}+E_{i2},
\end{eqnarray}
where $\hat {\mathcal{I}}_{13}|_{\epsilon_i^{(1)}}$ denotes the interval of $\epsilon_i^{(1)}$ induced from  the set $\hat {\mathcal{I}}_{13}$.

For $E_{i1}$, by the definition of $\hat {\mathcal{I}}_{13}$ and the symmetry of $f_{\epsilon}$,  it follows that
\begin{equation}\label{E-i1-U}
E_{i1} = \int_{\epsilon_i^{(1)} \in \hat {\mathcal{I}}_{13}|_{\epsilon_i^{(1)}}}  (\epsilon_i^{(1)} + \eta_i^{(1)}) f_{\epsilon}(\epsilon_i^{(1)} + \eta_i^{(1)}) d \epsilon_i^{(1)} =  \int_{A - r_i^{(1)0}}^{A + r_i^{(1)0}} \xi_i f_{\epsilon}(\xi_i) d \xi_i = 2 a_i f_{\epsilon}(a_i) r_i^{(1)0},
\end{equation}
where $\xi_i = \epsilon_i^{(1)} + \eta_i^{(1)}$ and $a_i\in[A - \theta_0,A + \theta_0]$ by recalling $\max\limits_{i\in \mathcal{D}_{13}} r_i^{(1)0}<\theta_0$ with $\theta_0$ being sufficiently small. Recall $C_{f_{\epsilon},A} = \max\limits_{a\in [A - \theta_0,A + \theta_0] }a f_{\epsilon}(a)$ and $\hat \bSig^{(1)} = \sum\limits_{i \in \mathcal{D}_{13}} \bm x_i^{(1)} (\bm x_i^{(1)})^{\top}/\tilde n_1  := (\hat \sigma_{ij}^{(1)})_{p \times p}$. It follows that
\begin{eqnarray}\nonumber
\left\|\frac{1}{\tilde n_1}\sum\limits_{i\in \mathcal{D}_{13} } \bm x_i^{(1)} E_{i1}\right\|_\infty  &\le&  2C_{f_{\epsilon},A}  \|\hat \bSig^{(1)} (\tilde \bbeta^{(0)} - \bbeta^{(0)})\|_{\infty}\\\nonumber
& \leq &  2C_{f_{\epsilon},A}  \|(\hat \bSig^{(1)}  - \bSig^{(1)})(\tilde \bbeta^{(0)} - \bbeta^{(0)})\|_{\infty} + 2C_{f_{\epsilon},A}
\|\bSig^{(1)} (\tilde \bbeta^{(0)} - \bbeta^{(0)})\|_{\infty}\\\nonumber
& \leq &   2C_{f_{\epsilon},A}
\left(
\underset{1 \leq i,j \leq p}{\rm max}|\hat \sigma_{ij}^{(1)} - \sigma_{ij}^{(1)}| + \underset{1 \leq i,j \leq p}{\rm max} |\sigma_{ij}^{(1)}|
\right) \|\tilde \bbeta^{(0)} - \bbeta^{(0)}\|_1\\\label{E-i1}
&\le & C
C_{f_{\epsilon},A} \gamma_0
\end{eqnarray}
for some constant $C>0$, where the last step uses the properties of the sample covariance matrix  that $\underset{1 \leq i,j \leq p}{\rm max}|\hat \sigma_{ij}^{(1)} - \sigma_{ij}^{(1)}| = o_p(1)$ \citep{Bickel2008} and the fact that $\sigma^{(1)}_{\max}=\max_{ij}\sigma_{ij}^{(1)}$ is bounded in Condition 7.

For $E_{i2}$, it follows that
\begin{align}
E_{i2} &= \int_{\epsilon_i^{(1)} \in \hat {\mathcal{I}}_{13}|_{\epsilon_i^{(1)}}}  (\epsilon_i^{(1)} + \eta_i^{(1)})  \omega_i^{(1)}f_{\epsilon^{(1)}}(\epsilon_i^{(1)}) (\hat\omega_i^{(1)} / \omega_i^{(1)}-1) d \epsilon_i^{(1)}\notag\\
&\le \int_{\epsilon_i^{(1)} \in \hat {\mathcal{I}}_{13}|_{\epsilon_i^{(1)}}} | (\epsilon_i^{(1)} + \eta_i^{(1)})  \omega_i^{(1)}f_{\epsilon^{(1)}}(\epsilon_i^{(1)}) | d \epsilon_i^{(1)}
\cdot \max_{i\in \hat{\mathcal{I}}_{13}}|\hat\omega_i^{(1)}/\omega_i^{(1)}-1|\notag\notag\\\nonumber
&=\int_{A - r_i^{(1)0}}^{A + r_i^{(1)0}} |\xi_i| f_{\epsilon}(\xi_i) d \xi_i \cdot\max_{i\in \hat{\mathcal{I}}_{13}}|\hat\omega_i^{(1)} / \omega_i^{(1)}-1| \le \sigma_\epsilon^2  \max_{i\in \hat{\mathcal{I}}_{13}}|\hat\omega_i^{(1)} / \omega_i^{(1)}-1|,
\end{align}
where $\sigma_\epsilon^2 = \mE_{f_{\epsilon}} (X^2)$ is bounded. When $K = 1$, with the notation $\mu_{\rm max} = \sum\limits_{i \in \mathcal{D}_{13}} \|\bm x_i^{(1)}\|_{\infty} / {\tilde n_1}$, it follows that
\begin{eqnarray}\nonumber
\left\|\frac{1}{\tilde n_1}\sum\limits_{i\in \mathcal{D}_{13}  } \bm x_i^{(1)} E_{i2}\right\|_\infty
&\le& \frac{1}{\tilde n_1} \sigma_\epsilon^2 \sum\limits_{i\in \mathcal{D}_{13} }\|\bm x_i^{(1)}\|_{\infty}  \max_{i\in \hat{\mathcal{I}}_{13}}|\hat\omega_i^{(1)}/\omega_i^{(1)}-1| \\\label{E-i2}
& =&  \sigma_{\epsilon}^2 \mu_{\max}  \max_{i\in \hat{\mathcal{I}}_{13}}|\hat\omega_i^{(1)} / \omega_i^{(1)}-1| \leq c \sigma_\epsilon^2 \mu_{\rm max} (\alpha_n + q_n),
\end{eqnarray}
 where the last inequality is due to Lemma \ref{lemma1} that $\max_{i\in \hat{\mathcal{I}}_{13}}|\hat\omega_i^{(1)}/\omega_i^{(1)}-1|\le c (q_n + \alpha_n)$ holds with overwhelming probability for some constant $c>0$.

 By inserting (\ref{E-i1}) and (\ref{E-i2}) into (\ref{bias}), we have
 \begin{eqnarray}\nonumber
 \frac{1}{\tilde n_0 + \tilde n_1} \|\Lambda_{12}\|_{\infty}
 &\leq&
 \frac{1}{\tilde n_1} \left\|\sum\limits_{i\in \mathcal{D}_{13} } \bm x_i^{(1)} \mathbb{E} (\check{\epsilon}_i^{(1)})\right\|_{\infty}\\\nonumber
 &\leq&
 \left\|\frac{1}{\tilde n_1}\sum\limits_{i\in \mathcal{D}_{13} } \bm x_i^{(1)} E_{i1}\right\|_\infty +
\left\|\frac{1}{\tilde n_1}\sum\limits_{i\in \mathcal{D}_{13} } \bm x_i^{(1)} E_{i2}\right\|_\infty\\\nonumber
&\leq &
c \sigma_\epsilon^2 \mu_{\rm max} (\alpha_n + q_n) + C C_{f_{\epsilon},A} \gamma_0.
 \end{eqnarray}
 Therefore, the probability of the event $\mathcal{A}_2$ is tending to 1.

 Finally, we consider the event $\mathcal{T} = \{\hat\omega_i^{(1)} \geq c_1, i \in \hat{\mathcal{I}}_{13}\}$ for some constants $c_1 > 0$. By Lemma \ref{lemma1}, for any $i \in \hat {\mathcal{I}}_{13}$, it holds in probability that
 \begin{equation}\label{weighthat-bound}
 (1 + \theta_0)\omega_i^{(1)} \leq \hat \omega_i^{(1)} \leq (1 + \theta_0) \omega_i^{(1)}
 \end{equation}
 for some $\theta_0$ being  sufficiently small. By Proposition \ref{prop2}, for all $i \in \mathcal{I}_{13}^{+}$ defined in \eqref{I_13_add}, we derive that $\omega_i^{(1)}$ is bounded away from 0 and $\infty$, and so is for all $i \in \hat{\mathcal{I}}_{13}$ since $\hat {\mathcal{I}}_{13} \subseteq \mathcal{I}_{13}^{+}$ in Proposition \ref{prop4}. Consequently, by \eqref{weighthat-bound}, we have $\omega_i^{(k)} \geq c_1$ for all $i \in \hat{\mathcal{I}}_{13}$ with some constants $c_1 > 0$.
 In summary, the probabilities of events $\mathcal{A}_i$'s $(i = 1,2)$ and $\mathcal{T}$ happening tend to 1.

\subsection{Proof of Corollary 2}
Recall that $\alpha_n$, depending  on  $\{\gamma_k, k\ge 0\}$,  is small only when all $n_k$'s (i.e. $k\ge 0$) are large. On the other hand,  $q_n$ that depends   only on  sources is small when $n_k(k\ge 1)$ are large. Hence, conditions  ($ii$) and ($iii$) imply that $q_n=O_p(\alpha_n)$. Moreover, by $\max\limits_{k\ge 1}\gamma_k\ll \gamma_0$ in condition ($ii$), the normality of $\bm x_i^{(k)}$ and the definition of $\alpha_n$,  we have
 $$\alpha_n \lesssim \max\limits_{0\le k\le K}   \gamma_k\max\limits_{1\le k\le K}\sqrt{{\rm log}\,(n_k p)}=O_p\left(\sqrt{\log (n_1p)} \gamma_0\right).$$
 Recall  that  $\mu_{\max}=O_p(\sqrt{\log p})$ in the discussions  just above the Theorem \ref{the2} for Gaussian predictors.
 Moreover, under the conditions of Corollary \ref{coro1} and the condition ($iv$), we see that $\rho_{\mathcal{I}}\asymp1$.
  Combining together, for the second term of  the error rate in Theorem \ref{the2}, we have
 $$
 \rho_{\mathcal{I}}^{-2}
s_0 [\mu_{\rm max} (q_n + \alpha_n) +  C_{f_{\epsilon},A}  \gamma_0]^2 \asymp s_0 \gamma_0^2 \log(p) \log (n_1 p) \asymp \gamma_0^2 \log(p) \log (n_1 p),
$$
where we use the fact that $s_0 \asymp 1$.
This completes the proof. $ \square$

\renewcommand\thesection{S.3}
\section{Proofs in Section \ref{sec:riw-tl-u}}\label{S.3}
\renewcommand{\theequation}{C.\arabic{equation}}
\setcounter{equation}{0}
\subsection{Proof of Proposition 5}
By the independence of $\epsilon_i^{(k)}$ and $\bm x_i^{(k)}$, it holds that
$\mE(n_{\mathcal{I}'_k}) = n_k \mathbb{P}(|\epsilon_i^{(k)}| \leq A) \mathbb{P}(|\eta_i^{(k)}| \leq M)$. When $\epsilon_i^{(k)}$ is distributed as $N(0,1)$, it is easy to see that $\mathbb{P}(|\epsilon_i^{(k)}| \leq A) = 2 \Phi(A) - 1$, where $\Phi(\cdot)$ is the probability distribution function of the standard normal distribution.

When $\bm x_i^{(k)}$ follows $N(\bm 0, \bSig^{(k)})$, it holds that $\eta_i^{(k)} \sim N(0,d_k^2)$, where $d_k^2 = (\bdelta^{(k)})^{\top} \bSig^{(k)} \bdelta^{(k)}$. Then
\begin{eqnarray}\nonumber
\mathbb{P}^2(|\eta_i^{(k)}| \leq M) &=& \frac{1}{2 \pi d_k^2} \int_{-M}^M \int_{-M}^M {\rm exp}\{-(x^2 + y^2)/(2 d_k^2)\} dxdy\\\nonumber
&\geq& \frac{1}{2 \pi d_k^2} \int_0^{2 \pi}d\theta \int_0^M r \cdot {\rm exp}\{-r^2/(2 d_k^2)\} dr\\\nonumber
& = & 1 - {\rm exp}\{-M^2 / (2 d_k^2)\}.
\end{eqnarray}
Consequently, we have
$$
\mE(n_{\mathcal{I}'_k}) \geq n_k (2 \Phi(A) - 1) \sqrt{1 - {\rm exp}\{-M^2 / (2 d_k^2)\}}.
$$

Since both $A$ and $M$ are bounded, the following conclusions hold: ($i$) when $d_k^2 = O(1)$, it holds $\mE(n_{\mathcal{I}_k^{\prime}}) \asymp n_k$; ($ii$) when $d_k$ diverges as $ n_k\to\infty$, it follows that
 $n_k/d_k  \lesssim \mE(n_{\mathcal{I}_k^{\prime}}) \leq n_k$. In summary, we get the conclusion, that is,
 $\mE(n_{\mathcal{I}'_k}) \gtrsim n_k/\max\{d_k,1\}$.

\subsection{Proof of Theorem 3}
The proof of Theorem \ref{the3} is similar to Theorem \ref{the1} except that here we need to bound $\|\Lambda_1\|_{\infty}$ with
$$
\Lambda_1 = \frac{1}{n_0 + n_1} \left\{(\bm \epsilon^{(0)})^{\T} \bX^{(0)} + (\bm \epsilon^{(1)} + \bm \eta^{(1)})^{\top} \bI_{\mathcal{I}'_1} \bW_T  \bX^{(1)}
\right\},
$$
where $\bW_T = {\rm diag}\{\omega_{1,T}^{(1)},\cdots,\omega_{n_1,T}^{(1)}\}$  denotes  the diagonal weights matrix and  $\mathcal{I}'_1 = \{i \in \mathcal{D}_1: |\epsilon_i^{(1)}| \leq A,~|\eta_i^{(1)}| \leq M\}$.
The population version of $\Lambda_1$ is $\mathbb{E}[\epsilon_i^{(0)} x_{ij}^{(0)} + (\epsilon_i^{(1)} + \eta_i^{(1)}) x_{ij}^{(1)} \omega_{i,T}^{(1)} {\rm I}(i \in \mathcal{I}'_1)]$ that   is nonzero here and is the counterpart of \eqref{f5} in Theorem \ref{the1}.

 The proof is similar to Step 2.1 of the  proof of Theorem \ref{the2}.
 Let $\check {\bm \epsilon}^{(1)} = (\check \epsilon_i^{(1)}, i \in \mathcal{D}_1)^{\top}$ with $\check \epsilon_i^{(1)} = (\epsilon_i^{(1)} + \eta_i^{(1)}) \omega_{i,T}^{(1)} {\rm I}(i \in \mathcal{I}'_1)$.  It follows that
$$
\Lambda_1 = \underbrace{
\frac{1}{n_0 + n_1}
\left\{
(\bm \epsilon^{(0)})^{\top} \bX^{(0)} + (\check{\bm\epsilon}^{(1)}-\mathbb{E}(\check{\bm\epsilon}^{(1)}))^{\top}\bX^{(1)}
\right\}
}_{\Lambda_{11}} +
\underbrace{
\frac{1}{n_0 + n_1}
[\mathbb{E}(\check{\bm\epsilon}^{(1)})]^{\top}\bX^{(1)}
}_{\Lambda_{12}}.
$$

For $\Lambda_{11}$, we had shown that $\check \epsilon_i^{(1)}$'s are independent and bounded variables, and consequently are sub-Gaussian variables. Let
$$
t = C_1 \tilde \rho_{\mathcal{I}_1^{\prime}} \sqrt{{\rm log}\,p /(n_0+ n_{ \mathcal{I}_1^{\prime}})}~~\text{with}~~\tilde \rho_{\mathcal{I}_1^{\prime}} = (n_0 + n_{\mathcal{I}_1^{\prime}})/(n_0 + n_1).
$$
Combining with the property of that $\bm\epsilon^{(0)}$ and $\check{\bm\epsilon}^{(1)}-\mathbb{E}(\check{\bm\epsilon}^{(1)})$ are sub-Gaussian with zero mean, similar to \eqref{prob-compute}, it follows that
$$
 {\mathbb{P}}\left(\|\Lambda_{11} \|_\infty  \leq t \right) > 1 -  p^{-c}\to 1
$$
for some constants $c>0$.

For $\Lambda_{12}$, since the weighted average distance $h_{\rm ave}$ is involved, we consider its general form which can be written as follows,
$$
\Lambda_{12} =  \frac{1}{n_0 + \Sigma_{k = 1}^K n_k} \sum\limits_{k = 1}^K
[\mathbb{E}(\check{\bm\epsilon}^{(k)})]^{\top}\bX^{(k)},
$$
where $\check {\bm \epsilon}^{(k)} = (\check \epsilon_i^{(k)}, i \in \mathcal{D}_k)^{\top}$ with $\check \epsilon_i^{(k)} = (\epsilon_i^{(k)} + \eta_i^{(k)}) \omega_{i,T}^{(k)} {\rm I}(i \in \mathcal{I}'_k)$.
Then we have
\begin{equation}\label{bias-2}
 \|\Lambda_{12}\|_{\infty} = \frac{1}{n_0 + \Sigma_{k = 1}^K n_k}
\left\|\sum\limits_{k = 1}^K \sum\limits_{i\in \mathcal{D}_k} \bm x_i^{(k)} \mathbb{E} (\check{\epsilon}_i^{(k)})\right\|_{\infty} \leq \frac{1}{\Sigma_{k = 1}^K n_k}
\left\|\sum\limits_{k = 1}^K \sum\limits_{i\in \mathcal{D}_k} \bm x_i^{(k)} \mathbb{E} (\check{\epsilon}_i^{(k)})\right\|_{\infty},
\end{equation}
where
\begin{align}\nonumber
 \mathbb{E} (\check{\epsilon}_i^{(k)})
 &= \int_{\epsilon_i^{(k)} \in \mathcal{I}'_k|_{\epsilon_i^{(k)}}}  (\epsilon_i^{(k)} + \eta_i^{(k)}) \omega_{i,T}^{(k)} f_{\epsilon^{(k)}}(\epsilon_i^{(k)}) d \epsilon_i^{(k)} \\\label{C3}
 &= \int_{A - \eta_i^{(k)}}^{A + \eta_i^{(k)}} \xi_i^{(k)} f_{\epsilon}(\xi_i^{(k)}) d \xi_i^{(k)}
 = \frac{1}{2T} \int_{A - \eta_i^{(k)}}^{A + \eta_i^{(k)}} \xi_i^{(k)} d \xi_i^{(k)}
 = \frac{A}{T} \eta_i^{(k)}
\end{align}
with $\xi_i^{(k)} = \epsilon_i^{(k)} + \eta_i^{(k)}$. By inserting the result of \eqref{C3} into \eqref{bias-2} and noticing the definition of $\eta_i^{(k)}$, we have
\begin{eqnarray}\nonumber
\|\Lambda_{12}\|_{\infty} &\leq&
\left\|\sum\limits_{k = 1}^K \sum\limits_{i\in \mathcal{D}_k} \bm x_i^{(k)} \mathbb{E} (\check{\epsilon}_i^{(k)})\right\|_{\infty} / \sum\limits_{k = 1}^K n_k\\\nonumber
& = & C_{A,T} \left\|\sum\limits_{k = 1}^K n_k \sum\limits_{i\in \mathcal{D}_k} \frac{1}{n_k} \bm x_i^{(k)} (\bm x_i^{(k)})^{\top} (\bbeta^{(k)} - \bbeta^{(0)})\right\|_{\infty} / \sum\limits_{k = 1}^K n_k \\\nonumber
& \leq & C_{A,T} \sum\limits_{k = 1}^K \pi_k \| \hat \bSig^{(k)} (\bbeta^{(k)} - \bbeta^{(0)})\|_{\infty} \\\nonumber
& \leq &  C_{A,T} \sum\limits_{k = 1}^K \pi_k \| (\hat \bSig^{(k)}  - \bSig^{(k)})(\bbeta^{(k)} - \bbeta^{(0)})\|_{\infty} +
C_{A,T} \sum\limits_{k = 1}^K \pi_k \|\bSig^{(k)} (\bbeta^{(k)} - \bbeta^{(0)})\|_{\infty}\\\nonumber
& \leq &   C_{A,T} \underset{1 \leq k \leq K}{\rm max}
\left(
\underset{1 \leq i,j \leq p}{\rm max}|\hat \sigma_{ij}^{(k)} - \sigma_{ij}^{(k)}| + \underset{1 \leq i,j \leq p}{\rm max} |\sigma_{ij}^{(k)}|
\right) \sum\limits_{k = 1}^K \pi_k \|\bbeta^{(k)} - \bbeta^{(0)}\|_1\\\label{E-i1-2}
&\le & C
 C_{A,T} h_{\rm ave}
\end{eqnarray}
for some constants $C > 0$, where $C_{A,T} = A/T$ and $h_{\rm ave} = \Sigma_{k = 1}^K \pi_k \|\bbeta^{(k)} - \bbeta^{(0)}\|_1$ with $\pi_k = n_k / \Sigma_{k = 1}^K n_k$. In the last step of \eqref{E-i1-2}, we use  the properties of sample covariance matrix  that $\underset{1 \leq i,j \leq p}{\rm max}|\hat \sigma_{ij}^{(k)} - \sigma_{ij}^{(k)}| = o_p(1)$ \citep{Bickel2008} and the fact that $\sigma^{(k)}_{\max}=\max_{ij}\sigma_{ij}^{(k)}$ is bounded in Condition 7.
Define the following two events:
 $$
    \mathcal{A}_1=\left\{\|\Lambda_{11}\|_\infty \leq
    \lambda_n^{(1)}\right\}, ~~~~  \mathcal{A}_2 =\left\{\|\Lambda_{12}\|_\infty\le  \lambda_n^{(2)}\right\},
$$
 where
 $$
 \lambda_n^{(1)} = C_1 \tilde \rho_{\mathcal{I}'_1} \sqrt{\log p / (n_0 + n_{\mathcal{I}'_1})},~~~~
 \lambda_{n}^{(2)}= C_2 C_{A,T} h_{\rm ave}
 $$
with $\tilde \rho_{\mathcal{I}'_1} = (n_0 + n_{\mathcal{I}'_1})/(n_0 + n_1)$ and $C_1, C_2$ are some constants. Hence, under the set $\mathcal{A}_1 \cap \mathcal{A}_2$, it holds that
$$
 \|\Lambda_1\|_{\infty} \leq  \|\Lambda_{11}\|_{\infty} +  \|\Lambda_{12}\|_{\infty} \leq
\lambda_n^{(1)} + \lambda_n^{(2)}.
$$

According to above discussion, both the probability of events $\mathcal{A}_i (i = 1,2)$ are tending to 1. Thereafter, similar to the proving procedures of step 2 in Theorem \ref{the2}, we derive that
\begin{equation}\label{oracle-rate-2}
\|\tilde \bbeta_{T,ora}^{(0)} - \bbeta^{(0)}\|^2 \leq
\frac{16 \lambda^2 s_0}{\tilde \rho_{\mathcal{I}'_1}^2 \phi^4}.
\end{equation}
Recall that $\lambda = 2(\lambda_n^{(1)} + \lambda_n^{(2)})$ and the event   $\{n_{\mathcal{I}'_1} \asymp \mathbb{E}(n_{\mathcal{I}'_1})\}$ holds in probability from Proposition \ref{prop4}. As $\underset{0 \leq k \leq K}{\rm min} n_k \rightarrow \infty$, then (\ref{oracle-rate-2}) can be rewritten as
$$
\|\tilde \bbeta_{T,ora}^{(0)} - \bbeta^{(0)}\|^2
 =  O_p \left\{
\frac{s_0 {\rm log}\,p}{n_0 + \mathbb{E}(n_{\mathcal{I}'_1})} +
 \rho_{\mathcal{I}'_1}^{-2} s_0  h_{\rm ave}^2
\right\}.
$$

\subsection{Proof of Lemma 2}
According to the definition of $\omega_{i,T}^{(k)}$, for any $i \in \hat {\mathcal{I}}'_{kj}$, it holds that
\begin{eqnarray}\nonumber
|\hat \omega_{i,T}^{(k)}/ \omega_{i,T}^{(k)} - 1| &=&
\left|
\frac{f_{\epsilon^{(k)}}(\epsilon_i^{(k)})}{\hat f_{\epsilon^{(k)}}(\epsilon_i^{(k)})} \cdot
\frac{\hat f_{\epsilon^{(k)}}(\epsilon_i^{(k)})}{\hat f_{\epsilon^{(k)}}(\hat \epsilon_i^{(k)})} - 1
\right|\\\nonumber
& = &
\left|
\frac{f_{\epsilon^{(k)}}(\epsilon_i^{(k)})}{\hat f_{\epsilon^{(k)}}(\epsilon_i^{(k)})} - 1
\right| \cdot
\left|
\frac{\hat f_{\epsilon^{(k)}}(\epsilon_i^{(k)})}{\hat f_{\epsilon^{(k)}}(\hat \epsilon_i^{(k)})} - 1
\right|\\\nonumber
& + &
\underbrace{
\left|
\frac{f_{\epsilon^{(k)}}(\epsilon_i^{(k)})}{\hat f_{\epsilon^{(k)}}(\epsilon_i^{(k)})} - 1
\right|}_{\Lambda_{i,1}} +
\underbrace{
\left|
\frac{\hat f_{\epsilon^{(k)}}(\epsilon_i^{(k)})}{\hat f_{\epsilon^{(k)}}(\hat \epsilon_i^{(k)})} - 1
\right|}_{\Lambda_{i,2}}\\\nonumber
& := & \Lambda_{i,1} \Lambda_{i,2} + \Lambda_{i,1} + \Lambda_{i,2}.
\end{eqnarray}
For $\Lambda_{i,1}$, by Lemma \ref{lemma0}, it holds that $\Lambda_{i,1} = O_p(q_n)$.
For $\Lambda_{i,2}$, similar to \eqref{B77}, we have
$$
\Lambda_{i,2} = \left|
\frac{\hat  f_{\epsilon^{(k)}}(\hat \epsilon_i^{(k)}) - \hat f_{\epsilon^{(k)}}(\epsilon_i^{(k)})}{\hat  f_{\epsilon^{(k)}}(\hat \epsilon_i^{(k)})}
\right| = O_p\left( \underset{i \in \mathcal{D}_{k3}}{\rm max}~|\hat \epsilon_i^{(k)} -  \epsilon_i^{(k)}|/b_k^2 \right) = O_p(u_n),
$$
where $u_n$ is defined in Lemma \ref{lemma1}. Due to the relationship $u_n \leq q_n$ in Lemma \ref{lemma1}, we derive that
$$
\max\limits_{1 \le k \le K}\max_{i \in \hat {\mathcal{I}}'_{kj}} |\hat \omega_{i,T}^{(k)}/ \omega_{i,T}^{(k)} - 1| =
O_p (q_n).
$$

\subsection{Proof of Theorem 4}
The proof of Theorem \ref{the4} is similar to Theorem \ref{the2} except that here we need to bound $\|\Lambda_1\|_{\infty}$ with
$$
\Lambda_1 = \frac{1}{\tilde n_0 + \tilde n_1}
\left[(\bm \epsilon^{(0)})^{\T} \bX^{(0)}   +  (\bm \epsilon^{(1)} + \bm \eta^{(1)})^{\T} \bI_{\hat {\mathcal{I}}_{13}^{\prime}} \hat\bW_T  \bX^{(1)}\right],
$$
where $\hat\bW_T = {\rm diag}\{\hat\omega_{1,T}^{(1)}, i \in \mathcal{D}_{13}\}$  denotes  the diagonal weights matrix, $\mathcal{I}'_{13} = \{i \in \mathcal{D}_{13}: |\hat\epsilon_i^{(1)}| \leq A,~|\hat\eta_i^{(1)}| \leq M\}$ and $\tilde n_k = n_k / 3$ for $k = 0,1$.

 The proof is similar to Step 2.1 of the  proof of Theorem \ref{the2}.
 Let $\check {\bm \epsilon}^{(1)} = (\check \epsilon_i^{(1)}, i \in \mathcal{D}_{13})^{\top}$ with $\check \epsilon_i^{(1)} = (\epsilon_i^{(1)} + \eta_i^{(1)}) \hat \omega_{i,T}^{(1)} {\rm I}(i \in \hat {\mathcal{I}}'_{13})$.
It follows that
$$
\Lambda_1 = \underbrace{
\frac{1}{\tilde n_0 + \tilde n_1}
\left\{
(\bm \epsilon^{(0)})^{\top} \bX^{(0)} + (\check{\bm\epsilon}^{(1)}-\mathbb{E}(\check{\bm\epsilon}^{(1)}))^{\top}\bX^{(1)}
\right\}
}_{\Lambda_{11}} +
\underbrace{
\frac{1}{\tilde n_0 + \tilde n_1}
[\mathbb{E}(\check{\bm\epsilon}^{(1)})]^{\top}\bX^{(1)}
}_{\Lambda_{12}},
$$
 where $\Lambda_{11}$ and $\Lambda_{12}$ are the counterpart in \eqref{the2-lambda12}.

For $\Lambda_{11}$, we had shown that $\check \epsilon_i^{(1)}$'s are independent and bounded variables in step 2.1 of Theorem \ref{the2}, and consequently are sub-Gaussian variables. Let
$$
t = C_1 \tilde \rho_{\hat{\mathcal{I}}_{13}^{\prime}} \sqrt{{\rm log}\,p /(\tilde n_0+ n_{\hat {\mathcal{I}}_{13}^{\prime}})}~~\text{with}~~\tilde \rho_{\hat{\mathcal{I}}_{13}^{\prime}} = (\tilde n_0 + n_{\hat{\mathcal{I}}_{13}^{\prime}})/(\tilde n_0 + \tilde n_1).
$$
Combining with the property of that $\bm\epsilon^{(0)}$ and $\check{\bm\epsilon}^{(1)}-\mathbb{E}(\check{\bm\epsilon}^{(1)})$ are sub-Gaussian with zero mean, similar to \eqref{Lambda_1}, it follows that
$$
 {\mathbb{P}}\left(\|\Lambda_{11} \|_\infty  \leq t \right) > 1 -  p^{-c}\to 1
$$
for some constant $c>0$.

For $\Lambda_{12}$, its general form for $K \geq 1$ is as follows,
$$
\Lambda_{12} =  \frac{3}{n_0 + \Sigma_{k = 1}^K n_k} \sum\limits_{k = 1}^K
[\mathbb{E}(\check{\bm\epsilon}^{(k)})]^{\top}\bX^{(k)} ,
$$
where $\check {\bm \epsilon}^{(k)} = (\check \epsilon_i^{(k)}, i \in \mathcal{D}_{k3})^{\top}$ with $\check \epsilon_i^{(k)} = (\epsilon_i^{(k)} + \eta_i^{(1)}) \hat \omega_{i,T}^{(k)} {\rm I}(i \in \hat{\mathcal{I}}'_k)$.
Since the quantity involves the average of distances of sources, we consider its general form for clarity. It holds that
\begin{equation}\label{bias-22}
 \|\Lambda_{12}\|_{\infty} = \frac{3}{n_0 + \Sigma_{k = 1}^{K} n_k}
\left\|\sum\limits_{k = 1}^K \sum\limits_{i\in \mathcal{D}_{k3}} \bm x_i^{(k)} \mathbb{E} (\check{\epsilon}_i^{(k)})\right\|_{\infty} \leq \frac{3}{\Sigma_{k = 1}^{K} n_k}
\left\|\sum\limits_{k = 1}^K \sum\limits_{i\in \mathcal{D}_{k3}} \bm x_i^{(k)} \mathbb{E} (\check{\epsilon}_i^{(k)})\right\|_{\infty},
\end{equation}
where
\begin{eqnarray}\nonumber
 \mathbb{E} (\check{\epsilon}_i^{(k)})&=& \int {\rm I}( i \in \hat{\mathcal{I}}'_{k3})  (\epsilon_i^{(k)} + \eta_i^{(k)}) \hat\omega_{i,T}^{(k)} f_{\epsilon^{(k)}}(\epsilon_i^{(k)}) d \epsilon_i^{(k)}\\\nonumber
& = & \underbrace{
\int_{\epsilon_i \in \hat{\mathcal{I}}'_{k3}|_{\epsilon_i^{(k)}}}  (\epsilon_i^{(k)} + \eta_i^{(k)}) \omega_{i,T}^{(k)} f_{\epsilon^{(k)}}(\epsilon_i^{(k)})d \epsilon_i^{(k)}
}_{E_{k,i1}} +
\underbrace{
\int_{\epsilon_i \in \hat{\mathcal{I}}'_{k3}|_{\epsilon_i^{(k)}}}  (\epsilon_i^{(k)} + \eta_i^{(k)}) (\hat\omega_{i,T}^{(k)} - \omega_{i,T}^{(k)}) f_{\epsilon^{(k)}}(\epsilon_i^{(k)})d \epsilon_i^{(k)}
}_{E_{k,i2}} \\\nonumber
&:=& E_{k,i1}+E_{k,i2}.
\end{eqnarray}

For $E_{k,i1}$, by the definition of $\hat {\mathcal{I}}'_{k3}$ and the symmetric of $f_{\epsilon}$ being $f_{\epsilon}(t) = (2T)^{-1} {\rm I}(|t| \leq T)$, it follows that
$$
E_{k,i1} = \int_{\epsilon_i^{(k)} \in \hat{\mathcal{I}}'_{k3}|_{\epsilon_i^{(k)}}}  (\epsilon_i^{(k)} + \eta_i^{(k)}) f_{\epsilon}(\epsilon_i^{(k)} + \eta_i^{(k)}) d \epsilon_i^{(k)}
 =  \int_{A - r_i^{(k)} - \eta_i^{(k)}}^{A + r_i^{(k)} + \eta_i^{(k)}} \xi_i^{(k)} f_{\epsilon}(\xi_i^{(k)}) d\xi_i^{(k)}
 =  C_{A,T} (r_i^{(k)} + \eta_i^{(k)}),
$$
where $C_{A,T} = A / T$, $r_i^{(k)} = (\bm x^{(k)})^{\top}(\tilde \bbeta^{(k)} - \bbeta^{(k)})$ and $\xi_i^{(k)} = \epsilon_i^{(k)} + \eta_i^{(k)}$. Then,
\begin{eqnarray}\nonumber
&&
3 \left\|\sum\limits_{k = 1}^K \sum\limits_{i\in \mathcal{D}_{k3}} \bm x_i^{(k)} E_{k,i1}\right\|_{\infty} / \sum\limits_{k = 1}^K n_k\\\nonumber
& = &
C_{A,T}
\left\|\sum\limits_{k = 1}^K n_k \sum\limits_{i\in \mathcal{D}_{k3}} \frac{3}{n_k} \bm x_i^{(k)} (\bm x_i^{(k)})^{\top} (\tilde \bbeta^{(k)} - \bbeta^{(0)}) \right\|_{\infty} / \sum\limits_{k = 1}^K n_k \\\nonumber
& = &
C_{A,T}
\left\|\sum\limits_{k = 1}^K \pi_k \hat \bSig^{(k)} (\tilde \bbeta^{(k)} - \bbeta^{(0)}) \right\|_{\infty}\\\nonumber
& \leq &
C_{A,T} \sum\limits_{k = 1}^K \pi_k \|(\hat \bSig^{(k)}  - \bSig^{(k)})(\tilde \bbeta^{(k)} - \bbeta^{(k)})\|_{\infty} +
C_{A,T} \sum\limits_{k = 1}^K \pi_k\|\bSig^{(k)} (\tilde \bbeta^{(k)} - \bbeta^{(k)})\|_{\infty}\\\nonumber
& + &
C_{A,T} \sum\limits_{k = 1}^K \pi_k \|(\hat \bSig^{(k)}  - \bSig^{(k)})( \bbeta^{(k)} - \bbeta^{(0)})\|_{\infty} +
C_{A,T} \sum\limits_{k = 1}^K \pi_k\|\bSig^{(k)} (\bbeta^{(k)} - \bbeta^{(0)})\|_{\infty}.
\end{eqnarray}
Continuing the last inequality can be found
\begin{eqnarray}\nonumber
&&
3 \left\|\sum\limits_{k = 1}^K \sum\limits_{i\in \mathcal{D}_{k3}} \bm x_i^{(k)} E_{k,i1}\right\|_{\infty} / \sum\limits_{k = 1}^K n_k\\\nonumber
& \leq &
C_{A,T} \max\limits_{1 \leq k \leq K} \left(
\underset{1 \leq i,j \leq p}{\rm max}|\hat \sigma_{ij}^{(k)} - \sigma_{ij}^{(k)}| + \underset{1 \leq i,j \leq p}{\rm max} |\sigma_{ij}^{(k)}|
\right) \max\limits_{1 \leq k \leq K} \|\tilde \bbeta^{(k)} - \bbeta^{(k)}\|_1\\\nonumber
& + &
C_{A,T} \max\limits_{1 \leq k \leq K} \left(
\underset{1 \leq i,j \leq p}{\rm max}|\hat \sigma_{ij}^{(k)} - \sigma_{ij}^{(k)}| + \underset{1 \leq i,j \leq p}{\rm max} |\sigma_{ij}^{(k)}|
\right) \sum\limits_{k = 1}^K \pi_k \|\bbeta^{(k)} - \bbeta^{(0)}\|_1\\\label{E-i1-22}
&\le & C
C_{A,T} (\gamma_{\rm max} + h_{\rm ave})
\end{eqnarray}
for some constant $C>0$, where the last step is similar to \eqref{E-i1-2} and $\gamma_{\rm max} = \underset{1 \le k \le K}{\rm max} \gamma_k $ with $\gamma_k$ being the convergence rate of $\tilde \bbeta^{(k)}$.

For $E_{k,i2}$, it follows that
\begin{align}
E_{k,i2} &= \int_{\epsilon_i^{(k)} \in \hat{\mathcal{I}}'_{k3}|_{\epsilon_i^{(k)}}}  (\epsilon_i^{(k)} + \eta_i^{(k)})  \omega_i^{(k)}f_{\epsilon^{(k)}}(\epsilon_i^{(k)}) (\hat\omega_{i,T}^{(k)} / \omega_{i,T}^{(k)}-1) d \epsilon_i^{(k)}\notag\\
&\le
\int_{\epsilon_i^{(k)} \in \hat{\mathcal{I}}'_{k3}|_{\epsilon_i^{(k)}}} | (\epsilon_i^{(k)} + \eta_i^{(k)}) \omega_{i,T}^{(k)} f_{\epsilon^{(k)}}(\epsilon_i^{(k)}) | d \epsilon_i^{(k)}  \max_{i\in \hat{\mathcal{I}}'_{k3}}|\hat\omega_{i,T}^{(k)}/\omega_{i,T}^{(k)}-1|\notag\notag\\\nonumber
&=
\int_{A - r_i^{(k)0}}^{A + r_i^{(k)0}} |\xi_i^{(k)}| f_{\epsilon}(\xi_i^{(k)}) d\xi_i^{(k)} \cdot\max_{i\in \hat{\mathcal{I}}'_{k3}}|\hat\omega_{i,T}^{(k)} / \omega_{i,T}^{(k)}-1| \le \sigma_\epsilon^2  \max_{i\in \hat{\mathcal{I}}'_{k3}}|\hat\omega_{i,T}^{(k)} / \omega_{i,T}^{(k)}-1|,
\end{align}
where $\sigma_\epsilon^2 = \mE_{f_{\epsilon}} (X^2)$ is bounded. Then it follows that
\begin{eqnarray}\nonumber
&&3 \left\|\sum\limits_{k = 1}^K \sum\limits_{i\in \mathcal{D}_{k3}} \bm x_i^{(k)} E_{k,i2}\right\|_{\infty} / \sum\limits_{k = 1}^K n_k\\\nonumber
&\le& \sigma_\epsilon^2  \left( 3 \sum\limits_{k = 1}^K \sum\limits_{i\in \mathcal{D}_{k3} }\|\bm x_i^{(k)}\|_{\infty} / \sum\limits_{k = 1}^K n_k \right) \cdot \max_{1 \le k \le K}\max_{i\in \hat{\mathcal{I}}'_{k3}}|\hat\omega_{i,T}^{(k)}/\omega_{i,T}^{(k)}-1| \\\label{E-i2-22}
& =&  \sigma_{\epsilon}^2 \mu_{\max} \cdot \max_{1 \le k \le K}\max_{i\in \hat{\mathcal{I}}'_{k3}}|\hat\omega_{i,T}^{(k)} / \omega_{i,T}^{(k)}-1| \leq c \sigma_\epsilon^2 \mu_{\rm max} q_n,
\end{eqnarray}
 where the last inequality is due to Lemma \ref{lemma2}, where $\max_{1 \le k \le K} \max_{i\in \hat{\mathcal{I}}'_{k3}}|\hat\omega_{i,T}^{(k)}/\omega_{i,T}^{(k)}-1|\le c q_n$ with overwhelming probability for some constant $c>0$. Inserting (\ref{E-i1-22}) and (\ref{E-i2-22}) into (\ref{bias-22}), we have
 \begin{eqnarray}\nonumber
\|\Lambda_{12}\|_{\infty}
 &\leq&
\frac{3}{\Sigma_{k = 1}^K n_k}
\left\|\sum\limits_{k = 1}^K \sum\limits_{i\in \mathcal{D}_{k3}} \bm x_i^{(k)} \mathbb{E} (\check{\epsilon}_i^{(k)})\right\|_{\infty}\\\nonumber
 &\leq&
\frac{3}{\Sigma_{k = 1}^K n_k}
\left\|\sum\limits_{k = 1}^K \sum\limits_{i\in \mathcal{D}_{k3}} \bm x_i^{(k)} E_{k,i1} \right\|_{\infty} +
\frac{3}{\Sigma_{k = 1}^K n_k}
\left\|\sum\limits_{k = 1}^K \sum\limits_{i\in \mathcal{D}_{k3}} \bm x_i^{(k)} E_{k,i2}\right\|_{\infty}\\\nonumber
&\leq &
C C_{A,T} (\gamma_{\rm max} + h_{\rm ave}) + c \sigma_\epsilon^2 \mu_{\rm max} q_n.
 \end{eqnarray}

Define the two events are
$$
\mathcal{A}_1=\left\{\|\Lambda_{11}\|_\infty \leq  \lambda_n^{(1)}\right\}, ~~~~  \mathcal{A}_2=\left\{\|\Lambda_{12}\|_\infty\le  \lambda_n^{(2)}\right\},
$$
 where
 $$
 \lambda_n^{(1)}= C_1 \tilde \rho_{\hat{\mathcal{I}}'_{13}} \sqrt{\log p / (\tilde n_0 + n_{\hat {\mathcal{I}}'_{13}})},~~~~
 \lambda_{n}^{(2)}= C_2
 [\mu_{\max} q_n + C_{A,T} (\gamma_{\rm max} + h_{\rm ave})]
 $$
with $\tilde \rho_{\hat{\mathcal{I}}'_{13}} = (\tilde n_0 + n_{\hat {\mathcal{I}}'_{13}})/(\tilde n_0 + \tilde n_1)$. Hence, under the set $\mathcal{A}_1 \cap \mathcal{A}_2$, it holds that
$$
 \|\Lambda_1\|_{\infty} \leq  \|\Lambda_{11}\|_{\infty} +  \|\Lambda_{12}\|_{\infty} \leq
\lambda_n^{(1)} + \lambda_n^{(2)}.
$$

According to above discussion, both the probability of events $\mathcal{A}_i (i = 1,2)$ are tending to 1. Thereafter, similar to the proving procedures of step 2 in Theorem \ref{the2}, it holds in probability that
\begin{equation}\label{rate22-2}
\|\hat \bbeta_{T,1}^{(0)} - \bbeta^{(0)}\|^2 \leq
\left(\frac{\tilde n_0 + \tilde n_1}{\tilde n_0 + n_{\mathcal{I}_{13}'^{-}}} \right)^2
\frac{16 \lambda^2 s_0}{\phi^4}.
\end{equation}
Let $\lambda = 2(\lambda_n^{(1)} + \lambda_n^{(2)})$. As $\underset{0 \leq k \leq K}{\rm min} n_k \rightarrow \infty$, by Proposition \ref{prop4}, \eqref{rate22-2} can be rewritten as
$$
\|\hat \bbeta_{T,1}^{(0)} - \bbeta^{(0)}\|_2^2
 = O_p \left\{
 \frac{s_0 {\rm log}\,p}{n_0 + \mathbb{E}(n_{\mathcal{I}'_1})} +
\rho_{\mathcal{I}'_1}^{-2}
s_0 (\mu_{\rm max} q_n + \gamma_{\rm max} + h_{\rm ave})^2
\right\}.
$$
Similarly, we obtain the same convergence rate for $\hat \bbeta_{T,2}^{(0)}$ and $\hat \bbeta_{T,3}^{(0)}$, which leads to the conclusion desired.

\renewcommand\thesection{S.4}
\section{Additional numerical results}\label{S.4}
We present additional simulation in a setting similar to that in the main paper but with the magnitude of the difference between $\bbeta^{(0)}$ and $\bbeta^{(k)}$ being random and the case where $\epsilon_i^{(0)}$ and $\epsilon_i^{(k)} (k \geq 1)$ follow different distributions. In particular, for $1 \leq k \leq K$, we specify $\bbeta^{(k)}$ as follows.
\begin{itemize}
  \item [$(i)$] For a given $\mathcal{B}$, if $k \in \mathcal{B}$, let
  $$
 \beta^{(k)}_j =
  \beta_j^{(0)} - \xi_j {\rm I}(j \in T_k),~~ \text{where}~ \xi_j \sim_{i.i.d} U(0,1),
 $$
 where $T_k$ is a random subset of $\{s_0 + 1,\cdots,p\}$ with $|T_k|= d$. The value of $d$ will change for different simulations.

  \item [$(ii)$]

  For a given $\mathcal{B}$, if $k \notin \mathcal{B}$, let
  $$
 \beta^{(k)}_j = \left\{
  \begin{array}{ll}
  \beta_j^{(0)} -  1,    &  j \in [s_0];\\
  \beta_j^{(0)} - \xi_j, &  j \in U_k;\\
   \beta_j^{(0)}, & \mbox{otherwise},
        \end{array}
        \right. ~~ \text{where}~ \xi_j \sim_{i.i.d} U(0,1),
 $$
 where $U_k$ is a random subset of $\{s_0 + 1,\cdots,p\}$ with $|U_k|= 2 s_0.$
\end{itemize}
Specifically, for $1 \leq k \leq K$, we focus on the following two scenarios under full distribution shift:

\begin{itemize}
  \item [$(S1)$]
   Both $\epsilon_i^{(0)}$ and $\epsilon_i^{(k)}$ are independently distributed as $N(0,1)$;

  \item [$(S2)$]
  The errors satisfy $\epsilon_i^{(0)} \sim N(0,1)$ and $\epsilon_i^{(k)} \sim t(5)$.
 \end{itemize}

 The results for ($S1$) and ($S2$) are reported in Figures \ref{fig1} and \ref{fig2} respectively. We can see that qualitatively similar conclusions to the main paper can be reached with regard to the performance of various methods under investigation. We also remark that RIW-TL-P assuming Gaussianity of the errors tend to underperform RIW-TL where the error distribution is nonparametrically estimated and RIW-TL-U where the errors are assumed following a uniform distribution.
\begin{figure}[!htbp]
\centering
  \includegraphics[scale=0.7]{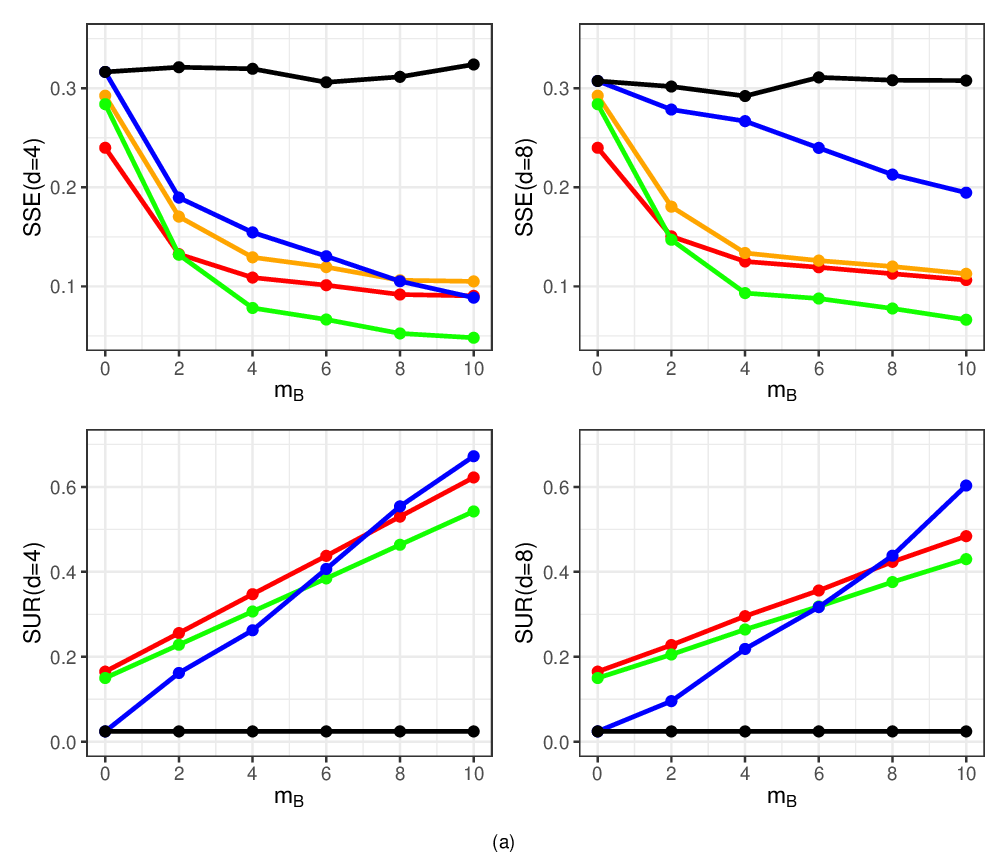}
  \includegraphics[scale=0.7]{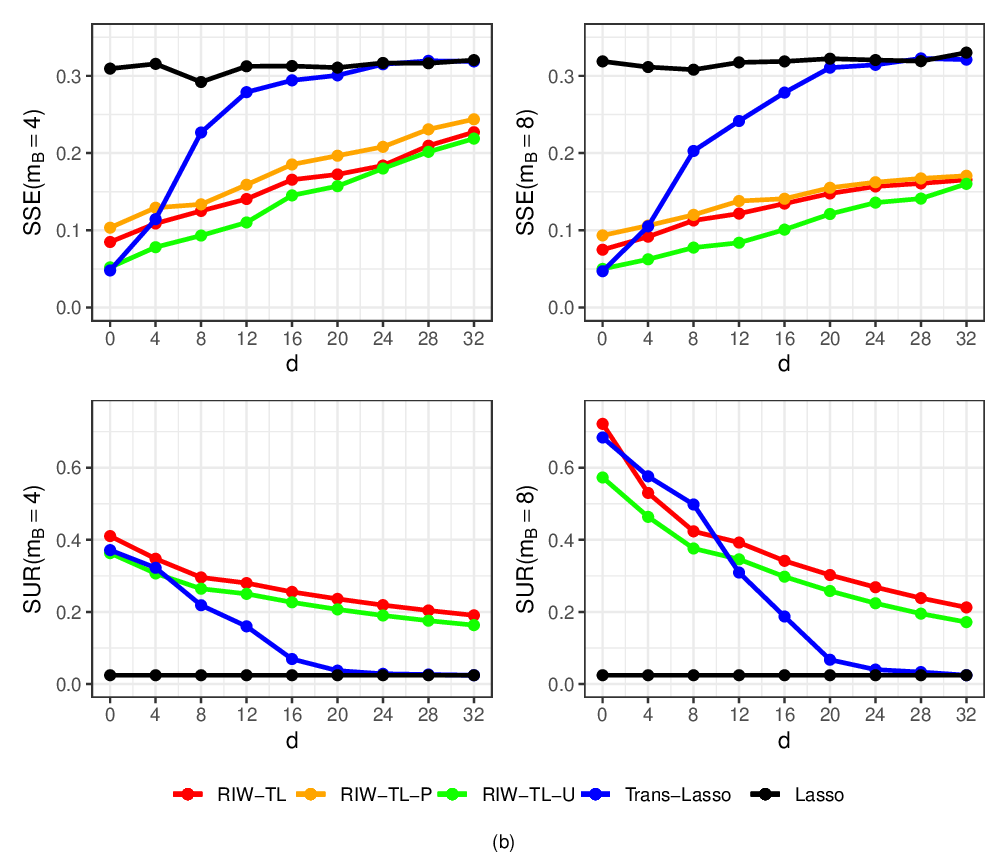}
  \caption{\small (a) The estimation errors (the first row) and sample usage rates (the second row) versus $d$ for different $m_{\mathcal{B}}$ in case $(S1)$. (b) The estimation errors (the third row) and sample usage rates (the fourth row) versus $m_{\mathcal{B}}$ for different $d$ in case $(S1)$. Note in the SUR plots, RIW-TL-P lines are invisible because they overlap with the corresponding RIW-TL lines.}\label{fig1}
\end{figure}

\begin{figure}[!htbp]
\centering
  \includegraphics[scale=0.7]{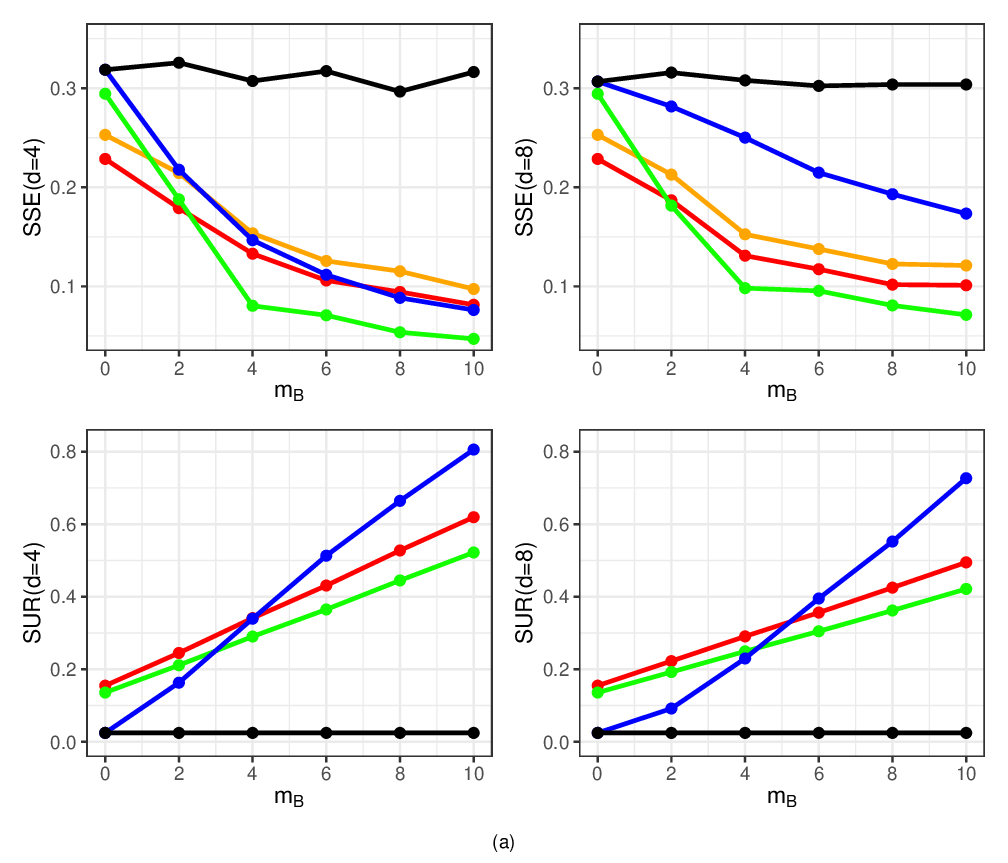}
  \includegraphics[scale=0.7]{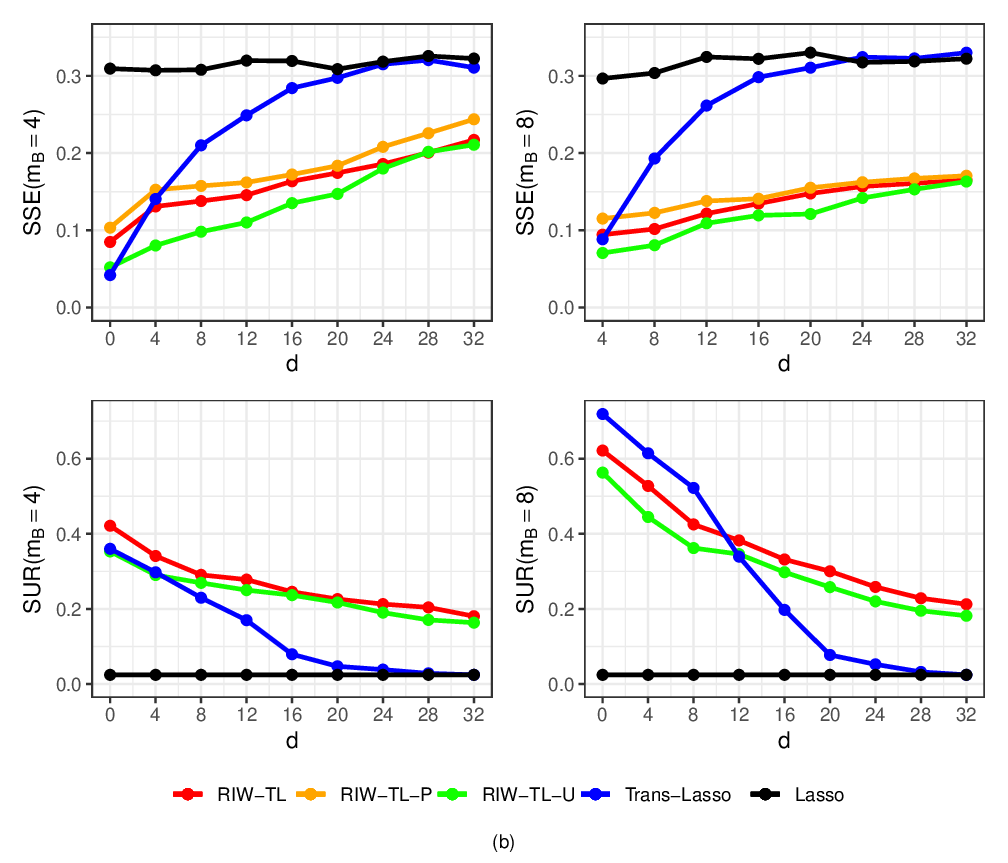}
  \caption{\small (a) The estimation errors (the first row) and sample usage rates (the second row) versus $d$ for different $m_{\mathcal{B}}$ in case $(S2)$. (b) The estimation errors (the third row) and sample usage rates (the fourth row) versus $m_{\mathcal{B}}$ for different $d$ in case $(S2)$. Note in the SUR plots, RIW-TL-P lines are invisible because they overlap with the corresponding RIW-TL lines.}\label{fig2}
\end{figure}

For the real data analysis, we conduct tests of normality of the residuals for the 31 sources once linear models are fitted.
The resulting $p$-values are plotted in Figures \ref{fig3}.  We can see that more than half of the tests reject the notion of normality for the residuals, suggesting that modelling the residuals as normal distributions may not be appropriate.
\clearpage
\begin{figure}[!htbp]
\centering
  \includegraphics[scale=0.55]{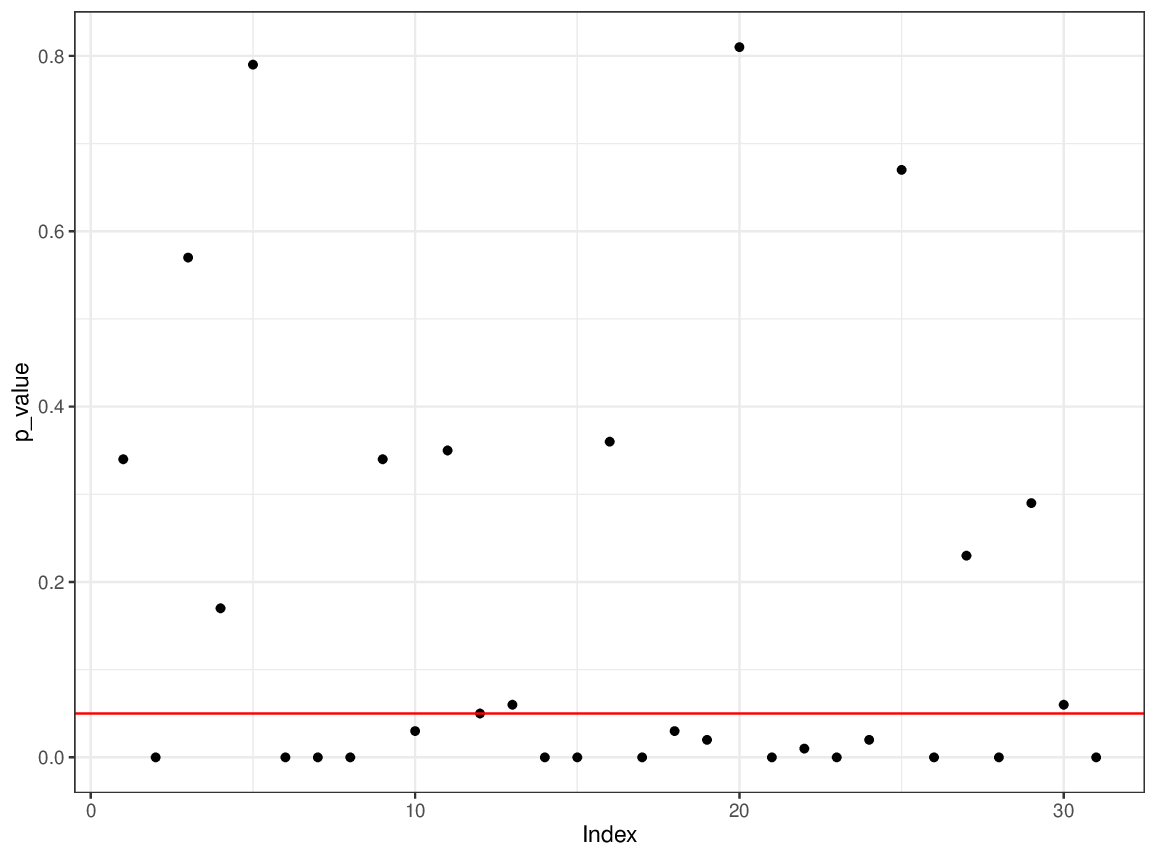}
  \caption{Plot for the $p$-values for testing the normality of the residuals in 31 sources. The red line is at the significance level $\alpha = 0.05$.}\label{fig3}
\end{figure}

\end{document}